\pdfoutput=1

\ifx\mycompatibility\undefined
\providecommand{\mydocumentoptionsdefault}{bibliography=totocnumbered,twocolumn=false,twoside=false,titlepage=false,parskip=half-,abstract=false,fleqn,draft=false,paper=a4,fontsize=11pt}
\else
\newcommand\mynoinkscapesvgtopdf{}
\newcommand{\mybiblatexgiveninits}{firstinits}
\newcommand{\mydocumentoptionsdefault}{enabledeprecatedfontcommands,bibtotocnumbered,onecolumn,oneside,notitlepage,halfparskip-,fleqn,final,a4paper,11pt}
\fi

\newcommand{\mypaper}{A4}\newcommand{\mydiv}{12}

\newcommand{\MyBibLaTeXStyle}{numeric}
\newcommand{\mydisablewikimode}{}
\newcommand{\mycustommaketitle}{1}

\providecommand{\mydocumentclass}{scrartcl}

\providecommand{\mydocumentoptions}{}
\providecommand{\mydocumentoptionsdefault}{twocolumn=false,twoside=false,titlepage=false,parskip=half-,abstract=false,fleqn,draft=false,paper=a4,fontsize=10pt}
\providecommand{\mypaper}{414.936062pt:586.838153pt}
\providecommand{\mydiv}{30}
\documentclass[\mydocumentoptionsdefault,\mydocumentoptions,paper=\mypaper,DIV=\mydiv]{\mydocumentclass}

\usepackage{amsmath}
\usepackage{amssymb}

\usepackage[english]{babel}
\usepackage[T1]{fontenc}
\usepackage{newpxtext}
\usepackage[amssymbols,vvarbb]{newpxmath}
\KOMAoptions{paper=\mypaper,DIV=\mydiv}

\setkomafont{disposition}{\rmfamily\mdseries\color{Mahogany}}
\setkomafont{title}{\rmfamily\mdseries}
\setkomafont{paragraph}{\rmfamily\bfseries\normalcolor}

\usepackage[standard,thmmarks,hyperref]{ntheorem}

\usepackage[amssymb]{SIunits}

\usepackage[format=plain,indention=0.5cm,font=small,margin={0cm,0cm,oneside}]{caption}

\usepackage[utf8]{inputenc}
\usepackage{textcomp}

\DeclareUnicodeCharacter{00A0}{~}

\usepackage{graphicx}
\usepackage{tabularx}
\usepackage{dcolumn}
\newcolumntype{.}{D{.}{.}{7}}
\newcolumntype{h}{D{.}{.}{3}}
\newcolumntype{d}[1]{D{.}{.}{#1}}
\graphicspath{{graphics/}}

\usepackage[dvipsnames,table]{xcolor}

\usepackage{chngcntr}
\counterwithin{equation}{section}
\counterwithin{theorem}{section}

\usepackage{scrpage2}
\usepackage{multirow}

\usepackage{wrapfig}

\usepackage{pdflscape}

\usepackage{csquotes}

\usepackage{afterpage}

\usepackage{soul}

\usepackage{varwidth}

\providecommand{\mycustommaketitle}{0}

\if \mycustommaketitle 1

\makeatletter
\renewcommand\maketitle{\par
  \begingroup%
  \newpage%
  \global\@topnum\z@   
  \null%
  \renewcommand{\thefootnote}{\fnsymbol{footnote}}%
  \begin{center}%
    {\LARGE\color{Mahogany} \@title \par}%
    \vskip 1.5em%
    {\large%
      \lineskip .5em%
      \begin{tabular}[t]{c}%
        \@author%
      \end{tabular}\par}%
    \vskip 1em%
    {\large \@date}%
  \end{center}%
  \par%
  \vskip 1.5em%
  \thispagestyle{plain}\@thanks%
  \setcounter{footnote}{0}
  \endgroup%
}
\makeatother

\fi

\usepackage[
    breaklinks,pdftex,
    colorlinks,linkcolor=Mahogany,citecolor=Mahogany,urlcolor=MidnightBlue,filecolor=MidnightBlue,
    hyperindex
]{hyperref}

\usepackage{nameref}
\usepackage{zref-xr}
\zxrsetup{tozreflabel=false,toltxlabel=true,verbose=true}
\usepackage[capitalise]{cleveref}

\providecommand{\MyBibLaTeXStyle}{authoryear-comp}
\usepackage[
  backend=biber,
style=\MyBibLaTeXStyle,
  natbib=false,
  sortcites=false,
]{biblatex}

\providecommand{\mybiblatexgiveninits}{giveninits}

\DeclareFieldFormat
  [article,inbook,incollection,inproceedings,patent,thesis,unpublished]
  {title}{`#1'}

\renewbibmacro{in:}{%
  \ifentrytype{article}{%
  }{%
    \printtext{\bibstring{in}\intitlepunct}%
  }%
}

\DeclareFieldFormat[article]{volume}{\mkbibbold{#1}}
\DeclareFieldFormat[article]{number}{\mkbibparens{#1}}

\renewbibmacro*{volume+number+eid}{%
  \printfield{volume}%
  \printfield{number}%
  \setunit{\addcomma\space}%
  \printfield{eid}}

\DeclareNameAlias{last-first/first-last}{last-first}

\DeclareFieldFormat{linked}{%
  \ifboolexpr{ test {\ifhyperref} and not test {\ifentrytype{online}} }
    {\iffieldundef{doi}
       {\iffieldundef{url}
          {\iffieldundef{isbn}
             {\iffieldundef{issn}
                {#1}
                {\href{\worldcatsearch\thefield{issn}}{#1}}}
             {\href{\worldcatsearch\thefield{isbn}}{#1}}}
          {\href{\thefieldfirstword{url}}{#1}}}
       {\href{http://dx.doi.org/\thefield{doi}}{#1}}}
    {#1}}

\def\worldcatsearch{http://www.worldcat.org/search?qt=worldcat_org_all&q=}

\makeatletter
\def\thefieldfirstword#1{%
  \expandafter\expandafter
  \expandafter\firstword
  \expandafter\expandafter
  \expandafter{\csname abx@field@#1\endcsname}}
\def\firstword#1{\firstword@i#1 \@nil}
\def\firstword@i#1 #2\@nil{#1}
\makeatother

\DeclareFieldFormat{url}{\url{\firstword{#1}}}

\renewbibmacro*{title}{
  \ifboolexpr{ test {\iffieldundef{title}} and test {\iffieldundef{subtitle}} }
    {}
    {\printtext[title]{\printtext[linked]{%
       \printfield[titlecase]{title}%
       \setunit{\subtitlepunct}%
       \printfield[titlecase]{subtitle}}}%
     \newunit}%
  \printfield{titleaddon}}

\renewbibmacro*{periodical}{
  \iffieldundef{title}
    {}
    {\printtext[title]{\printtext[linked]{%
       \printfield[titlecase]{title}%
       \setunit{\subtitlepunct}%
       \printfield[titlecase]{subtitle}}}}}

\DeclareFieldFormat{title}{\printtext[linked]{\mkbibemph{#1}}}
\DeclareFieldFormat
  [article,inbook,incollection,inproceedings,patent,thesis,unpublished]
  {title}{\mkbibquote{\printtext[linked]{#1}\isdot}}
\DeclareFieldFormat[suppbook,suppcollection,suppperiodical]{title}
  {\printtext[linked]{#1}}

\newcommand{\parena}[1]{{\left(#1\right)}}

\newcommand{\parend}[1]{{\Bigl(#1\Bigr)}}

\newcommand{\sqbc}[1]{{\bigl[#1\bigr]}}
\newcommand{\sqbd}[1]{{\Bigl[#1\Bigr]}}

\newcommand{\cura}[1]{{\left\{#1\right\}}}

\newcommand{\curc}[1]{{\bigl\{#1\bigr\}}}

\newcommand{\bra}[1]{{\langle#1\rvert}}

\newcommand{\brac}[1]{{\bigl\langle#1\bigr\rvert}}

\newcommand{\ket}[1]{{\lvert#1\rangle}}

\newcommand{\ketc}[1]{{\bigl\lvert#1\bigr\rangle}}

\newcommand{\braket}[2]{{\langle#1\vert#2\rangle}}

\newcommand{\braketc}[2]{{\bigl\langle#1\big\vert#2\bigr\rangle}}

\newcommand{\skp}[2]{{\langle#1,#2\rangle}}

\newcommand{\skpc}[2]{{\bigl\langle#1,#2\bigr\rangle}}

\newcommand{\ketbra}[2]{{\vert#1\rangle\!\langle#2\vert}}

\newcommand{\abs}[1]{{\lvert#1\rvert}}

\newcommand{\norm}[1]{{\lVert#1\rVert}}
\newcommand{\norma}[1]{{\left\lVert#1\right\rVert}}

\newcommand{\ceil}[1]{{\lceil#1\rceil}}

\newcommand{\mymathbb}[1]{\mathbb{#1}}

\newcommand{\ee}{{\mathrm e}}
\newcommand{\ii}{{\mathrm i}}

\renewcommand\Re{\operatorname{Re}}
\renewcommand\Im{\operatorname{Im}}

\newcommand{\id}{\operatorname{id}}
\newcommand{\idm}{\mathbb{1}}
\newcommand{\tr}{\operatorname{Tr}}
\newcommand{\rk}{\operatorname{rk}}
\DeclareMathOperator{\osr}{OSR}

\newcommand{\rg}{\operatorname{im}}

\newcommand{\R}{{\mymathbb R}}
\newcommand{\C}{{\mymathbb C}}

\newcommand{\tra}{\intercal}

\newenvironment{spmatrix}{%
  \left(%
  \begin{smallmatrix}%
}{%
  \end{smallmatrix}%
  \right)%
}

\newcommand{\executeiffilenewer}[3]{%
\ifnum\pdfstrcmp{\pdffilemoddate{#1}}%
{\pdffilemoddate{#2}}>0%
{\immediate\write18{#3}}\fi%
}

\ifx\mynoinkscapesvgtopdf\undefined

\newcommand{\isvgc}[2][]{%
\executeiffilenewer{#2.svg}{#2-isvgt_full.pdf}%
{/home/user/.config/inkscape/extensions/svgtextext_to_pdflatex.py -full-pdf #2.svg #2-isvgt.pdf%
}%
\includegraphics[#1]{#2-isvgt_full.pdf}%
}
\else

\newcommand{\isvgc}[2][]{%
\includegraphics[#1]{#2-isvgt_full.pdf}%
}

\fi

\providecommand{\theoremcounter}{theorem}

\theoremstyle{plain} 
\theoremheaderfont{\normalfont\bfseries}
\newtheorem{thm}[\theoremcounter]{Theorem}
\newtheorem{cor}[\theoremcounter]{Corollary}

\newtheorem{lem}[\theoremcounter]{Lemma}
\newtheorem{prop}[\theoremcounter]{Proposition}

\theorembodyfont{\upshape}

\theoremsymbol{\ensuremath{_{\Box}}}
\renewtheorem{rem}[\theoremcounter]{Remark}

\theoremstyle{nonumberplain}
\theoremheaderfont{\bfseries}
\theorembodyfont{\normalfont}
\theoremsymbol{\ensuremath{_\blacksquare}}
\renewtheorem{proof}{Proof}
\renewtheorem{beweis}{Beweis}

\AtBeginDocument{

  \newif\ifPreviewTest
  \newif\ifPreviewX

  \ifx\ifPreview\undefined
    \PreviewXfalse
  \fi

  \PreviewTestfalse
  \ifx\ifPreview\ifPreviewTest
    \PreviewXfalse
  \fi

  \PreviewTesttrue
  \ifx\ifPreview\ifPreviewTest
    \PreviewXtrue
  \fi

}

\newcommand{\tocIfNoPreview}{
  \ifPreviewX
  \section*{Contents not available (preview active)}
  \else
  \tableofcontents
  \fi
}

\ifx\mydisablewikimode\undefined

\input{aux/header-wiki}

abc bla

\fi

\Crefname{prop}{Proposition}{Propositions}
\Crefname{lem}{Lemma}{Lemmata}
\Crefname{thm}{Theorem}{Theorems}
\Crefname{cor}{Corollary}{Corollaries}
\Crefname{rem}{Remark}{Remarks}

\usepackage{authblk}
\renewcommand\Affilfont{\itshape\footnotesize}

\renewenvironment{abstract}{}{}

\author[1]{Milan Holz\"apfel\thanks{\texttt{mail@mholzaepfel.de}}$^,$}
\author[1,2]{Marcus Cramer}
\author[3]{Nilanjana Datta\thanks{\texttt{n.datta@damtp.cam.ac.uk}}$^,$}
\author[1]{Martin B. Plenio}
\affil[1]{Institut f\"ur Theoretische Physik and IQST, Albert-Einstein-Allee 11, Universit\"at Ulm, 89069 Ulm, Germany}
\affil[2]{Institut f\"ur Theoretische Physik, Leibniz Universit\"at Hannover, 30167 Hannover, Germany}
\affil[3]{Department of Applied Maths and Theoretical Physics, Centre for Mathematical Sciences, University of Cambridge, \authorcr{}\Affilfont{}Wilberforce Road, Cambridge CB3 0WA, United Kingdom}

\title{Petz recovery versus matrix reconstruction}
\date{September 14, 2017}

\newcommand{\lunf}{\mathrm L}
\newcommand{\runf}{\mathrm R}

\newcommand{\linmap}{\mathcal B}
\newcommand{\linop}{\mathcal B}
\newcommand{\sop}[2]{\linmap(\linop(#1); \linop(#2))}
\newcommand{\issop}[3]{#1\colon \linop(#2) \to \linop(#3)}
\newcommand{\qs}{\mathcal D}
\newcommand{\hilb}{\mathcal H}

\newcommand{\sysand}{,}
\newcommand{\hilbsysand}[1]{#1\sysand}
\newcommand{\andhilbsysand}[1]{\sysand{}#1\sysand}

\newcommand{\relent}{D}
\newcommand{\ldim}{d}
\newcommand{\ns}{n}
\newcommand{\supp}{\operatorname{supp}}
\newcommand{\poly}{\operatorname{poly}}

\newcommand{\mck}{\mathcal M}
\newcommand{\mcm}{\mathcal M}
\newcommand{\mcn}{\mathcal N}
\newcommand{\mct}{\mathcal T}
\newcommand{\mcu}{\mathcal U}
\newcommand{\petzrec}{\mathcal R^{\mathrm P}}
\newcommand{\cptprec}{\mathcal R}
\newcommand{\matrec}{\mathcal R^{\mathrm M}}
\newcommand{\linrec}{\mathcal R}
\newcommand{\recmap}{\mathcal R}
\newcommand{\cR}{\mathcal R}
\newcommand{\cN}{\mathcal N}
\newcommand{\cH}{\mathcal H}
\newcommand{\cD}{\mathcal D}
\newcommand{\cB}{\mathcal B}

\newcommand{\bigo}[1]{= \mathcal O(#1)}

\newcommand{\leftsetp}[1]{X'_{#1}}
\newcommand{\leftsetvalnb}[1]{1 \dots #1} 
\newcommand{\leftsetval}[1]{\{\leftsetvalnb{#1}\}}

\newcommand{\rightsetp}[1]{Y'_{#1}}
\newcommand{\rightsetvalnb}[1]{#1+1 \dots \ns} 
\newcommand{\rightsetval}[1]{\{\rightsetvalnb{#1}\}}

\newcommand{\leftsetname}[1]{X_{#1}}
\newcommand{\leftsetdef}[1]{\leftsetname{#1} = \leftsetval{#1}}
\newcommand{\leftset}[1]{\leftsetname{#1}}
\newcommand{\leftsetnb}[1]{\leftsetname{#1}}
\newcommand{\rightsetdef}[1]{\rightsetval{#1}}

\newcommand{\rightsetnb}[1]{\rightsetvalnb{#1}}

\newcommand{\pinv}{+}

\newcommand{\adjm}{*}
\newcommand{\adjs}{*}

\begin{document}
\setcounter{section}{0}
\maketitle

\begin{abstract}
  The reconstruction of the state of a multipartite quantum mechanical system represents a fundamental task in quantum information science.
  At its most basic, it concerns a state of a bipartite quantum system whose subsystems are subjected to local operations. 
  We compare two different methods for obtaining the original state from the state resulting from the action of these operations.
  The first method involves quantum operations called Petz recovery maps, acting locally on the two subsystems.
  The second method is called matrix (or state) reconstruction and involves local, linear maps which are not necessarily completely positive.
  Moreover, we compare the quantities on which the maps employed in the two methods depend.
  We show that any state which admits Petz recovery also admits state reconstruction.
  However, the latter is successful for a strictly larger set of states.
  We also compare these methods in the context of a finite spin chain.
  Here, the state of a finite spin chain is reconstructed from the reduced states of a few neighbouring spins.
  In this setting, state reconstruction is the same as the MPO (i.e.~matrix product operator) reconstruction proposed by \citeauthor{Baumgratz2013}~\parencite{Baumgratz2013}.
  Finally, we generalize both these methods so that they employ long-range measurements instead of relying solely on short-range correlations embodied in such local reduced states.
  Long-range measurements enable the reconstruction of states which cannot be reconstructed from measurements of local few-body observables alone and hereby we improve existing methods for quantum state tomography of quantum many-body systems. 
\end{abstract}

\clearpage

\tocIfNoPreview

\section{Introduction}
\label{sec:rec-intro}

Consider a bipartite quantum state $\rho_{XY}$ which is transformed to a state $\tau_{X'Y'}$ under the action of local
quantum operations $\cN_X: X \to X'$ and $\cN_Y: Y \to Y'$. These local operations could either correspond to (i) undesirable noise (resulting from unavoidable interactions of the quantum
system $XY$ with its environment) or they could correspond to (ii) local measurements made by an experimenter doing quantum state tomography. We are interested in
determining the conditions under which the state $\tau_{X'Y'}$ can be transformed back to the original state $\rho_{XY}$ with maps which act locally on $X'$ and $Y'$.
In the case (i), these would be the conditions under which the effect of the noise can be reversed, whereas in the case (ii) these would be the conditions under which reconstruction of the original state from the outcome of the experimenter's chosen measurements is possible.

The question whether $\tau_{X'Y'}$ can be transformed back to $\rho_{XY}$ can be answered with different methods.
If the transformation is to be achieved with quantum operations, an answer is provided by the Petz recovery map \parencite{Petz2003,Hayden2004} under a condition on the mutual information of the two states. 
If general linear (not necessary completely positive) maps are allowed in the transformation, one can use a matrix reconstruction method.
This matrix reconstruction method is related to MPO (i.e.~matrix product operator) reconstruction and so-called pseudoskeleton (or CUR) matrix decompositions \parencite{Baumgratz2013,Goreinov1997,Caiafa2010}. 
In either case, the construction of the maps which transform $\tau_{X'Y'}$ into $\rho_{XY}$ does not require complete information of $\rho_{XY}$ if suitable maps $\mcn_{X}$ and $\mcn_{Y}$ are used.
In this case, the transformation can be used for efficient quantum state tomography of $\rho_{XY}$ with less measurements than necessary for standard quantum state tomography.

A fundamental quantity in quantum information theory is the quantum relative entropy $D(\rho\|\sigma)$ between a state $\rho$ and a positive semi-definite operator $\sigma$ (see \cref{sec:rec-intro-petz-recovery} for its definition).
It acts as a parent quantity for other entropic quantities arising in quantum information theory, e.g.~von Neumann entropy, conditional entropy and mutual information.
When $\rho$ and $\sigma$ are both states, $D(\rho\|\sigma)$ also has an operational interpretation as a measure of distinguishability between the two states \parencite{Hiai1991,Ogawa2000}.
One of its most important properties is its monotonicity under the joint action of a quantum operation (say, $\cN$).
This is also called the \emph{data processing inequality (DPI)} and is given by,
\[D(\cN(\rho)\|\cN(\sigma)) \leq D(\rho\|\sigma).\]
The condition under which the above inequality is saturated was obtained by Petz \parencite{Petz1986,Petz1988} and has found important applications in quantum information theory.
Petz proved that equality in the DPI holds if and only if there exists a \emph{recovery map}, given by a quantum operation $\cR$ which reverses the action of $\cN$ on both $\rho$ and $\sigma$, i.e.~ $\cR(\cN(\rho)) = \rho$ and $\cR(\cN(\sigma)) = \sigma$.
Petz also obtained an explicit form of such a recovery map, which is often called the \emph{Petz recovery map}.
Petz's condition on the equality in the DPI immediately yields a necessary and sufficient under which the conditional mutual information  $I(A:C|B)$ of a tripartite state $\rho_{ABC}$ is zero \parencite{Hayden2004}, which in turn is the condition under which strong subadditivity of the von Neumann entropy (arguably the most powerful entropic inequality in quantum information theory) is saturated.
Petz's result, when applied to the problem studied in this paper, implies that the original state $\rho_{XY}$ can be recovered from the transformed state $\tau_{X'Y'}$ if and only if the mutual information $I(X:Y)_{\rho}$ of $\rho_{XY}$ is equal to the mutual information $I(X':Y')_{\tau}$ of $\tau_{X'Y'}$ \parencite{Hayden2004,Piani2008}.
Moreover, a valid recovery map is a tensor product of maps acting locally on $X'$ and $Y'$, each having the structure of a Petz recovery map. A detailed discussion of Petz's result and of the quantities on which the Petz recovery maps depend is given in \cref{sec:rec-intro-petz-recovery}. 

The data processing inequality of the relative entropy implies a DPI for the mutual information, 
\begin{align*}I(X':Y')_\tau \leq I(X:Y)_\rho.\end{align*}
The mutual information quantifies the amount of correlations that exist between the two subsystems of a bipartite a quantum state.
Another measure of such correlations is the operator Schmidt rank \parencite{Nielsen2003,Datta2007} of the state, which we denote as $\osr(X:Y)_{\rho}$ for a bipartite state $\rho_{XY}$ (see \cref{eq:rec-intro-osr} for its definition).

In the following, we discuss the main results of this paper. 
We show that the operator Schmidt rank also satisfies a DPI:
\begin{align*}\osr(X':Y')_\tau \leq \osr(X:Y)_\rho,\end{align*}
where $\tau_{X'Y'}$ is the state obtained from $\rho_{XY}$ via the local quantum operations $\cN_X$ and $\cN_Y$, as discussed above.
The DPI for the operator Schmidt rank is directly implied by the fact that the matrix rank satisfies $\rk(MN) \le \rk(M) \rk(N)$ for any two matrices $M$ and $N$ (see \cref{cor:rec-osr-data-processing} for details). 
We show that $\tau_{X'Y'}$ can be transformed into $\rho_{XY}$ with local maps if and only if the DPI of the operator Schmidt rank is saturated.
Our proof does not guarantee that the maps which transform $\tau_{X'Y'}$ into $\rho_{XY}$ are completely positive but it also does not require that $\rho_{XY}$ and $\tau_{X'Y'}$ are positive semidefinite or that $\mcn_{X}$ and $\mcn_{Y}$ are completely positive. 
The proof proceeds by transforming the reconstruction problem into a reconstruction problem for a general, rectangular matrix.
Here, we provide an extension of the known pseudoskeleton decomposition \parencite{Goreinov1997,Caiafa2010}, which is also known as CUR decomposition and which can reconstruct a low-rank matrix from few of its rows and columns.
Our method reconstructs a matrix $M$ from the matrix products $LM$ and $MR$ if the rank of $M$ equals the rank of $LMR$; $L$, $M$ and $R$ are general rectangular matrices.

We explore the relation between Petz recovery and state/MPO reconstruction for the case of $2$, $4$ and $n$ parties.
State/MPO reconstruction, when compared to Petz recovery, is shown to be possible for a strictly larger set of states but requires more information. 

The state of an $\ns$-partite quantum system, such as $\ns$ spins in a linear chain, can be represented as a matrix product operator (MPO) with MPO bond dimensions given by the operator Schmidt ranks $\osr(1\dots k:k+1\dots\ns)$ (between the sites $1,\ldots, k$ and $k+1, \ldots, n$; \parencite{DelasCuevas2013,Schollwoeck2011}).
If the operator Schmidt ranks are all bounded by a constant $D$, the MPO representation is given in terms of $\sim n D^{2}$ complex numbers, which is much less than the number of entries of the density matrix of the $n$-partite quantum system. 
\textcite{Baumgratz2013} presented a condition under which an MPO representation of the state of an $\ns$-partite quantum system can be reconstructed from the reduced states of few neighbouring systems (\emph{MPO reconstruction}).
We will demonstrate that their work implies, for the case where the local operations $\mcn_{X}$ and $\mcn_{Y}$ are partial traces, that $\tau_{X'Y'}$ can be transformed into $\rho_{XY}$ if the two states have equal operator Schmidt rank.

The ability to reconstruct the state of an $\ns$-partite quantum system from reduced states of $l < \ns$ systems, as provided e.g.\ by MPO reconstruction, is advantageous for quantum state tomography of many-body systems.
Standard quantum state tomography requires the expectation values of a number of observables which grows exponentially with $\ns$.
If the full state can be reconstructed from $l$-body reduced states, then the number of observables grows exponentially with $l$ but only linearly with the number of reduced states.
MPO reconstruction uses the reduced states of blocks of $l$ neighbouring sites on a linear chain.
As the number of such blocks increases linearly with $\ns$, MPO reconstruction enables quantum state tomography with a number of observables which increases only linearly with $\ns$.

We call a method for quantum state tomography \emph{efficient} if it requires only polynomially many (in $\ns$) sufficiently simple observables (more details on permitted observables are given in \cref{sec:rec-spinchain-from-marginals}, \cref{rem:rec-efficient}).
Here we assume that exact expectation values are available. 
For a given method to be useful in practice, it is however necessary that the quantum state can be estimated up to a fixed estimation error using approximate expectation values from measurements on at most polynomially many (in $\ns$) copies of the state. 
In this paper, we discuss only the number of necessary observables but not the number of necessary copies of the state.
Numerical simulations indicate that e.g.\ MPO reconstruction and similar methods are efficient also in the number of necessary copies \parencite{Baumgratz2013,Cramer2010,Baumgratz2013a,Lanyon2016}. 

There are multipartite quantum states (e.g.~states of a spin chain) which admit an efficient matrix product state (MPS) or MPO representation but which cannot be reconstructed from reduced states of 
a few of its parties (e.g.~a few neighbouring sites of the spin chain).
The $\ns$-qubit GHZ state is an example of such a state (\cref{sec:rec-comparison-fourpartite}). However, it has been shown that the GHZ state can be reconstructed from a number of observables linear in 
$\ns$, provided global observables (i.e.\ those which act on the whole system) are allowed \parencite{Cramer2010,Baumgratz2013a}.
The necessary observables are given by simple tensor products \parencite{Baumgratz2013a} or simple tensor products and unitary control of few neighbouring sites \parencite{Cramer2010}.
We generalize MPO reconstruction and a similar technique based on the Petz recovery map \parencite{Poulin2011} to use a certain class of long-range measurements which includes those just mentioned as special cases (\cref{sec:rec-spinchain}).
We represent a long-range measurement as a sequence of local quantum operations followed by the measurement of a local observable.
However, a tensor product of single-party observables, whose expectation value can be obtained by a simple, sequential measurement of the single-party observables, already constitutes an allowed long-range measurement.

The example of the GHZ state shows that long-range measurements enable the recovery or reconstruction of a larger set of states than those obtained by local few-body observables. 
Our reconstruction and recovery methods provide a representation of the reconstructed state in terms of a sequence of local linear maps which is equivalent to an MPO representation.
For methods based on the Petz recovery map, the local linear maps are quantum operations and, because of this, a PMPS (locally purified MPS) representation can be obtained (\parencite{Kliesch2014a}, \cref{sec:rec-appendix-pmps-and-seq-prep}).
A PMPS representation is advantageous because it can be computationally demanding to determine whether a given MPO representation represents a positive semidefinite operator \parencite{Kliesch2014a} whereas a PMPS representation always represents a positive semidefinite operator.
Our work on the reconstruction of spin chain states is partially based on similar ideas developed in the context of tensor train representations \parencite{Oseledets2010a} and there is also related work on Tucker and hierarchical Tucker representations \parencite{Caiafa2010,Ballani2013,Caiafa2015}.

The remainder of the paper is structured as follows: 
\Cref{sec:rec-preliminaries} introduces notation, definitions, MPS/MPO representations and known results on the Petz recovery map.
\Cref{sec:rec-reconstr-bipartite-states} shows how a low-rank matrix reconstruction technique enables bipartite state reconstruction, i.e.\ a transformation of $\tau_{X'Y'}$ into $\rho_{XY}$.
We also prove that approximate matrix reconstruction is possible if a low-rank matrix is perturbed by a small high-rank component (\cref{sec:rec-reconstr-matrix-stability}).
We apply the Petz recovery map to the bipartite setting in \cref{sec:rec-recov-bipartite-states} and investigate the relation between Petz recovery and state reconstruction in \cref{sec:rec-comparison}.
Any state which admits Petz recovery is found to also admit state reconstruction.
In \cref{sec:rec-spinchain}, we discuss reconstruction of spin chain states with Petz recovery maps and state reconstruction.
If reconstruction is performed with local reduced states (\cref{sec:rec-spinchain-from-marginals}), a known application of the Petz recovery map \parencite{Poulin2011} and the known MPO reconstruction technique \parencite{Baumgratz2013} are obtained.
In \cref{sec:rec-spinchain-from-long-range-meas}, we reconstruct spin chain states from recursively defined long-range measurements.
We show that successful recovery of a given spin chain state implies successful reconstruction both for local reduced states and for long-range measurements.
The set of states which can be reconstructed with long-range measurements is seen to be strictly larger than the set of states which can be reconstructed with measurements on local reduced states.
Long-range measurements were used in earlier work on the reconstruction of pure states \parencite{Cramer2010} and we show that our methods can recover or reconstruct these states if the same long-range measurements are used.

\section{Preliminaries
  \label{sec:rec-preliminaries}}

\subsection{Notation and basic definitions
  \label{sec:rec-notation}}

In this paper, all Hilbert spaces $\hilb$ are finite-dimensional.
We use capital letters $A$, $B$, $C$, \dots to denote quantum systems with Hilbert spaces $\hilb_{A}$, $\hilb_{B}$, $\hilb_{C}$..., and set $\ldim_{A} = \dim(\cH_A)$. 
For notational simplicity, we often use $A$ to denote both the system and its associated Hilbert space, when there is no cause for confusion.
If $\ns$ systems are involved, we denote their Hilbert spaces by $\hilb_{1}$, \dots, $\hilb_{\ns}$ and tensor products of the latter by $\hilb_{1\dots \ns} = \hilb_{1} \otimes \dots \otimes \hilb_{\ns}$.

We denote the set of linear maps from $A$ to $B$ by $\linmap(A; B) = \linmap(\hilb_{A}; \hilb_{B})$, and the set of linear operators on $A$ by $\linop(A) \equiv \linmap(A; A)$.
If tensor products are involved, we use the notation $\linmap(AB; CD) \equiv \linmap(A,B; C,D) \equiv  \linmap(A \otimes B; C \otimes D)$.
The trace of a linear operator $F \in \linop(A)$ (or a square matrix $F \in \C^{m \times m}$) is denoted by $\tr (F) $.
A quantum state (or density matrix) of a system $A$ is a positive semi-definite operator $\rho \in \linop(A)$ with unit trace.
Let $\qs(A)$ denote the set of quantum states in $\linop(A)$.
In case of a pure quantum state $\rho = \ketbra\psi\psi$, $\ket\psi \in \hilb_{A}$, we refer to both $\rho = \ketbra\psi\psi$ and $\ket\psi$ as the pure state. For any state $\rho \in \cD(A)$, its von~Neumann entropy is defined as $S(A)_{\rho} = - \tr(\rho_{A} \log(\rho_{A}))$. In this paper, all logarithms are taken to base $2$.

For any $F \in \linop(A)$, let $F^{\adjm}$ denote its Hermitian adjoint, $\supp (F)$ its support, $\rk(F)$ its rank, $\norm{F}$ its operator norm (largest singular value) and $\sigma_\text{min}(F)$ its smallest non-zero singular value.
A linear operator $F \in \linop(A)$ is an observable if it is Hermitian. 
The Hilbert--Schmidt inner product on $\linop(A)$ is denoted by
\begin{align}
  \label{eq:rec-intro-hs-inner-product}
  \skp FG &= \tr(F^{\adjm}G) \quad  \forall \, F, G \in \linmap(A).
\end{align}
The vector space $\linop(A)$ becomes a Hilbert space when equipped with this inner product.

The notation $\sop{A}{B}$ denotes the set of linear maps from $\linop(A)$ to $\linop(B)$.
This includes the set of quantum operations (or superoperators) from $A$ to $B$ which are given by linear, completely positive, trace-preserving (CPTP) linear maps $\mcn \in \sop{A}{B}$.
We use the shorthand notation $\mcn \colon A \to B$ to indicate such a quantum operation. 
Given any linear map $\mcn \in \sop{A}{B}$, its Hermitian adjoint (with respect to the Hilbert--Schmidt inner product) is denoted by $\mcn^{\adjs} \in \sop{B}{A}$, i.e., $\skp{F}{\mcn(G)} = \skp{\mcn^{\adjs}(F)}{G},$ for all $F \in \linmap(B)$, $G \in \linmap(A)$. 

Since we are dealing with finite-dimensional Hilbert spaces, all linear operators and maps are represented by matrices.
Given a matrix $M \in \C^{m \times n}$ (or a linear map $M$), let $M^{\adjm}$ denote its conjugate transpose matrix, $\overline{M}$ denotes its (element-wise) complex conjugate matrix, and $M^{\pinv}$ denotes its Moore--Penrose pseudoinverse. The following four properties of the pseudoinverse also define it uniquely
\parencite{Penrose1955,Rado1956}:
\begin{align}
  \label{eq:rec-intro-pseudoinv-properties}
  M M^{\pinv} M &= M,
  & M^{\pinv} M M^{\pinv} &= M^{\pinv},
  & (M M^{\pinv})^{\adjm} &= M M^{\pinv},
  & (M^{\pinv} M)^{\adjm} &= M^{\pinv} M.
\end{align}
Given a real number $t \ge 0$, we define $M_{t}^{\pinv} = (M_{t})^{\pinv}$ where $M_{t}$ is obtained from $M$ by replacing its singular values which are smaller than or equal to $t$ by zero.

For a system $A$, we choose an operator basis $\{ F^{(A)}_{i} \}_{i=1}^{\ldim_{A}^{2}}$ which is orthonormal in the Hilbert--Schmidt inner product:
\begin{align}
  \label{eq:rec-intro-operatorbasis}
  F^{(A)}_{i} &\in \linop(A),
  & \skp{F^{(A)}_{i}}{F^{(A)}_{j}} &= \delta_{ij},
  & i, j \in \{1, \dots, \ldim_{A}^{2}\}.
\end{align}
Given a basis element $F^{(A)}_{i}$, we denote its dual element (in the Hilbert--Schmidt inner product) by $\tilde F^{(A)}_{i}$:
\begin{align}
  \label{eq:rec-intro-hs-dual}
  \tilde F^{(A)}_{i}
  &\in \linmap(\linop(A); \C)
    \colon \quad
    G \mapsto \tilde F^{(A)}_{i}(G) = \skp{F^{(A)}_{i}}{G}.
\end{align}
The identity map, $\id$, on $\linop(A)$ can then be expressed as
\begin{align}
  \label{eq:rec-intro-op-identity}
  \id
  &= \sum_{i=1}^{\ldim_{A}^{2}} F^{(A)}_{i} \tilde F^{(A)}_{i}.
\end{align}
This is nothing but the resolution of the identity operator for the Hilbert space $\linop(A)$. 

Consider a linear map $\mcm \in \sop{Y}{X}$. 
Since the vector spaces $\linop(XY)$ and $\sop{Y}{X}$ have the same finite dimension, $(\ldim_{X}\ldim_{Y})^{2} = \ldim_{X}^{2} \ldim_{Y}^{2}$, it is possible to define a bijective linear map between the two spaces. 
To do so, we define the components of $\rho \in \linop(XY)$ and $\mcm$ in terms of the operator bases from above:
\begin{align}
  \label{eq:rec-intro-components}
  [\rho]_{ij} &= \skp{F^{(X)}_{i} \otimes F^{(Y)}_{j}}{\rho},
  & [\mcm]_{ij} &= \skp{F^{(X)}_{i}}{\mcm(F^{(Y)}_{j})}.
\end{align}
Given a linear operator $\rho \in \linop(XY)$, we define a linear map $\mcm_{\rho}$ by
\begin{align}
  \label{eq:rec-intro-state-to-map}
  [\mcm_{\rho}]_{ij} &= [\rho]_{ij}, \quad \mcm_{\rho} \in \sop{Y}{X}.
\end{align}
We denote the matrix representation of $\mcm_{\rho}$ in the operator basis chosen above by $M_{\rho}$.
The maps $\rho \mapsto \mcm_{\rho}$ and $\rho \mapsto M_{\rho}$ defined by \cref{eq:rec-intro-state-to-map} are of course bijective.
Note that $\rho$ can be represented by a matrix of size $\ldim_{X}\ldim_{Y} \times \ldim_{X}\ldim_{Y}$ while $\mcm_{\rho}$ can be represented by the matrix $M_{\rho} \in \C^{\ldim_{X}^{2} \times \ldim_{Y}^{2}}$.
The transpose map $\mcm^{\tra}$ is defined in the same operator basis, i.e.~$[\mcm^{\tra}]_{ij} = [\mcm]_{ji}$.

Given a linear operator $\rho \in \linop(XY)$, its operator Schmidt rank is given by
\begin{align}
  \label{eq:rec-intro-osr}
  \osr(X:Y)_{\rho}
  &=
    \min\cura{r \colon \quad \rho = \sum_{k=1}^{r} G'_{k} \otimes G''_{k}, \quad G'_{k} \in \linop(X), G''_{k} \in \linop(Y) }.
\end{align}
The operator Schmidt rank is equal to
\begin{align}
  \label{eq:rec-intro-osr-rank}
  \osr(X:Y)_{\rho} &= \rk(\mcm_{\rho}),
\end{align}
which can be shown as follows:
The matrix representation $M_{\rho}$ of $\mcm_{\rho}$ can be written as
\begin{align}
  M_{\rho} &= G H
  & {\hbox{where}}&& G &\in \C^{d_{X}^{2} \times s},
  & H &\in \C^{s \times d_{Y}^{2}}
  & {\hbox{and}} && s &= \rk(M_{\rho}) = \rk(\mcm_{\rho}).
\end{align}
Since the components of $M_{\rho}$ and $\rho$ are related by $[M_{\rho}]_{ij} = [\rho]_{ij}$, we have
\begin{align}
  \label{eq:rec-intro-osr-vs-rank-rho-from-matrix}
  \rho &= \sum_{ij} [\rho]_{ij} F^{(X)}_{i} \otimes F^{(Y)}_{j}
         = \sum_{k=1}^{s} \parena{ \sum_{i} G_{ik} F^{(X)}_{i} } \otimes \parena{ \sum_{j} H_{kj} F^{(Y)}_{j} },
\end{align}
where $G_{ik}$ and $H_{kj}$ are the components of the matrices $G$ and $H$.
This shows that the operator Schmidt rank cannot exceed $s = \rk(M_{\rho})$.
Now suppose that the operator Schmidt rank was less than that, i.e.\ $r = \osr(X:Y)_{\rho} < s = \rk(M_{\rho})$.
Then, a decomposition of $\rho$ as in \cref{eq:rec-intro-osr} implies that
\begin{align}
  \label{eq:rec-intro-osr-vs-rank-matrix-from-rho}
  [M]_{ij} = [\rho]_{ij} = \sum_{k=1}^{r} G_{ik} H_{kj},
  \quad
  G_{ik} = \skp{F^{(X)}_{i}}{G'_{k}},
  \quad
  H_{kj} = \skp{F^{(Y)}_{j}}{G''_{k}},
\end{align}
i.e.\ that $\rk(M_{\rho}) \le r < \rk(M_{\rho})$.
This contradiction shows that the operator Schmidt rank must equal $\rk(M_{\rho}) = \rk(\mcm_{\rho})$. 
The operator Schmidt rank is also equal to the (smallest possible) bond dimension of a matrix product operator (MPO) representation of the linear operator \parencite{DelasCuevas2013}. This is discussed in \cref{sec:rec-intro-mpo}.

\subsection{MPS, MPO and PMPS representations
  \label{sec:rec-intro-mpo}}

In this section, we introduce frequently-used efficient representations of pure and mixed quantum states on $\ns$ systems.
We call a representation \emph{efficient} if it describes a state with a number of parameters (i.e.~complex numbers) which increases at most polynomially with $\ns$. 
The number of parameters of a particular representation of a state is accordingly given by the total number of  entries of all involved vectors, matrices and tensors.
For example, a pure state $\ket\psi \in \C^{\ldim^{\ns}}$ of $\ns$ quantum systems of dimension $\ldim$ has $\ldim^{\ns}$ parameters and is not an efficient representation.
To discuss whether a given representation is efficient or not, we use the following notation: 
For a function $f(\ns)$, we write $f \bigo{\poly(\ns)}$ or $f(\ns) \bigo{\poly(\ns)}$ if there is a polynomial $g(\ns)$ such that $f(\ns) \le g(\ns)$.
We write $f \bigo{\exp(\ns)}$ if there are constants $c_{1}$, $c_{2}$ such that $f(\ns) \le c_{1} \exp(c_{2} \ns)$. 

First, we introduce the matrix product state (MPS) representation (see e.g.\ \parencite{Schollwoeck2011}), which is also known as tensor train (TT) representation \parencite{Oseledets2011}. 
Consider $\ns$ quantum systems of dimensions $\ldim_{1}$, \dots, $\ldim_{\ns}$ respectively, and let $\curc{\ketc{\phi^{(k)}_{i_{k}}}}_{i_{k}=1}^{\ldim_{k}}$ be an orthonormal basis of the $k$-th system.
An MPS representation of a pure state on $\ns$ systems is given by
\begin{align}
  \label{eq:rec-intro-mps-repr}
  \braketc{\phi^{(1)}_{i_{1}} \dots \phi^{(\ns)}_{i_{\ns}}}{\psi}
  =\;
  &
    G_{1}(i_{1}) G_{2}(i_{2}) \dots G_{\ns}(i_{\ns})
\end{align}
where $D_{0} = D_{\ns} = 1$, $G_{k}(i_{k}) \in \C^{D_{k-1} \times D_{k}}$ and $i_{k} \in \{1, \dots, \ldim_{k}\}$.
The condition $D_{0} = D_{\ns} = 1$ ensures that $G_{1}(i_{1})$ and $G_{\ns}(i_{\ns})$ are row and column vectors, while the $G_{k}(i_{k})$ for $k$ between $1$ and $\ns$ can be matrices. 
The matrix sizes $D_{k}$ are called the \emph{bond dimensions} of the representation.
The maximal local dimension and the maximal bond dimension are indicated by $\ldim = \max_{k} \ldim_{k}$ and $D = \max_{k} D_{k}$.
For $\ldim \bigo{\poly(\ns)}$ and $D \bigo{\poly(\ns)}$, the total number of parameters of the MPS representation is $\ns \ldim D^{2} \bigo{\poly(\ns)}$ and the representation is efficient.
The bond dimension $D_{k}$ of any MPS representation of $\ket\psi$ is larger than or equal to the Schmidt rank of $\ket\psi$ for the bipartition $1 \dots k \vert k+1 \dots \ns$ and a representation with all bond dimensions equal to the corresponding Schmidt ranks can always be determined (see, for example, \parencite{Schollwoeck2011}). 
We discuss the analogous property of the matrix product operator (MPO) representation in more detail.

A matrix product operator (MPO) representation \parencite{Zwolak2004,Verstraete2004} of a mixed state on $\ns$ systems is given by
\begin{align}
  \label{eq:rec-intro-mpo-repr}
  \brac{\phi^{(1)}_{i_{1}} \dots \phi^{(\ns)}_{i_{\ns}}} \,\rho\, \ketc{\phi^{(1)}_{j_{1}} \dots \phi^{(\ns)}_{j_{\ns}}}
  =\;
  &
    G_{1}(i_{1}, j_{1}) G_{2}(i_{2}, j_{2}) \dots G_{\ns}(i_{\ns}, j_{\ns})
\end{align}
where $D_{0} = D_{\ns} = 1$, $G_{k}(i_{k}, j_{k}) \in \C^{D_{k-1} \times D_{k}}$ and $i_{k}, j_{k} \in \{1, \dots, \ldim_{k}\}$.
Alternatively, an MPO representation may be given in terms of operator bases $F^{(k)}_{i_{k}}$: 
\begin{align}
  \label{eq:rec-intro-mpo-repr-opbasis}
  \skpc{F^{(1)}_{i_{1}} \otimes \dots \otimes F^{(\ns)}_{i_{\ns}}}{\rho}
  =\;
  &
    G_{1}(i_{1}) G_{2}(i_{2}) \dots G_{\ns}(i_{\ns})
\end{align}
where $D_{0} = D_{\ns} = 1$, $G_{k}(i_{k}) \in \C^{D_{k-1} \times D_{k}}$ and $i_{k} \in \{1, \dots, \ldim_{k}^{2}\}$.
If the operator basis $F^{(k)}_{(i_{k},j_{k})} = \ketbra{\phi^{(k)}_{j_{k}}}{\phi^{(k)}_{i_{k}}}$ is used, \cref{eq:rec-intro-mpo-repr-opbasis} turns into \cref{eq:rec-intro-mpo-repr}.
As before, we denote the maximal local and bond dimensions by $d = \max_{k} d_{k}$ and $D = \max_{k} D_{k}$.
The number of parameters of an MPO representation is at most $\ns \ldim^{2} D^{2}$ and it is an efficient representation if $\ldim \bigo{\poly(\ns)}$ and $D \bigo{\poly(\ns)}$. 
The operator Schmidt ranks of $\rho$ provide lower bounds to the bond dimensions of any MPO representation of $\rho$ \parencite{DelasCuevas2013}:
\begin{align}
  \label{eq:rec-intro-osr-vs-mpo-bonddim}
  \osr(1\dots k : k + 1 \dots \ns)_{\rho} \le D_{k}
\end{align}
This becomes clear if we rewrite \cref{eq:rec-intro-mpo-repr} as follows:
\begin{align}
  \label{eq:rec-intro-mpo-show-osr}
  \rho
  &=
    \sum_{b_{0}\dots b_{\ns}} H_{1}(b_{0}, b_{1}) \otimes H_{2}(b_{1},b_{2}) \otimes \dots \otimes H_{\ns}(b_{\ns-1}, b_{\ns})
\end{align}
where $H_{k}(b_{k-1},b_{k}) \in \C^{\ldim_{k} \times \ldim_{k}}$ and $[H_{k}(b_{k-1},b_{k})]_{i_{k},j_{k}} = [G_{k}(i_{k},j_{k})]_{b_{k-1},b_{k}}$. 
The sum runs over $b_{k} \in \{1, \dots, D_{k}\}$. 
It can also be shown that a representation with equality in \cref{eq:rec-intro-osr-vs-mpo-bonddim} always exists (see, for example, \parencite{Schollwoeck2011}).
If the linear operator $\rho$ represented by an MPO is a quantum state, it is desirable to ensure that $\rho$ is positive semi-definite.
However, deciding whether a given MPO represents a positive semi-definite operator is an NP-hard problem in the number of parameters of the representation \parencite{Kliesch2014a}, i.e.\ a numerical solution in polynomial (in $n$) time may not be obtained. 
As an alternative, one can use a PMPS (locally purified MPS) representation of the mixed state.
A PMPS representation represents a positive semidefinite linear operator by definition.
PMPS representations are also called \emph{evidently positive} representations and they are introduced in \cref{sec:rec-appendix-pmps-and-seq-prep}. 

\newcommand{\preparemap}{\mathcal W}

Suppose that a quantum state $\rho \in \qs(\hilb_{1\dots \ns})$ was prepared via quantum operations
\\
$\issop{\preparemap_{k}}{\hilb_{k-1}}{\hilb_{k-1,k}}$,
i.e.
\begin{align}
  \label{eq:rec-seq-prep-maintext}
  \rho = \preparemap_{\ns} \preparemap_{\ns-1} \dots \preparemap_{3} \preparemap_{2}(\sigma)
\end{align}
where $\sigma \in \qs(\hilb_{1})$.
Clearly, this is an efficient representation of the quantum state $\rho$ as it is described by at most $\ns \ldim^{6}$ parameters.
It is known that such a representation can be efficiently, i.e.~with at most $\poly(\ns)$ computational time, converted into an MPO representation or a PMPS representation \parencite{Poulin2011,Kliesch2014a}.
\cref{sec:rec-appendix-pmps-and-seq-prep} provides the details of the conversion and of the PMPS (locally purified MPS) representation.
The state recovery and reconstruction techniques presented in \cref{sec:rec-spinchain} provide a representation of the reconstructed state which is similar to \cref{eq:rec-seq-prep-maintext}.
\Cref{lem:rec-spinchain-seq-prep} in the appendix provides PMPS and MPO representations of the recovered state for techniques based on the Petz recovery map and an MPO representation of the reconstructed state for state reconstruction results.

\subsection{The Petz recovery map
  \label{sec:rec-intro-petz-recovery}}

The (quantum) relative entropy, for two quantum states $\rho, \sigma \in \qs(A)$, was defined by \textcite{Umegaki1962} as
\begin{align}
  \label{eq:rec-intro-relent}
  \relent(\rho \| \sigma) = \tr(\rho \log(\rho) - \rho \log(\sigma)),
\end{align}
if $\supp \rho \subseteq \supp \sigma$, and is set equal to $+\infty$ otherwise.
For a bipartite quantum state $\rho \equiv \rho_{AB} \in \qs(AB)$, the mutual information between the subsystems $A$ and $B$ is defined in terms of the von Neumann entropies of $\rho_{AB}$ and its reduced states $\rho_A = \tr_B \rho_{AB}$ and $\rho_A = \tr_B \rho_{AB}$:
\begin{align}
  \label{eq:rec-intro-mi-vn}
  I(A:B)_{\rho} &= S(A)_{\rho} + S(B)_{\rho} - S(AB)_{\rho}.
\end{align}
It can also be expressed in terms of the relative entropy as follows:
\begin{align}
  \label{eq:rec-intro-mi}
  I(A:B)_{\rho} &= D(\rho \| \rho_{A} \otimes \rho_{B}).
\end{align}

The quantum conditional mutual information (QCMI) of a tripartite quantum state $\rho \in \qs(ABC)$ is given by
\begin{align}
  \label{eq:rec-intro-cmi}
  I(A:C|B) := I(A:BC) - I(A:B) = I(AB:C) - I(B:C),
\end{align}
and is expressed in terms of the von~Neumann entropy as follows:
\begin{align}
  \label{eq:rec-intro-cmi-vn}
  I(A:C|B) &= S(AB) + S(BC) - S(ABC) - S(B).
\end{align}

As mentioned in the Introduction, a fundamental property of the quantum relative entropy is its monotonicity under quantum operations. This is given by the data processing inequality (DPI): for quantum states $\rho, \sigma \in \cD(A)$ and a 
quantum operation $\cN$ acting on $\cD(A)$, 
\begin{align}\label{eq:DPI}
  D(\rho \| \sigma) \ge D(\mcn (\rho) \| \mcn (\sigma)).
\end{align}
For the choice $\rho= \rho_{ABC}$, $\sigma=\rho_{AB} \otimes \rho_C$ and $\cN = \tr_A \otimes \id_{BC}$, the DPI~\eqref{eq:DPI} implies that the QCMI of a tripartite state $\rho_{ABC}$ is always non-negative. Using the definition~\eqref{eq:rec-intro-cmi-vn} of the QCMI, we further infer that
\begin{align}\label{SSA}
S(ABC) + S(B) &\leq S(AB) + S(BC),
\end{align}
which is the well-known \emph{strong subadditivity} (SSA) property of the von Neumann entropy.

A necessary and sufficient condition for equality in the DPI~\eqref{eq:DPI} was derived in~\textcite{Hayden2004,Petz2003} and is stated in the following theorem.
\begin{thm}[Petz recovery map]
  \label{thm:rec-petz}
  Let $\rho, \sigma \in \cD(A)$ be quantum states and $\mcn \colon A \to A'$ be a quantum operation. The equality
  \begin{align}
    \label{eq:rec-petz-relent}
    \relent(\mcn (\rho) \| \mcn (\sigma)) = \relent(\rho \| \sigma)
  \end{align}
  holds if and only if there is a linear CPTP map $\cptprec \colon A' \to A$ which satisfies
  \begin{align}
    \label{eq:rec-petz-recov}
    \rho &= \cptprec \mcn (\rho), \quad {\hbox{and}} \quad
      \sigma = \cptprec \mcn (\sigma).
  \end{align}
  If the above condition is satisfied, the so-called Petz recovery map $\petzrec_{\sigma,\mcn}$ satisfies the two equations.
  On the support of $\mcn(\rho)$ and for $\omega \in \linop(A')$, this map is given by
  \begin{align}
    \label{eq:rec-petz-map}
    \petzrec_{\sigma,\mcn}(\omega) &= \sigma^{1/2} \mcn^{\adjs} \left( \mcn(\sigma)^{-1/2} \omega \mcn(\sigma)^{-1/2} \right) \sigma^{1/2}
  \end{align}
  where $\mcn^{\adjs}$ is the Hermitian adjoint of $\mcn$ in the Hilbert--Schmidt inner product.
\end{thm}
Further, \citeauthor{Hayden2004} derived the following necessary and sufficient condition on the structure of tripartite states satisfying equality in the SSA~\eqref{SSA} (see Theorem~6 in \parencite{Hayden2004}).

\begin{thm}
  \label{thm:petz-structure}  
  Let $\rho \equiv \rho_{ABC}\in \cD(ABC)$ be a tripartite quantum state.
  The equality $I(A:BC)_{\rho} = I(A:B)_{\rho}$ (which is equivalent to equality in the SSA~(\ref{SSA})) holds if and only if there is a decomposition of $\cH_B$ into $\hilb_{B_{L_{j}}}$ and $\hilb_{B_{R_{j}}}$ as
\begin{align}
\cH_B &= \bigoplus_{j} \cH_{B_{L_{j}}} \otimes \cH_{B_{R_{j}}},
\end{align}
such that
$\rho$ can be
  written as
  \begin{align*}
    \rho_{ABC} = \bigoplus_{j} p_{j} \rho_{AB_{L_{j}}} \otimes \rho_{B_{R_{j}}C}, \quad
  \end{align*}
  for a probability distribution $\{p_{j}\}$, and sets of quantum states
$\{\rho_{AB_{L_{j}}}\}$ and $\{\rho_{B_{R_{j}}C}\}$.
\end{thm}

\section{Reconstruction of bipartite states subjected to local operations
  \label{sec:rec-reconstr-bipartite-states}}

Let $\rho_{XY} \in \qs(XY)$ be a bipartite state and let $\tau_{X'Y'} \in \qs(X'Y')$ be a state obtained from $\rho_{XY}$ by the action of local operations:
\begin{align}
\tau_{X'Y'} &= (\mcn_X \otimes \mcn_Y) \rho_{XY},
\end{align}
where $\mcn_X: X \to X'$ and $\mcn_Y: Y \to Y'$ denote quantum operations (or more generally, linear maps).
We are interested in the conditions under which the original state $\rho_{XY}$ can be reconstructed from $\tau_{X'Y'}$ with local maps, i.e.~$\rho_{XY} = (\recmap_{X'} \otimes \recmap_{Y'})(\tau_{X'Y'})$ with reconstruction maps $\recmap_{X'} \colon X' \to X$ and $\recmap_{Y'} \colon Y' \to Y$. 

Our reconstruction scheme is particularly useful for states $\rho$ with low operator Schmidt rank because then a reconstruction of $\rho$ can be achieved with fewer measurements than required for standard quantum state tomography, as discussed in \cref{rem:rec-efficient,rem:rec-mat-longrange-eff} (see also \cref{sec:rec-comparison-fourpartite}).
The operator Schmidt rank of $\rho$ is equal to the rank of the matrix $M_{\rho}$ (\cref{eq:rec-intro-state-to-map,eq:rec-intro-osr-rank}; $M_{\rho}$ has size $d_{X}^{2} \times d_{Y}^{2}$).
Hence, in \cref{sec:rec-reconstr-matrices} we first consider the more general problem of reconstruction of low-rank matrices (which are not necessarily states).
\Cref{sec:rec-reconstr-matrix-stability} discusses the stability of our matrix reconstruction technique and \cref{sec:rec-reconstr-bipartite-states-subsec} shows how it can be used to reconstruct a quantum state.

\subsection{Reconstruction of low-rank matrices
  \label{sec:rec-reconstr-matrices}}

Suppose that we want to obtain a matrix $M \in \C^{m \times n}$ but we only know the entries of the matrix products $LM$ and $MR$ where $L$ and $R$ are $r\times m$ and $n \times s$ complex matrices.
We refer to $LM$, $MR$ and $LMR$ as the $\emph{marginals}$ of the matrix $M$.
\Cref{prop:rec-matrixdec} states that $M$ can indeed be obtained from $LM$ and $MR$ if the condition $\rk(LMR) = \rk(M)$ holds.
This rank condition implies $r, s \ge \rk(M)$.
If the rank of $M$ is much smaller than its maximal value, $\min\{m, n\}$, this provides a way to obtain $M$ from $LM$ and $MR$ which, taken together, have much fewer entries than $M$.
If the matrices $L$ and $R$ are restricted to submatrices of permutation matrices, the matrix products $LM$ and $MR$ comprise selected rows and columns of $M$.
In this case, \cref{prop:rec-matrixdec} provides a reconstruction of a low-rank matrix $M$ from few rows and columns (cf.\ \parencite{Goreinov1997}).

\begin{prop}
  \label{prop:rec-matrixdec}
  Let $L \in \C^{r \times m}$, $M \in \C^{m \times n}$ and $R \in \C^{n \times s}$ be matrices.
  Then
  \begin{align}
    \rk(LMR) = \rk(M) \quad\Leftrightarrow\quad \exists\, X \in \C^{s \times r}\colon M = MR \, X \, LM.
    \label{eq:prop-rec-matrixdec}
  \end{align}
  If the condition is satisfied, $M = MR \, X \, LM$ holds for any matrix $X$ with $CXC = C$, $C = LMR$.
  The Moore--Penrose pseudoinverse $X = (LMR)^{\pinv}$ has the required property $CXC = C$.

  Furthermore, $\rk(LM) = \rk(M)$ implies $\rk(LMR) = \rk(MR)$.
\end{prop}

\begin{proof}
  ``$\Rightarrow$'' of \cref{eq:prop-rec-matrixdec}: Assume that $\rk(LMR) = \rk(M)$ holds.
  The property $C X C = C$, $C = LMR$ implies that $LMRX u = u$ holds for all $u \in \rg(LMR)$.
  Let $q = \rk(M) = \rk(LMR)$. 
  Let $u_{i}, \ldots, u_{q}$ be a basis of $\rg(LMR)$ and set $v_{i} = X u_{i}$, $w_{i} = MR v_{i}$.
  The $v_{i}$ are linearly independent because $LMR v_{i} = LMR X u_{i} = u_{i}$. The $w_{i}$ are linearly independent because $L w_{i} = LMR v_{i} = u_{i}$. The $w_{i}$ are a linearly independent sequence of length $q = \rk(M)$ and they satisfy $w_{i} \in \rg(M)$, i.e.~they are a basis of $\rg(M)$. Now observe
  \begin{align}
    MR\, X\, L w_{i} = MR\, X u_{i}
    = MR\, v_{i} = w_{i}.
  \end{align}
  As a consequence, $MR\, X\, L$ maps any vector from $\rg(M)$ to itself. Accordingly, $(MR \, X \, L)M = M$ holds.

  $\rk(LM) = \rk(M)$ implies $\rk(LMR) = \rk(MR)$:
  The equality $\rk(LM) = \rk(M)$ implies $M = M (LM)^{\pinv} LM$ (use the ``$\Rightarrow$'' direction of \cref{eq:prop-rec-matrixdec} for $R = \idm$).
  As a consequence, $MR = M (LM)^{\pinv} LMR$ and $\rk(MR) \le \rk(LMR)$ hold.
  The converse inequality $\rk(LMR) \le \rk(MR)$ always holds and we arrive at $\rk(LMR) = \rk(MR)$.
  
  ``$\Leftarrow$'' of \cref{eq:prop-rec-matrixdec}: Assume that $M = MR\, X\, LM$ holds for some matrix $X$.
  The equality $M = MR\, X \, LM$ implies $\rk(M) \le \rk(MR)$ and $\rk(M) \le \rk(LM)$.
  The converse inequalities $\rk(MR) \le \rk(M)$ and $\rk(ML) \le \rk(M)$ always hold.
  As a consequence, we have $\rk(LM) = \rk(M)$ and $\rk(MR) = \rk(M)$.
  Above, we saw that the former equality implies $\rk(LMR) = \rk(MR)$ which, together with the latter equality $\rk(MR) = \rk(M)$, proves the theorem.
\end{proof}

\begin{rem}
A violation of the rank condition $\rk(LMR) = \rk(M)$ does not in general imply that there is no method to obtain $M$ from $LM$ and $MR$. As a trivial example, consider $L = \idm$ and $R = 0$. Then, the rank condition is violated for all $M \ne 0$, but $M$ is obtained trivially from $LM = M$.
\end{rem}

\begin{rem}[Related work]
\label{rem:rec-matrixdec}
\Cref{prop:rec-matrixdec} states that $M$ can be obtained from $LM$ and $MR$ if $\rk(LMR) = \rk(M)$ holds.
Special cases of \cref{prop:rec-matrixdec} have appeared before in several places.
If $r = s = \rk(M)$ and $L$ and $R$ select exactly $r = \rk(M)$ rows and columns of $M$, the decomposition $M = MR (LMR)^{-1} LM$ is known as skeleton decomposition of $M$ \parencite{Goreinov1997}.
Decompositions of the form $M = MR \, X\,  LM$ where $L$ and $R$ select rows and columns of $M$ are known as pseudoskeleton/CUR decomposition of $M$ and it has been recognized that the truncated Moore--Penrose pseudoinverse $X = (LMR)_{\tau}^{\pinv}$ may provide a good approximation if $r = s < \rk(M)$ and suitable rows, columns and threshold $\tau$ are chosen \parencite{Goreinov1997}; we come back to the case of approximately low rank in \cref{sec:rec-reconstr-matrix-stability}.
The case $r = s = \rk(M)$, $X = (LMR)^{\pinv}$ is contained in the results on tensor decompositions by \textcite{Caiafa2015}.
This matrix decomposition with $X = (LMR)^{\pinv}$, restricted $L$ and $R$ but general $r, s \ge \rk(M)$ forms the basis of MPO reconstruction \parencite{Baumgratz2013} which is discussed in \cref{sec:rec-spinchain}. 
\end{rem}

\subsection{Stability of the reconstruction under perturbation
  \label{sec:rec-reconstr-matrix-stability}}

Suppose that we have a matrix $S$ which satisfies the rank condition
\begin{align}
  \rk(LSR) = \rk(S)
\end{align}
for given matrices $L$ and $R$. We want to reconstruct the perturbed matrix
\begin{align}
  M = S + E, \quad \epsilon = \norm{E}.
\end{align}
and $\epsilon = \norm{E}$ is the operator norm of the perturbation.
In \cref{thm:rec-matrixdec-stability} we provide a reconstruction $\check M_{\tau}$ and show that it is close to $M$ if the operator norm $\epsilon$ of the perturbation $E$ is small enough.
A bound on the distance in operator norm between the reconstruction $\check M_{\tau}$ and $M$ is provided by
\begin{align}
  \norm{\check M_{\tau} - M} \le \norm{\check M_{\tau} - S} + \norm{S - M} \le \norm{\check M_{\tau} - S} + \epsilon
\end{align}
and \cref{thm:rec-matrixdec-stability} provides a bound on $\norm{\check M_{\tau} - S}$.

Recall that given a matrix $M$, we define $M_{\tau}^{\pinv} = (M_{\tau})^{\pinv}$ and $M_{\tau}$ is given by $M$ with singular values smaller or equal to $\tau$ replaced by zero.

\begin{thm}
  \label{thm:rec-matrixdec-stability}

  Let $M = S + E$, $\rk(S) = \rk(LSR)$, $\eta = \norm{L} \norm{S} \norm{R}$, $\gamma = \sigma_\text{min}(LSR) / \eta$, $\epsilon = \norm{E} / \norm{S}$. Let $\gamma > 2\epsilon$ and $\epsilon \le \tau < \gamma - \epsilon$. Then
  \begin{align*}
    \norm{ \check M_{\tau}  - S }
    &\le
    \frac{\norm{S}}{\gamma - \epsilon}
    \left(
    \frac{4\epsilon}{\gamma} + 2\epsilon + \epsilon^{2}
    \right),
    \le
    \frac{7\epsilon \norm{S}}{\gamma(\gamma - \epsilon)},
    &
    \check M_{\tau} &= MR (LMR)_{\eta\tau}^{\pinv} LM.
  \end{align*}

\end{thm}

\begin{proof}
  We prove the proposition for $\norm{L} = \norm{S} = \norm{R} = 1$ (without loss of generality as explained in \cref{sec:rec-appendix-normalization}). 
  We have $M = S + E$, $\epsilon = \norm{E}$, $\gamma = \sigma_\text{min}(LSR)$ and $LMR = LSR + LER$ with $\norm{LER} \le \epsilon$.
  We insert $S = SR (LSR)^{\pinv} LS$ and use \cref{lem:rec-pi-perturbation-using-wedin} (provided at the end of this subsection):
  \begin{align}
  &\norm{MR (LMR)_{\tau}^{\pinv} LM - S}
  \\ \le \, 
  &
    \norm{(LMR)_{\tau}^{\pinv} - (LSR)^{\pinv}}
    + 2 \norm{(LMR)_{\tau}^{\pinv}} \epsilon
    + \norm{(LMR)_{\tau}^{\pinv}} \epsilon^{2}
  \\ \le \,
  &
    \frac{4\epsilon}{\gamma(\gamma - \epsilon)} 
    + \frac{2\epsilon}{\gamma - \epsilon}
    + \frac{\epsilon^{2}}{\gamma - \epsilon}
    \le
    \frac{7 \epsilon}{\gamma(\gamma - \epsilon)}.
  \end{align}
Note that by premise, we have $\epsilon < \gamma < 1$.
As a consequence, $1 \le \frac1\gamma$ and $\epsilon \le 1$ (which were used in the last equation) hold.
This proves the theorem.
\end{proof}

For the interpretation of the theorem, it is convenient to use the case with $\norm{L} = \norm{S} = \norm{R} = 1$ and $\eta = 1$.
\Cref{thm:rec-matrixdec-stability} shows that the reconstruction $\check M_{\tau}$ reconstructs the low-rank component $S$ of $M = S + E$ up to a small error if the smallest singular value $\gamma$ of the low-rank component $LSR$ is much larger than the norm $\epsilon$ of the noise component. In addition, the threshold $\tau$ must be chosen larger than the noise norm $\epsilon$ but smaller than $\gamma - \epsilon$.
\Cref{sec:rec-appendix-stability-optimal} discusses examples which show that the bound from \cref{thm:rec-matrixdec-stability} is optimal up to constants and that the reconstruction error can diverge as $\epsilon$ approaches zero if small singular values in $LMR$ are not truncated.

Choosing a suitable threshold $\tau$ is equivalent to estimating the rank of the low rank contribution $S$. If the rank and support of $S$ are known, the measurements $L$ and $R$ can be chosen such that $LSR$ becomes invertible. For this special case, an upper bound on the reconstruction error has been given by \textcite{Caiafa2015}. Their bound also includes constants which depend on $LSR$ and may diverge as $\gamma$ approaches zero. In \cref{sec:rec-appendix-caiafa-based-bound}, we generalize their approach to our more general setting and obtain a bound which is similar to \cref{thm:rec-matrixdec-stability}.

The following Lemma was used in the proof of \cref{thm:rec-matrixdec-stability}:
\begin{lem}
  \label{lem:rec-pi-perturbation-using-wedin}
  Let $A, B, F \in \C^{m \times n}$.
  Let $\gamma = \sigma_\text{min}(A)$, $B = A + F$ with $\norm{F} \le \epsilon$.   
Let $\gamma > 2\epsilon$ and choose $\tau$ such that $\epsilon \le \tau < \gamma - \epsilon$. 
  Then $\sigma_\text{min}(B_{\tau}) \ge \gamma - \epsilon$, $\norm{B_{\tau}^{\pinv}} \le 1 / (\gamma - \epsilon)$, $B_{\tau} = B_{\epsilon}$ and $\norm{B - B_{\tau}} \le \epsilon$. In addition,
  \begin{align}
    \norm{B_{\tau}^{\pinv} - A^{\pinv}} \le
    \frac{4\epsilon}{\gamma(\gamma - \epsilon)}.
    \label{eq:rec-pi-perturbation-using-wedin}
  \end{align}
\end{lem}

\begin{proof}
  The singular values of $B = A + F$ satisfy
  (see e.g.\ \parencite{Stewart1977})
  \begin{align*}
    \abs{\sigma_{i}(A + F) - \sigma_{i}(A)} 
    \le \sigma_{1}(F) \le \epsilon
  \end{align*}
  and therefore, with $r = \rk(A)$, we obtain
  \begin{align*}
    \sigma_{1}(B) \ge \ldots \ge
    \sigma_{r}(B) \ge \gamma - \epsilon > \tau \ge \epsilon \ge \sigma_{r+1}(B) \ge \ldots \ge 0.
  \end{align*}
  This shows already everything except the inequality in \cref{eq:rec-pi-perturbation-using-wedin}.
  To show the latter one, we use $\norm{X^{\pinv} - A^{\pinv}} \le 2 \norm{X^{\pinv}} \norm{A^{\pinv}} \norm{X - A}$ \parencite{Wedin1973,Stewart1977}.
  Inserting $\norm{B_{\tau} - A} \le \norm{B_{\tau} - B} + \norm{B - A}$ shows the desired inequality.
\end{proof}

\subsection{Reconstruction of bipartite states
  \label{sec:rec-reconstr-bipartite-states-subsec}}

Let $\rho \in \qs(XY)$ be a bipartite quantum state and let $\tau \in \qs(X'Y')$ be a state obtained from it by the action of local operations:
\begin{align}
  \tau &= (\mcn_X \otimes \mcn_Y) \rho,
\end{align}
where $\mcn_X: X \to X'$ and $\mcn_Y: Y \to Y'$ denote quantum operations.
We are interested in the conditions under which the original state $\rho$ can be reconstructed from $\tau$ with local quantum operations, i.e.~$\rho = (\recmap_{X'} \otimes \recmap_{Y'})\tau$ with $\recmap_{X'} \colon X' \to X$ and $\recmap_{Y'} \colon Y' \to Y$.
This question can be answered with the matrix decomposition from \cref{sec:rec-reconstr-matrices} without using the positivity properties of $\rho$, $\mcn_{X,Y}$ and $\recmap_{X',Y'}$. 
The result is provided by the following \namecref{thm:rec-reconstr-bipartite}:

\begin{thm}
  \label{thm:rec-reconstr-bipartite}
  Let $\rho \in \linop(XY)$ be a linear operator and let $\mcn_{X} \in \sop{X}{X'}$,
  $\mcn_{Y} \in \linmap(\linop(Y);$ $\linop(Y'))$
  be linear maps.
  Let $\tau = \mcn_{X} \otimes \mcn_{Y}(\rho)$.
  The original operator $\rho$ can be reconstructed from $\tau$ with local linear maps $\recmap_{X'}$, $\recmap_{Y'}$, i.e.\ 
  \begin{align}
    \label{eq:rec-reconstr-bipartite-reconstr}
    \rho = (\recmap_{X'} \otimes \recmap_{Y'})(\tau)
  \end{align}
  if and only if the following equality holds:
  \begin{align}
    \label{eq:rec-reconstr-bipartite-rkcond}
    \osr(X':Y')_{\tau} = \osr(X:Y)_{\rho}.
  \end{align}
  If the condition is satisfied, the following linear maps $\recmap_{X'} \in \sop{X'}{X}$ and $\recmap_{Y'} \in \sop{Y'}{Y}$ satisfy \cref{eq:rec-reconstr-bipartite-reconstr}:
  \begin{align}
    \label{eq:rec-reconstr-bipartite-maps}
    \recmap_{X'}
    &= \matrec_{\mcm_{\rho}\mcn_{Y}^{\tra},\mcn_{X}},
    &
      \recmap_{Y'}
    &=
      \matrec_{\mcm_{\rho}^{\tra}\mcn_{X}^{\tra},\mcn_{Y}},
    &
      \matrec_{\mathcal L,\mathcal N}
    &=
      \mathcal L (\mathcal N\mathcal L)^{\pinv}.
  \end{align}
  The operators $(\id \otimes \mcn_{Y})(\rho)$ and $(\mcn_{X} \otimes \id)(\rho)$ are sufficient to construct the two $\matrec_{\mathcal L, \mcn}$ maps if $\mcn_{X}$ and $\mcn_{Y}$ are known.
  The superscript M indicates that the reconstruction map is based on matrix reconstruction.
  If the condition is satisfied, the following equation also holds for $\recmap_{X'}$ and $\recmap_{Y'}$ from \cref{eq:rec-reconstr-bipartite-maps}:
  \begin{align}
    \label{eq:rec-reconstr-partial}
    \rho &= (\recmap_{X'}\mcn_{X} \otimes \id)(\rho) = (\id \otimes \recmap_{Y'}\mcn_{Y})(\rho).
  \end{align}
\end{thm}

\begin{rem}
If the rank condition \eqref{eq:rec-reconstr-bipartite-rkcond} is satisfied and $\sigma^{(XY')} = (\id \otimes \mcn_{Y})(\rho)$ and $\sigma^{(X'Y)} = (\mcn_{X} \otimes \id)(\rho)$ are given, one can obtain $\rho$ by computing the maps $\recmap_{X'}$ and $\recmap_{Y'}$ (\cref{eq:rec-reconstr-bipartite-maps}), $\tau = \mcn_{X} \otimes \id(\sigma^{(X'Y)})$ and $\rho = (\recmap_{X'} \otimes \recmap_{Y'})(\tau)$.
More directly, one can also compute only the map $\recmap_{X'}$ followed by computing $\rho = (\recmap_{X'} \otimes \id)(\sigma^{(X'Y)})$ (\cref{eq:rec-reconstr-partial}).
The two options correspond to using either \cref{eq:rec-reconstr-bipartite-proof-rho-withmaps} or \cref{eq:rec-reconstr-bipartite-proof-rho-direct} to obtain $\mcm_{\rho}$.
\end{rem}

\begin{proof}[of \cref{thm:rec-reconstr-bipartite}]
  The operator $\tau$ is given by $\tau = \mcn_{X} \otimes \mcn_{Y}(\rho)$, therefore $\osr(X':Y')_{\tau} \le \osr(X:Y)_{\rho}$ always holds (\cref{cor:rec-osr-data-processing}). 
  Let \cref{eq:rec-reconstr-bipartite-reconstr} hold.
  Again by \cref{cor:rec-osr-data-processing}, the converse inequality $\osr(X:Y)_{\rho} \le \osr(X':Y')_{\tau}$ also holds.
  As a consequence, the two operator Schmidt ranks must be equal.

  Let the rank condition \eqref{eq:rec-reconstr-bipartite-rkcond} hold.
  \Cref{lem:rec-state-as-matrix-composition} and $\tau = \mcn_{X} \otimes \mcn_{Y}(\rho)$ 
  We obtain:
  \begin{subequations}
    \begin{align}
      \label{eq:rec-reconstr-bipartite-proof-start}
      & [\recmap_{X'} \otimes \recmap_{Y'}(\tau)]_{ij}
        = [\recmap_{X'} \mcm_{\tau} \recmap_{Y'}^{\tra}]_{ij}
      \\ &=
           \label{eq:rec-reconstr-bipartite-proof-rho-withmaps}
           [\mcm_{\rho} \mcn_{Y}^{\tra} (\mcn_{X} \mcm_{\rho} \mcn_{Y}^{\tra})^{\pinv}
           (\mcn_{X} \mcm_{\rho} \mcn_{Y}^{\tra})
           (\mcn_{X} \mcm_{\rho} \mcn_{Y}^{\tra})^{\pinv} \mcn_{X} \mcm_{\rho}]_{ij}
      \\ &=
           \label{eq:rec-reconstr-bipartite-proof-rho-direct}
           [\mcm_{\rho} \mcn_{Y}^{\tra} (\mcn_{X} \mcm_{\rho} \mcn_{Y}^{\tra})^{\pinv}
           \mcn_{X} \mcm_{\rho}]_{ij}
      \\ &=
           \label{eq:rec-reconstr-bipartite-proof-result}
           [\mcm_{\rho}]_{ij}.
    \end{align}
  \end{subequations}
  In \cref{eq:rec-reconstr-bipartite-proof-start} and \cref{eq:rec-reconstr-bipartite-proof-rho-withmaps}, we used \cref{lem:rec-state-as-matrix-composition} and inserted the maps from \cref{eq:rec-reconstr-bipartite-maps}.
  In \cref{eq:rec-reconstr-bipartite-proof-rho-direct}, we used the property $A^{\pinv} A A^{\pinv} = A^{\pinv}$ of the Moore--Penrose pseudoinverse.
  In \cref{eq:rec-reconstr-bipartite-proof-result}, we applied the matrix reconstruction result from \cref{prop:rec-matrixdec}.
  The therefor needed rank condition $\rk(\mcn_{X} \mcm_{\rho} \mcn_{Y}^{\tra}) = \rk(\mcm_{\rho})$ is equivalent to the rank condition \eqref{eq:rec-reconstr-bipartite-rkcond} because of $\tau = \mcn_{X} \otimes \mcn_{Y}(\rho)$ and \cref{lem:rec-state-as-matrix-composition}.
  This shows that \cref{eq:rec-reconstr-bipartite-reconstr} holds if \cref{eq:rec-reconstr-bipartite-rkcond} is assumed and the maps from \cref{eq:rec-reconstr-bipartite-maps} are inserted.
  \cref{eq:rec-reconstr-partial} can be shown by omitting $\recmap_{Y'}$ (or $\recmap_{X'}$) from left hand side of \cref{eq:rec-reconstr-bipartite-proof-start}. This finishes the proof of the theorem.
\end{proof}

The remainder of the section provides the ingredients used in the preceding proof.
It also provides a \emph{data processing inequality (DPI)} for the operator Schmidt rank which is used below.

\begin{lem}
  \label{lem:rec-state-as-matrix-composition}

Let $\rho \in \linop(XY)$ be a linear operator and let $\mcn_{X}: \cB(X) \to \cB(X')$, $\mcn_{Y}: \cB(Y) \to \cB(Y')$, be linear maps. Set $\tau = (\mcn_{X} \otimes \mcn_{Y})(\rho)$. Then
  \begin{align}
    \label{eq:rec-state-as-matrix-composition}
    [\mcm_{\tau}]_{ij} &= [\mcn_{X} \mcm_{\rho} \mcn_{Y}^{\tra}]_{ij}.
  \end{align}
\end{lem}

\begin{proof}
  First, note that
  \begin{subequations}
    \begin{align}
      [(F^{(X)}_{k} \tilde F^{(X)}_{k}) \otimes (F^{(Y)}_{l} \tilde F^{(Y)}_{l})](\rho)
      &=
        (F^{(X)}_{k} \otimes F^{(Y)}_{l})[(\tilde F^{(X)}_{k} \otimes \tilde F^{(Y)}_{l})(\rho)]
      \\ &=
           (F^{(X)}_{k} \otimes F^{(Y)}_{l}) [\rho]_{kl}.
    \end{align}
  \end{subequations}
  The proof involves several basic steps:
  \begin{subequations}
  \begin{align}
    [\mcm_{\tau}]_{ij}
    &=
      [\tau]_{ij}
      =
      [\mcn_{X} \otimes \mcn_{Y}(\rho)]_{ij}
    \\ &=
         \skpc{F^{(X')}_{i} \otimes F^{(Y')}_{j}}{\mcn_{X} \otimes \mcn_{Y}(\rho)}
    \\ &=
         \sum_{kl}\skpc{F^{(X')}_{i} \otimes F^{(Y')}_{j}}{(\mcn_{X} \otimes \mcn_{Y})[(F^{(X)}_{k} \tilde F^{(X)}_{k}) \otimes (F^{(Y)}_{l} \tilde F^{(Y)}_{l})](\rho)}
    \\ &=
         \sum_{kl}
         \skpc{F^{(X')}_{i}}{\mcn_{X}(F^{(X)}_{k})}
         \skpc{F^{(Y')}_{j}}{\mcn_{Y}(F^{(Y)}_{l})}
         [\rho]_{kl}
    \\ &=
         \sum_{kl}
         [\mcn_{X}]_{ik} [\mcn_{Y}]_{jl} [\mcm_{\rho}]_{kl}
    \\ &=
         [\mcn_{X} \mcm_{\rho} \mcn_{Y}^{\tra}]_{ij}.
  \end{align}
  \end{subequations}
\end{proof}

In the introduction, we saw that the operator Schmidt rank is given by $\osr(X:Y)_{\rho} = \rk(\mcm_{\rho})$ where $\mcm_{\rho} \in \sop{Y}{X}$ is a linear map.
As corollary from \cref{lem:rec-state-as-matrix-composition}, we obtain the monotonicity of the operator Schmidt rank under local maps, i.e.~a data processing inequality:

\begin{cor}
  \label{cor:rec-osr-data-processing}
  Let $\rho$, $\mcn_{X}$, $\mcn_{Y}$ and $\tau = \mcn_{X} \otimes \mcn_{Y}(\rho)$ as in \cref{lem:rec-state-as-matrix-composition}.
  Then,
  \begin{align}
    \label{eq:rec-osr-data-processing}
    \osr(X':Y')_{\tau} \le \osr(X:Y)_{\rho}
  \end{align}
  and
  \begin{align}
    \label{eq:rec-osr-prep-map-rank}
    \osr(X':Y')_{\tau} \le \min\{\rk(\mcn_{X}), \rk(\mcn_{Y})\}.
  \end{align}
\end{cor}

\begin{proof}
  Use the property $\osr(X:Y)_{\rho} = \rk(\mcm_{\rho})$ (\cref{eq:rec-intro-osr-rank}), the identity $\rk(\mcm_{\tau}) = \rk(\mcn_{X} \mcm_{\rho} \mcn_{Y}^{\tra})$ (\cref{lem:rec-state-as-matrix-composition}) and the rank inequality $\rk(AB) \le \min\{\rk(A), \rk(B)\}$ for arbitrary matrices or linear maps $A$ and $B$.
\end{proof}

\section{Petz recovery of bipartite states subjected to local quantum operations
  \label{sec:rec-recov-bipartite-states}}

In the previous section, we considered a linear operator $\rho \in \linop(XY)$ subjected to local linear maps $\mcn_{X} \in \sop{X}{X'}$ and $\mcn_{Y} \in \sop{Y}{Y'}$,
\begin{align}
  \tau = (\mcn_{X} \otimes \mcn_{Y})(\rho).
\end{align}
In \cref{sec:rec-reconstr-bipartite-states-subsec} we presented a condition under which $\rho$ can be reconstructed from $\tau$ via local linear maps.
Here, we discuss the same question for a bipartite quantum state $\rho\equiv \rho_{XY} \in \cD(XY)$ and quantum operations $\mcn_{X}$ and $\mcn_{Y}$.
The answer is obtained by restricting \cref{thm:rec-petz} to the bipartite setting, i.e.~by inserting $\rho = \rho_{XY}$, $\sigma = \rho_{X} \otimes \rho_{Y}$ and $\mcn = \mcn_{X} \otimes \mcn_{Y}$ \parencite{Piani2008}:

\begin{cor}[Bipartite Petz recovery map \parencite{Piani2008}]
  \label{thm:rec-petz-bipartite}
  Let $\rho \in \linop(XY)$ a quantum state and $\mcn_{X} \colon X \to X'$, $\mcn_{Y} \colon Y \to Y'$ quantum operations.
  Set $\tau = (\mcn_{X} \otimes \mcn_{Y})(\rho)$.
  The equality
  \begin{align}
    \label{eq:rec-petz-bipartite-mi}
    I(X':Y')_{\tau} = I(X:Y)_{\rho}
  \end{align}
  holds if and only if there are quantum operations $\cptprec_{X'} \colon X' \to X$ and $\cptprec_{Y'} \colon Y' \to Y$ which satisfy
  \begin{align}
    \label{eq:rec-petz-bipartite-recov}
    \rho &= (\cptprec_{X'} \otimes \cptprec_{Y'})(\tau).
  \end{align}
  If the condition is satisfied, the two Petz recovery maps $\cptprec_{X'} = \petzrec_{\rho_{X},\mcn_{X}}$ and $\cptprec_{Y'} = \petzrec_{\rho_{Y},\mcn_{Y}}$ satisfy the equation.
\end{cor}

In the next section, we explore the relation between bipartite state reconstruction (\cref{thm:rec-reconstr-bipartite}) and bipartite Petz recovery (\cref{thm:rec-petz-bipartite}).

\section{Comparison of Petz recovery and state reconstruction
  \label{sec:rec-comparison}}

In this section, we compare Petz recovery with state reconstruction for a bipartite quantum state $\rho \in \qs(XY)$ subject to local quantum operations $\mcn_{X}\colon X \to X'$ and $\mcn_{Y}\colon Y \to Y'$:
\begin{align}
  \label{eq:rec-comp-tau-from-rho}
  \tau = (\mcn_{X} \otimes \mcn_{Y})(\rho).
\end{align}
The reconstruction is to be achieved via local linear maps:
\begin{align}
  \label{eq:rec-comp-rho-from-tau}
  \rho = (\recmap_{X'} \otimes \recmap_{Y'})(\tau).
\end{align}
State reconstruction and the Petz recovery map both provide maps $\recmap_{X'}$ and $\recmap_{Y'}$ under the assumption of different conditions on $\rho$ and $\tau$ (\cref{thm:rec-reconstr-bipartite,thm:rec-petz-bipartite}).
There is the following evident relation between state reconstruction and Petz recovery:
\begin{thm}
  \label{thm:rec-comparison}
  Let $\rho \in \qs(XY)$ a quantum state and let $\mcn_{X} \colon X \to X'$, $\mcn_{Y} \colon Y \to Y'$ quantum operations.
  Let $\tau = \mcn_{X} \otimes \mcn_{Y}(\rho)$.
  Then
  \begin{align}
    \label{eq:rec-comparison-conditions}
    I(X':Y')_{\tau} = I(X:Y)_{\rho} \quad\Rightarrow\quad \osr(X':Y')_{\tau} = \osr(X:Y)_{\rho}.
  \end{align}
  The converse implication does not hold.
  The premise of the implication \eqref{eq:rec-comparison-conditions} is equivalent to Petz recovery being possible, i.e.~there are CPTP maps $\cptprec_{X'}$ and $\cptprec_{Y'}$ which recover $\rho$ from $\tau$ (\cref{thm:rec-petz-bipartite}):
  \begin{align}
    \label{eq:rec-comparison-recov}
    \rho = \cptprec_{X'} \otimes \cptprec_{Y'}(\tau).
  \end{align}
  CPTP maps which recover $\rho$ in this way are given by the Petz recovery maps $\cptprec_{X'} = \petzrec_{\rho_{X},\mcn_{X}}$ and $\cptprec_{Y'} = \petzrec_{\rho_{Y},\mcn_{Y}}$ (see \cref{thm:rec-petz}).
  The conclusion of the implication \eqref{eq:rec-comparison-conditions} is equivalent to state reconstruction being possible, i.e.~there are linear maps $\recmap_{X'}$ and $\recmap_{Y'}$ which reconstruct $\rho$ from $\tau$ (\cref{thm:rec-reconstr-bipartite}):
  \begin{align}
    \label{eq:rec-comparison-reconstr}
    \rho = \recmap_{X'} \otimes \recmap_{Y'}(\tau).
  \end{align}
  Linear maps which recover $\rho$ in this way are given by the reconstruction maps $\recmap_{X'} = \matrec_{\mcm_{\rho}\mcn_{Y}^{\tra},\mcn_{X}}$ and $\recmap_{Y'} = \matrec_{\mcm_{\rho}^{\tra}\mcn_{X}^{\tra},\mcn_{Y}}$ (see \cref{thm:rec-reconstr-bipartite}).
\end{thm}

\begin{proof}
  The premise of \cref{eq:rec-comparison-conditions} implies that the CPTP maps from \cref{eq:rec-comparison-recov} exist (\cref{thm:rec-petz-bipartite}).
  These CPTP maps are linear maps which satisfy \cref{eq:rec-comparison-reconstr}, which in turn implies that the conclusion of \cref{eq:rec-comparison-conditions} holds (\cref{thm:rec-reconstr-bipartite}).

  A counterexample for the converse implication of \cref{eq:rec-comparison-conditions} will be provided in \cref{sec:rec-comparison-fourpartite}.
\end{proof}

\begin{rem}
  Suppose that the conclusion of \cref{eq:rec-comparison-conditions} holds while its premise does not hold. 
  If both linear maps $\recmap_{X'}$ and $\recmap_{Y'}$ were CPTP, the equality $I(X':Y')_{\rho} \le I(X:Y)_{\tau}$ would be implied by \cref{eq:rec-comparison-reconstr} and $I(X':Y')_{\rho} = I(X:Y)_{\tau}$ would follow (since the converse inequality always holds because of $\tau = (\mcn_{X} \otimes \mcn_{Y})\rho$).
  This would contradict our assumption; i.e.~at least one of $\recmap_{X'}$ and $\recmap_{Y'}$ is not CPTP\@.
  For example, the reconstruction maps for the $W$ state on four qubits (\cref{sec:rec-comparison-fourpartite}) are non-positive.
\end{rem}

\cref{thm:rec-comparison} implies that any state which admits Petz recovery also admits state reconstruction. 
\begin{table}[t]
  \centering
  \begin{tabular}[t]{lll}
    \hline
    Method & $\recmap_{X'}$ depends on & $\recmap_{Y'}$ depends on \\\hline
    Petz recovery & $\rho_{X}$, $\mcn_{X}$ & $\rho_{Y}$, $\mcn_{Y}$ \\
    Reconstruction & $(\id \otimes \mcn_{Y})(\rho)$, $\mcn_{X}$ & $(\mcn_{X} \otimes \id)(\rho)$, $\mcn_{Y}$ \\\hline
  \end{tabular}
  \caption{Necessary input data for a recovery or reconstruction of a quantum state $\rho \in \qs(XY)$ from $\tau = (\mcn_{X} \otimes \mcn_{Y})(\rho)$. In both cases, $\rho$ is obtained from $\rho = (\recmap_{X'} \otimes \recmap_{Y'})(\tau)$.}
  \label{tab:rec-bipartite-input}
\end{table}
\Cref{tab:rec-bipartite-input} compares the quantities on which the state reconstruction maps and the Petz recovery maps depend:
Both methods require knowledge of the quantum operations $\mcn_{X}$ and $\mcn_{Y}$.  
The marginal states $\rho_{X}$ and $\rho_{Y}$ are sufficient for computing the Petz recovery maps.
To compute the state reconstruction maps, the states $(\id \otimes \mcn_{Y})(\rho) \in \qs(XY')$ and $(\mcn_{X} \otimes \id)(\rho) \in \qs(X'Y)$ are needed.
The marginals $\rho_{X}$, $\rho_{Y}$, $\tau_{X'} \equiv \mcn_{X}(\rho_{X})$ and $\tau_{Y'} \equiv \mcn_{Y}(\rho_{Y})$ (which are also needed for the recovery maps) 
can be inferred from the states $(\id \otimes \mcn_{Y})(\rho)$ and $(\mcn_{X} \otimes \id)(\rho)$. However, in addition, these states contain correlations between the systems $X$ and $Y'$ and between $X'$ and $Y$, respectively. 
This means that state reconstruction requires more input data for a reconstruction of $\rho$ than state recovery.
On the other hand, \cref{thm:rec-comparison} and the examples in the following subsection show that state reconstruction is successful for a strictly larger set of states than Petz recovery.

\subsection{Comparison for four-partite systems
  \label{sec:rec-comparison-fourpartite}}

After the comparison of state reconstruction and Petz recovery for bipartite quantum systems, we apply this result to the more specific case of quantum systems which comprise four subsystems.
Specifically, we consider four systems $A$, $B$, $C$, and $D$.
We insert the partial trace $\tr_{A}\colon AB \to B$ for $\mcn_{X} \colon X\to X'$ and the partial trace $\tr_{D}\colon CD \to C$ for $\mcn_{Y} \colon Y \to Y'$ in \cref{thm:rec-comparison}. 
Accordingly, reconstruction of $\rho \in \qs(ABCD)$ from its reduced state $\rho_{BC} = \tr_{AD}(\rho)$ is achieved with maps $\recmap_{B} \in \sop{B}{AB}$ and $\recmap_{C} \in \sop{C}{CD}$ as
\begin{align}
  \label{eq:rec-comp-fourpartite-rho-from-rho-bc}
  \rho &= \recmap_{B} \otimes \recmap_{C}(\rho_{BC}).
\end{align}
A straightforward application of \cref{thm:rec-comparison} provides the implication
\begin{align}
  \label{eq:rec-comp-fourpartite-mi-implies-osr-equality}
  I(B:C)_{\rho} = I(AB:CD)_{\rho} \quad \Rightarrow \quad \osr(B:C)_{\rho} = \osr(AB:CD)_{\rho}.
\end{align}
Since the W state on four qubits satisfies the conclusion of the last Equation but not its premise it constitutes a counterexample for the converse implication.
In addition, it provides a counterexample for the converse implication of \cref{eq:rec-comparison-conditions} in \cref{thm:rec-comparison}. 
On $\ns$ qubits, the W state is given by
\begin{align}
  \label{eq:rec-comp-fourpartite-w-state}
  W_{\ns}
  &=
    \ketbra{w_{\ns}}{w_{\ns}}, \quad
    \ket{w_{\ns}}
    = \frac1{\sqrt \ns} \sqbd{
    \ket{10\ldots0} + \ket{010\ldots0} + \cdots + \ket{0\ldots01} }
\end{align}
and the operator Schmidt rank and mutual information values of the four-qubit W state $W_{ABCD}$ are provided in \cref{tab:inv-recov-comparison} on \cpageref{tab:inv-recov-comparison}.

\begin{table}[t]
  \centering
  \begin{tabular}[t]{lllll}
    \hline
    $X \to X'$ & $Y \to Y'$ & $\mcn_{X}$ & $\mcn_{Y}$ & Condition \\\hline
    $AB \to B$ & $CD \to C$ & $\tr_{A}$ & $\tr_{D}$ & \eqref{eq:rec-comp-cond-osr}, \eqref{eq:rec-comp-cond-mi-A-B-C-D} \\
    $AB \to B$ & $CD \to CD$ & $\tr_{A}$ & $\id$ & \eqref{eq:rec-comp-cond-mi-A-B-CD} \\
    $ABC \to BC$ & $D \to D$ & $\tr_{A}$ & $\id$ & \eqref{eq:rec-comp-cond-mi-A-BC-D} \\
    \hline
  \end{tabular}
  \caption{
    Possible applications of state reconstruction and Petz recovery (\cref{thm:rec-comparison}) to a quadripartite system.
    Corresponding conditions for successful recovery are stated in \cref{lem:rec-comp-fourpartite-implications}.
  }
  \label{tab:rec-comp-fourpartite-details}
\end{table}
Above, we presented one possible application of bipartite Petz recovery (\cref{thm:rec-petz-bipartite}) to a quadripartite system.
It turns out that Petz recovery can be applied to a quadripartite system in three different ways.
The first row of \cref{tab:rec-comp-fourpartite-details} corresponds to the application of state reconstruction and Petz recovery to a quadripartite system as presented above.
Rows two and three of \cref{tab:rec-comp-fourpartite-details} present two different ways to apply Petz recovery to a quadripartite system.
In total, we have one possible application of state reconstruction and three possible applications of Petz recovery and for each application, there is a condition for successful reconstruction/recovery.
These conditions read as follows:
\begin{lem}
  \label{lem:rec-comp-fourpartite-implications}
  Given a quantum state $\rho \in \qs(ABCD)$, consider the equations
  \begin{align}
    \tag{C1}
    \label{eq:rec-comp-cond-osr}
    \osr(B:C)_{\rho} &= \osr(AB:CD)_{\rho} \\
    \tag{C2}
    \label{eq:rec-comp-cond-mi-A-B-C-D}
    I(B:C)_{\rho} &= I(AB:CD)_{\rho} \\
    \tag{C3}
    \label{eq:rec-comp-cond-mi-A-B-CD}
    I(A:B)_{\rho} &= I(A:BCD)_{\rho} \\
    \tag{C4}
    \label{eq:rec-comp-cond-mi-A-BC-D}
    I(A:BC)_{\rho} &= I(A:BCD)_{\rho}
  \end{align}
  The following implications hold, but the converse implications do not:
  \begin{align}
    \text{\eqref{eq:rec-comp-cond-mi-A-B-C-D} } &\Rightarrow \text{ \eqref{eq:rec-comp-cond-osr}}, &
    \text{\eqref{eq:rec-comp-cond-mi-A-B-C-D} } &\Rightarrow \text{ \eqref{eq:rec-comp-cond-mi-A-B-CD}}, &
    \text{\eqref{eq:rec-comp-cond-mi-A-B-CD} } &\Rightarrow \text{ \eqref{eq:rec-comp-cond-mi-A-BC-D}}.
  \end{align}
\end{lem}

\begin{proof}
  The implication \eqref{eq:rec-comp-cond-mi-A-B-C-D} $\Rightarrow$ \eqref{eq:rec-comp-cond-osr} follows from \cref{thm:rec-comparison} with the substitutions given in the first row of \cref{tab:rec-comp-fourpartite-details}.

  \cref{eq:rec-comp-cond-mi-A-B-C-D} $\Rightarrow$ \cref{eq:rec-comp-cond-mi-A-B-CD}:
  The inequality
  $I(B:C)_{\rho} \le I(B:CD)_{\rho} \le I(AB:CD)_{\rho}$ always holds,
  therefore $I(B:C)_{\rho} = I(AB:CD)_{\rho}$ implies
  $I(B:CD)_{\rho} = I(AB:CD)_{\rho}$. The latter can be written with
  the conditional mutual information as $I(A:CD|B) = 0$ (\cref{eq:rec-intro-cmi}). The CMI in
  turn is also equal to $I(A:CD|B) = I(A:BCD) - I(A:B)$, which shows
  the desired equality $I(A:B) = I(A:BCD)$.

  \cref{eq:rec-comp-cond-mi-A-B-CD} $\Rightarrow$ \cref{eq:rec-comp-cond-mi-A-BC-D}:
  The inequality
  $I(A:B)_{\rho} \le I(A:BC)_{\rho} \le I(A:BCD)_{\rho}$ always holds,
  therefore $I(A:B)_{\rho} = I(A:BCD)_{\rho}$ implies
  $I(A:BC)_{\rho} = I(A:BCD)_{\rho}$.

  \cref{tab:inv-recov-comparison} contains states which show that the converse implications do not hold. The states are constructed from the $\ns$-qubit W state from \cref{eq:rec-comp-fourpartite-w-state} and the GHZ and classical GHZ states on $\ns$ qubits:
  \begin{subequations}
  \label{eq:rec-comp-fourpartite-cc-ghz-states}%
  \begin{align}
  \operatorname{GHZ}_{\ns}
  &=
    \ketbra{\text{GHZ}_{\ns}}{\text{GHZ}_{\ns}}, \quad
    \ket{\text{GHZ}_{\ns}}
    = \frac1{\sqrt 2} \sqbd{
    \ket{0\ldots0} + \ket{1\ldots1}
    }
  \\
  \operatorname{cGHZ}_{\ns}
  &=
    \frac12\sqbd{ \ketbra{0\ldots0}{0\ldots0} + \ketbra{1\ldots1}{1\ldots1} }
  \end{align}
  \end{subequations}
  The values of the operator Schmidt rank and mutual information given in \cref{tab:inv-recov-comparison} show that the converse implications do not hold.
\end{proof}

\begin{figure}[t]
  \centering
  \isvgc{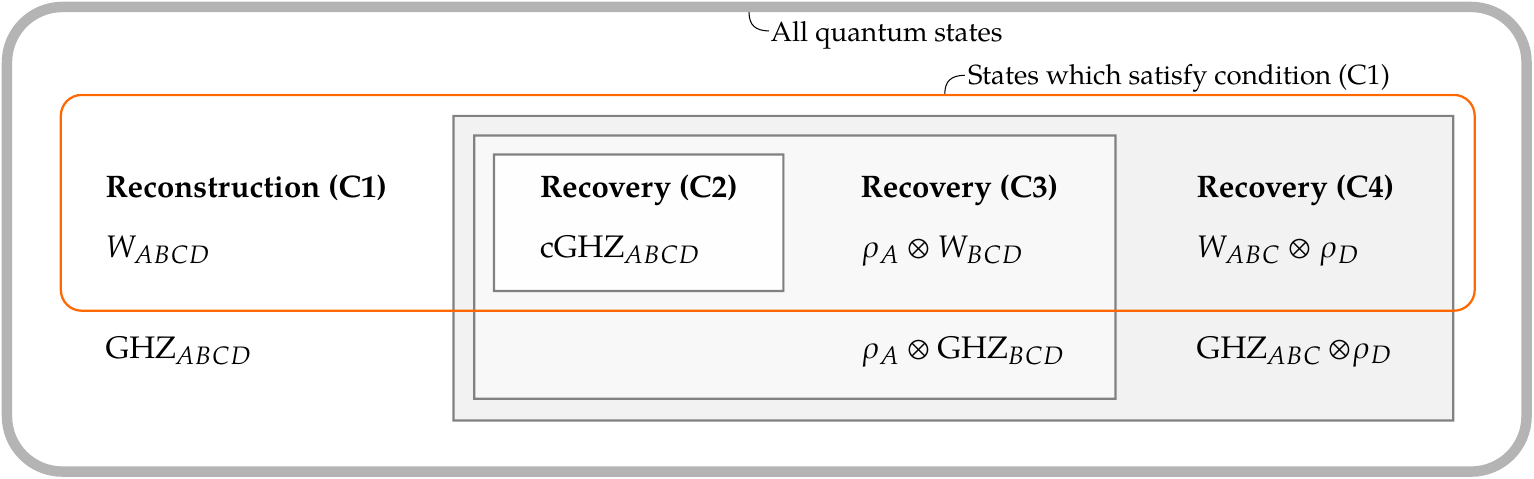}
  \caption{
    State reconstruction vs.~Petz recovery for a quadripartite system.
    The figure shows the relations between conditions \eqref{eq:rec-comp-cond-osr}--\eqref{eq:rec-comp-cond-mi-A-BC-D} from \cref{lem:rec-comp-fourpartite-implications} and several states whose position indicates which of the conditions are satisfied by each state. e.g.\ $W_{ABCD}$ satisfies \eqref{eq:rec-comp-cond-osr} but it does not satisfy \eqref{eq:rec-comp-cond-mi-A-B-C-D}--\eqref{eq:rec-comp-cond-mi-A-BC-D}.
    $\rho_{A}$ and $\rho_{D}$ denote arbitrary states while the W, GHZ and classical GHZ states are defined in \cref{eq:rec-comp-fourpartite-w-state,eq:rec-comp-fourpartite-cc-ghz-states}.
  }
  \label{fig:inv-recov-comparison}
  \vspace*{\floatsep}
  \centering
  \footnotesize
  \begin{tabular}{lllllll}
    \hline
    Method & Cond. & Input & Domain/Range & Depends on \\\hline
    Reconstruction & \eqref{eq:rec-comp-cond-osr} & $\rho_{BC}$ & $B \to AB$, $C \to CD$ & $\rho_{ABC}$, $\rho_{BCD}$ \\
    Recovery & \eqref{eq:rec-comp-cond-mi-A-B-C-D} & $\rho_{BC}$ & $B \to AB$, $C \to CD$ & $\rho_{AB}$, $\rho_{CD}$ \\
    Recovery & \eqref{eq:rec-comp-cond-mi-A-B-CD} & $\rho_{BCD}$ & $B \to AB$ & $\rho_{AB}$ \\
    Recovery & \eqref{eq:rec-comp-cond-mi-A-BC-D} & $\rho_{BCD}$ & $BC \to ABC$ & $\rho_{ABC}$ \\\hline
  \end{tabular}
  \captionof{table}{
    Properties of the four reconstruction/recovery settings considered in \cref{fig:inv-recov-comparison}.
    The last three columns \emph{Input}, \emph{Domain/Range} and \emph{Depends on} refer to the reconstruction map(s).
    See \cref{tab:rec-comp-fourpartite-details} for further details.
  }
  \label{tab:rec-comp-fourpartite}
\end{figure}

The relations between \crefrange{eq:rec-comp-cond-osr}{eq:rec-comp-cond-mi-A-BC-D} from \cref{lem:rec-comp-fourpartite-implications} are illustrated in \cref{fig:inv-recov-comparison}.
The figure also shows which of the conditions are satisfied by the example states from \cref{tab:inv-recov-comparison}.
For example, the W state $W_{ABCD}$ on four qubits does not satisfy \eqref{eq:rec-comp-cond-mi-A-B-C-D}--\eqref{eq:rec-comp-cond-mi-A-BC-D}.
We can understand that $W_{ABCD}$ cannot satisfy \eqref{eq:rec-comp-cond-mi-A-BC-D} by considering the following known result: 
If \eqref{eq:rec-comp-cond-mi-A-BC-D} holds, then \cref{thm:petz-structure} tells us that the reduced state $\rho_{AD}$ must be a separable state.
However, the reduced state $\tr_{BC}(W_{ABCD})$ has a non-positive semidefinite partial transpose and therefore is inseparable, i.e.~entangled \parencite{Peres1996,Horodecki1996}:
The entanglement in the reduced state on $AD$ mandates that \cref{eq:rec-comp-cond-mi-A-BC-D} is not satisfied.

\Cref{fig:inv-recov-comparison} illustrates that reconstruction and the different applications of Petz recovery work for different subsets of all quadripartite states but one should not forget that they also require different reduced states of $\rho$ in order to recover $\rho \in \qs(ABCD)$.
\cref{tab:rec-comp-fourpartite} shows the necessary reduced states for each case.
In all four cases, the full state $\rho \in \qs(ABCD)$ can be reconstructed from marginal states on only two or three of the systems.
Each scheme enables quantum state tomography with incomplete information (i.e.~the necessary marginals) if the corresponding condition is assumed to hold.
Each scheme also relies on the fact that correlations as measured by the operator Schmidt rank or the mutual information are less than maximal; this restriction is imposed by the conditions \eqref{eq:rec-comp-cond-osr}--\eqref{eq:rec-comp-cond-mi-A-BC-D}.

\begin{table}[t]
  \centering%
  \footnotesize%
  \setlength{\tabcolsep}{0.4em}%
  \begin{tabular}{|c|cccc|ccc|cccc|}
    \hline
    & $\kappa(B{:}C)$ & $\kappa(AB{:}CD)$ & $I(B{:}C)$ & $I(AB{:}CD)$
    & $I(A{:}B)$ & $I(A{:}BC)$ & $I(A{:}BCD)$ 
    & C1 & C2 & C3 & C4
    \\ \hline
    $\operatorname{cGHZ}_{ABCD}$ & 1 & 1 & 1 & 1 
    & 1 & 1 & 1
    & $\checkmark$ & $\checkmark$ & $\checkmark$ & $\checkmark$
    \\ \hline
    $\rho_{A} \otimes W_{BCD}$
    & 2 & 2 & $\approx$ 0.92 & $\approx$ 1.84
    & 0 & 0 & 0
    & $\checkmark$ & $-$ & $\checkmark$ & $\checkmark$
    \\
    $W_{ABC} \otimes \rho_{D}$
    & 2 & 2 & $\approx$ 0.92 & $\approx$ 1.84
    & $\approx$ 0.92 & $\approx$ 1.84 & $\approx$ 1.84
    & $\checkmark$ & $-$ & $-$ & $\checkmark$
    \\
    $W_{ABCD}$
    & 2 & 2 & $\approx$ 0.62 & 2
    & $\approx$ 0.62 & 1 & $\approx$ 1.62
    & $\checkmark$ & $-$ & $-$ & $-$
    \\ \hline
    $\rho_{A} \otimes \operatorname{GHZ}_{BCD}$
    & 1 & 2 & 1 & 2
    & 0 & 0 & 0
    & $-$ & $-$ & $\checkmark$ & $\checkmark$
    \\
    $\operatorname{GHZ}_{ABC} \otimes \rho_{D}$
    & 1 & 2 & 1 & 2
    & 1 & 2 & 2
    & $-$ & $-$ & $-$ & $\checkmark$
    \\
    $\operatorname{GHZ}_{ABCD}$
    & 1 & 2 & 1 & 2
    & 1 & 1 & 2
    & $-$ & $-$ & $-$ & $-$
    \\ \hline
  \end{tabular}
  \caption{
    Log-operator Schmidt rank $\kappa(X:Y) = \log_{2}(\osr(X:Y))$ and mutual information\protect\footnotemark{} $I(X:Y)$ of the states shown in \cref{fig:inv-recov-comparison}.
    \Crefrange{eq:rec-comp-cond-osr}{eq:rec-comp-cond-mi-A-BC-D} are defined in \cref{lem:rec-comp-fourpartite-implications}.
  }
  \label{tab:inv-recov-comparison}
\end{table}

\Cref{tab:rec-comp-fourpartite} shows that the state $\rho \in \qs(ABCD)$ can be obtained from $\rho_{ABC}$ and $\rho_{BCD}$ with reconstruction under \eqref{eq:rec-comp-cond-osr} but recovery under \eqref{eq:rec-comp-cond-mi-A-B-C-D} requires only the marginal states on $AB$, $BC$ and $CD$.
This prompts the question whether smaller marginal states are sufficient to reconstruct a state under the reconstruction condition  \eqref{eq:rec-comp-cond-osr}.
For example, one could hope to obtain $\rho$ from $\rho_{AB}$ and $\rho_{BCD}$ but the following two states $\sigma_{+}$ and $\sigma_{-}$ show that this is not possible:
\begin{align}
  \label{eq:rec-comp-4partite-sigma4pm}
  \sigma_{\pm} 
  &= \frac1{32} \parend{ 2 \idm \otimes \idm \otimes \idm + (\idm \pm Z) \otimes Z \otimes Z }
    \otimes \idm,
    \quad
  \idm = \left(\begin{smallmatrix} 1 & 0 \\ 0 &
      1 \end{smallmatrix}\right),
  \quad
  Z = \left(\begin{smallmatrix} 1 & 0 \\ 0 &
      -1 \end{smallmatrix}\right).
\end{align}
\footnotetext{
  The exact values of the numerical constants are $2 \log(3)-\frac43 \approx 1.84$, $4 - \frac32 \log(3) \approx 1.62$, $\log(3)-\frac23 \approx 0.92$ and $3 - \frac32 \log(3) \approx 0.62$.
}
The states $\sigma_{+}$ and $\sigma_{-}$ have the same marginals on $AB$ and $BCD$ but they do satisfy \eqref{eq:rec-comp-cond-osr}.\footnote{
  The reduced states $\sigma_{\pm,AB} = \tr_{CD}(\sigma_{\pm}) = \frac14 \idm \otimes \idm$ and $\sigma_{\pm,BCD} = \tr_{A}(\sigma_{\pm}) = \frac1{16} \parena{ 2 \idm \otimes \idm + Z \otimes Z } \otimes \idm$ do not depend on the sign.
  The values of the operator Schmidt rank are $\osr(AB:CD)_{\sigma_{\pm}} = \osr(B:C)_{\sigma_{\pm}} = 2$.
}
As a consequence, $\rho_{AB}$ and $\rho_{BCD}$ are not sufficient to obtain $\rho$ under \eqref{eq:rec-comp-cond-osr} and it is now apparent that $\rho_{ABC}$ and $\rho_{BCD}$ are \emph{necessary} to reconstruct a state under that condition.%
\footnote{
  This holds true if $\rho$ is to be reconstructed from marginal states of $\rho$ which include at most $\rho_{ABC}$ and $\rho_{BCD}$.
}

\section{Efficient reconstruction of states on spin chains via recursively defined measurements
  \label{sec:rec-spinchain}}

Under suitable conditions, the state of a linear spin chain with $\ns$ spins can be reconstructed from marginal states of few neighbouring spins with the Petz recovery map \parencite{Poulin2011} or with state reconstruction \parencite{Baumgratz2013}.
In \cref{sec:rec-spinchain-from-marginals}, we explore the relation between Petz recovery and state reconstruction in that setting. 
In \cref{sec:rec-spinchain-from-long-range-meas}, we generalize both techniques to use long-range measurements instead of or in addition to short-ranged correlations found in marginal states of few neighbouring spins.
This allows for the efficient recovery/reconstruction of a larger set of states, as is explained in the following.

\paragraph{Motivation for long-range measurements. }

Consider the following quantum states on $\ns$ qubits:
\begin{subequations}
\begin{align}
  \operatorname{GHZ}_{\alpha,\ns}
  &=
    \ketbra{\text{GHZ}_{\alpha,\ns}}{\text{GHZ}_{\alpha,\ns}}, \quad
  \ket{\text{GHZ}_{\alpha,\ns}}
  = \frac1{\sqrt 2} \sqbd{
    \ket{0\ldots0} + \ee^{\ii \alpha} \ket{1\ldots1}
    }
  \\
  \operatorname{cGHZ}_{\ns}
  &=
    \frac12\sqbd{ \ketbra{0\ldots0}{0\ldots0} + \ketbra{1\ldots1}{1\ldots1} }
\end{align}
\end{subequations}
All states from the set $S_{\ns} = \{\operatorname{cGHZ}_{\ns}\} \cup \{\operatorname{GHZ}_{\alpha,\ns} \colon \alpha \in \R\}$ have the same reduced state $\operatorname{cGHZ}_{k}$ on $k < \ns$ qubits.
No recovery or reconstruction method which receives only local reduced states as input can distinguish between the states from the set $S_{\ns}$ and this is also the reason why no method could recover or reconstruct the four-qubit state $\ket{\text{GHZ}_{0,4}} = \ket{\text{GHZ}_{4}}$ in \cref{sec:rec-comparison-fourpartite}. 
Note that the pure states $\ket{\text{GHZ}_{\alpha,\ns}}$ can be represented as an MPS with bond dimension two (because they are the superposition of two pure product states) and that all states from the set $S_{\ns}$ can be represented as an MPO with bond dimension at most four (because they are the sum of at most four tensor product operators) \parencite{Schollwoeck2011}.

We call an MPS representation \emph{efficient} if its bond dimension is at most $D \bigo{\poly(\ns)}$ and a we call a tomography scheme \emph{efficient} if expectation values of at most $\poly(\ns)$ simple observables are needed; a possible definition of a \emph{simple} observable is provided in \cref{rem:rec-efficient}.
Standard quantum state tomography is not efficient because it requires $\sim \exp(\ns)$ expectation values. 
In contrast, it has been shown that any pure state which admits an efficient MPS representation can be determined efficiently from observables with a simple structure.%
\footnote{
  This is shown by the tomography scheme based on unitary operations introduced in Ref.~\parencite{Cramer2010}.
  We discuss it in more detail in \cref{rem:rec-petz-longrange-mps}.
}
The tomography scheme from Ref.~\parencite{Cramer2010} is efficient for the states $\ket{\text{GHZ}_{\alpha,\ns}}$ but recovery/reconstruction methods based on local reduced states must fail for these states. 
In \cref{sec:rec-spinchain-from-long-range-meas}, we extend both Petz recovery and state reconstruction in a way which allows the long-range measurements from Ref.~\parencite{Cramer2010} to be used and thus the states $\ket{\text{GHZ}_{\alpha,\ns}}$ to be reconstructed successfully.
What is more, we show that there are mixed states which cannot be reconstructed from local reduced states but can be reconstructed from long-range measurements (\cref{rem:rec-petz-longrange-addstate}).
This shows that Petz recovery and state reconstruction with long-range measurements can reconstruct more states than prior techniques (recovery/reconstruction from local reduced states and the tomography scheme from Ref.~\parencite{Cramer2010}). 
Furthermore, state reconstruction can reconstruct any MPO of bond dimension $D$ from $\sim \ns D$ expectation values of global tensor product observables, as has been shown in related prior work \parencite{Oseledets2010a}.
We build upon that to show that successful, efficient Petz recovery with long-range measurements implies that efficient state reconstruction with long-range measurements is also possible (\cref{thm:rec-petz-implies-mat-longrange}).

\paragraph{Prior work: MPO reconstruction \parencite{Baumgratz2013}.}

Many physically interesting quantum states $\rho \in \qs(\hilb_{1\dots\ns})$ can be represented efficiently, i.e.~with $\poly(\ns)$ parameters, via an MPO representation \parencite{Baumgratz2013}.
However, standard quantum state tomography requires $\sim \exp(\ns)$ different expectation values in $\rho$ to reconstruct $\rho$, even if $\rho$ admits such an efficient MPO representation.
As an improvement over that, it has been shown%
\footnote{
  Assume that all spins have the same dimension $\ldim = \ldim_{k}$, $k \in \{1, \dots, \ns\}$. 
  Lemma~1 in the supplementary material of \parencite{Baumgratz2013} states that $\rho$ can be reconstructed from reduced states on $l+r+1$ neighbouring spins if a certain condition is satisfied.
  This condition can be satisfied only if both $D^{2} \le \ldim^{2l}$ and $D^{2} \le \ldim^{2r}$ hold.
  However, if these two inequalities are satisfied, the given conditions almost always hold for MPO matrices with random entries. 
}
that almost all states with an MPO representation of bond dimension $D$ can be reconstructed from their reduced states on $\sim \log(D)$ neighbouring spins if a suitable reconstruction scheme is used \parencite{Baumgratz2013}.
We refer to this reconstruction scheme as \emph{MPO reconstruction} and we rederive it below in \cref{thm:rec-mat-marginals} as a consequence of our result on bipartite state reconstruction (\cref{thm:rec-reconstr-bipartite}).

\paragraph{Prior work: Cross approximation of tensor trains \parencite{Oseledets2010a}.}

Our generalization of state reconstruction to long-range measurements in \cref{thm:rec-mat-longrange} can be used to construct an MPO representation of the quantum state (\cref{rem:rec-mat-longrange-eff}).
An MPO representation of a quantum state $\rho \in \qs(\hilb_{1\dots\ns})$ is exactly the same as a tensor train representation of $\rho$ if the operator $\rho$ is regarded as vector from the tensor product vector space $\linop(\hilb_{1}) \otimes \dots \otimes \linop(\hilb_{\ns})$.
Ref.~\parencite{Oseledets2010a} provides a means to reconstruct a tensor of low tensor train rank (i.e.~an MPO of low bond dimension) from few entries.
This procedure is called \emph{tensor train cross approximation}. 
When applied to quantum states, tensor train cross approximation allows for the reconstruction of a quantum state from the expectation values of few tensor product observables.
\Cref{thm:rec-mat-longrange} is more general because it admits more general measurements; e.g.\ it also permits the measurements introduced in Ref.~\parencite{Cramer2010} (cf.\ \cref{rem:rec-petz-longrange-mps,rem:rec-petz-implies-mat-longrange-U-k-matrices}).

\paragraph{Prior work: Markov entropy decomposition \parencite{Poulin2011}.}

The strong subadditivity (SSA) property of the von~Neumann entropy of a tripartite state $\rho \equiv \rho_{ABC}$ (cf.~\eqref{SSA} of \cref{sec:rec-intro-petz-recovery}) can be expressed in terms of the conditional entropy $S(A|B)_{\rho} = S(AB)_{\rho} - S(B)_{\rho}$:
\begin{align}
  \label{eq:rec-vn-cond-reduces-entropy}
  S(A|BC)_\rho &\le S(A|B)_\rho
\end{align}
If we choose arbitrary subsets $\mcm_{k} \subset \{1\dots k\}$, the entropy $S(\rho) = S(1\dots \ns)_{\rho}$ can be rewritten and upper-bounded as follows:
\begin{align}
  S(1\dots \ns)_{\rho}
  &=
    S(12)_{\rho} + \sum_{k=2}^{\ns-1} S(k+1|1\dots k)_{\rho}
  \\
  &\le
    S(12)_{\rho} + \sum_{k=2}^{\ns-1} S(k+1|\mcm_{k})_{\rho} =: S_{M}(\rho).
  \label{eq:rec-vn-upper-bound}
\end{align}
In the second step, we applied \cref{eq:rec-vn-cond-reduces-entropy} $\ns-2$ times.
The sets $\mcm_{k}$ are called \emph{Markov shields} and the upper bound $S_{M}(\rho)$ is called the \emph{Markov entropy} \parencite{Poulin2011}.
In the following, we consider the particular choice $\mcm_{k} = \{k\}$.
In that case, the conditional entropies $S(k+1|\mcm_{k})_{\rho} = S(k+1|k)_{\rho}$ depend only on the reduced state $\rho_{k,k+1}$. 
As a consequence, the Markov entropy $S_{M}(\rho)$ is an upper bound on $S(\rho)$ which depends only on the nearest-neighbour reduced states $\rho_{k,k+1}$ ($k \in \{1, \dots, \ns-1\}$).
For a nearest-neighbour Hamiltonian $H = \sum_{k=1}^{\ns-1} h_{k,k+1}$ the energy $E = \tr(\rho H) = \sum_{k=1}^{\ns-1} \tr(\rho_{k,k+1} h_{k,k+1})$ is determined by the same reduced states $\rho_{k,k+1}$.
Therefore, lower bounds to the free energy $F = E - TS$ of a thermal state at temperature $T$ can be found with a variational algorithm which only uses the reduced states $\rho_{k,k+1}$ \parencite{Poulin2011}. 

\Cref{eq:rec-vn-upper-bound} was obtained by applying \cref{eq:rec-vn-cond-reduces-entropy} $n-2$ times for $A = \{k+1\}$, $B = \{k\}$ and $C = \{1, \dots, k-1\}$ ($k \in \{2, \dots, \ns-1\}$). 
These inequalities are equivalent to the following inequalities (because \cref{eq:rec-vn-cond-reduces-entropy} is equivalent to $I(B:C) \le I(AB:C)$):
\begin{align}
  \label{eq:rec-vn-upper-bound-as-mi}
  I(1\dots k-1:k)_{\rho} &\le I(1\dots k-1:k,k+1)_{\rho}, \quad (k \in \{2, \dots, \ns-1\}).
\end{align}
If equality holds in \cref{eq:rec-vn-upper-bound-as-mi} or, equivalently, in \cref{eq:rec-vn-upper-bound}, the global state $\rho$ can be obtained from the reduced states $\rho_{k,k+1}$ ($k \in \{1, \dots, \ns-1\}$) via Petz recovery maps (see supplementary material and main text of \parencite{Poulin2011}).
We state this known result in \cref{thm:rec-petz-marginals} and show that these conditions imply that MPO reconstruction (as stated in \cref{thm:rec-mat-marginals}) is possible  (\cref{thm:rec-petz-implies-mat-marginals}).

\subsection{Reconstruction of states from local marginal states
  \label{sec:rec-spinchain-from-marginals}}

The state of the spin chain is $\rho \in \qs(\hilb_{1\dots\ns})$ where $\hilb_{1\dots\ns} = \hilb_{1} \otimes \dots \hilb_{\ns}$ is the tensor product of the single-spin Hilbert spaces.
For each $k \in \{1, \dots, \ns\}$, we partition the spins on the chain into two parts:
\begin{align}
  \leftsetdef k
  \quad\text{and}\quad
  \rightsetdef k.
\end{align}
The marginal states $\rho_{\leftsetnb k} = \tr_{\rightsetnb k}(\rho)$, for $k \in \{2, \dots, \ns\}$, can be defined recursively via
\begin{align}
  \rho_{\leftsetnb k} = (\id_{\leftsetnb k} \otimes  \tr_{k+1})(\rho_{\leftsetnb{k+1}})
  \quad\quad (k \in \{2, \dots, \ns-1\}),
  \quad\quad
  \rho_{\leftsetnb{\ns}} = \rho.
  \label{eq:rec-petz-marginals-rho-k}
\end{align}
Each partial trace $\tr_{k+1}$ is a local CPTP map.
If the partial trace $\tr_{k+1}$ does not decrease the mutual information between $\leftsetval{k-1}$ and $\{k, k+1\}$ for all $k \in \{2 \dots \ns-1\}$, then the $\ns$-spin state $\rho$ can be recovered from marginal states of two neighbouring spins \parencite{Poulin2011}:

\begin{thm}
  \label{thm:rec-petz-marginals}
  Let $\rho \in \qs(\hilb_{1\dots \ns})$ be a quantum state.
  If the conditions
  \begin{align}
    \label{eq:rec-petz-marginals-mi}
    I(\leftsetnb{k-1}:k)_{\rho}
    &= I(\leftsetnb{k-1}: k,k+1)_{\rho}
      \quad\quad (k \in \{2, \dots, \ns-1\}),
  \end{align}
  are satisfied, then the marginal states $\rho_{\leftsetnb{k}} = \tr_{\rightsetnb{k}}(\rho)$ are given by
  \begin{align}
    \label{eq:rec-petz-marginals-rec}
    \rho_{\leftsetnb{k+1}} &= (\id_{\leftsetnb{k-1}} \otimes \cptprec_{k})(\rho_{\leftsetnb{k}}),         \quad\quad \rho
                     =
        \rho_{\leftsetnb{\ns}} = \cptprec_{\ns-1}\cptprec_{\ns-2}\dots \cptprec_{3}\cptprec_{2}(\rho_{12}),
     \end{align}
  where $k \in \{2, \dots, \ns-1\}$. The recovery maps $\issop{\recmap_{k}}{\hilb_{k}}{\hilb_{k,k+1}}$ are given by Petz recovery maps $\cptprec_{k} = \petzrec_{\rho_{k,k+1},\tr_{k+1}}$ (\cref{thm:rec-petz}). In the above, $\rho_{ij}$ denotes the reduced state of $\rho$ on sites $i$ and $j$.
\end{thm}

\begin{proof}
  For $k \in \{2, \dots, \ns - 1\}$, apply \cref{thm:rec-petz-bipartite} with $X = X' = \leftsetnb{k-1}$, $\mcn_{X} = \id_{X}$, $Y = \hilb_{k,k+1}$, $Y' = \hilb_{k}$ and $\mcn_{Y}(\cdot) = \tr_{k+1}(\cdot)$.
\end{proof}

In a similar way, if the partial traces do not decrease certain operator Schmidt ranks, their actions can be reverted with state reconstruction \parencite{Baumgratz2013,Oseledets2010a}:

\begin{thm}
  \label{thm:rec-mat-marginals}
  Let $\rho \in \linop(\hilb_{1 \dots \ns})$ be a linear operator.
  If the conditions
  \begin{align}
    \osr({k-1}:{k})_{\rho}
    &= \osr(\leftsetnb{k-1}: {k,k+1})_{\rho}
      \quad\quad (k \in \{3, \dots, \ns-1\})
      \label{eq:rec-mat-marginals-osr}
  \end{align}
  are satisfied, the marginal states $\rho_{\leftsetnb{k}} = \tr_{\rightsetnb{k}}(\rho)$ are given by
\begin{align}
  \label{eq:rec-mat-marginals-rec}
  \rho_{\leftsetnb{k+1}}
  &=
    (\id \otimes \recmap_{k})(\rho_{\leftsetnb{k}}),
  &
    \rho
  &=
    \recmap_{\ns-1} \recmap_{\ns-2} \dots \recmap_{4} \recmap_{3}(\rho_{123})
\end{align}
where $k \in \{3, \dots, \ns-1\}$ and $\linrec_{k} \in \sop{\hilb_{k}}{\hilb_{k,k+1}}$.
The maps are given by $\linrec_{k} = \matrec_{\mathcal L_{k},\tr_{k+1}}$ (\cref{thm:rec-reconstr-bipartite}), $\mathcal L_{k} = \mcm_{\sigma_{k}}^{\tra}$ with $\sigma_{k} = \rho_{k-1,k,k+1}$ and $\mcm_{\sigma_{k}} \in \linmap(\linop(\hilb_{k,k+1}); \linop(\hilb_{k-1}))$ (\cref{eq:rec-intro-state-to-map}).
\end{thm}

\begin{proof}
  For $k \in \{3, \dots, n-1\}$, apply \cref{thm:rec-reconstr-bipartite} with $X = \leftsetnb{k-1}$, $X' = \{k-1\}$, $\mcn_{X}(\cdot) = \tr_{\leftsetnb{k-2}}(\cdot)$, $Y = \hilb_{k,k+1}$, $Y' = \hilb_{k}$, $\mcn_{Y}(\cdot) = \tr_{k+1}(\cdot)$.
  Recall that $\sigma_{k} = \rho_{k-1,k,k+1} = (\tr_{\leftsetnb{k-2}} \otimes \id)(\rho_{\leftsetnb{k+1}})$ implies $\mcm_{\sigma_{k}} = (\tr_{\leftsetnb{k-2}}) \mcm_{\rho_{\leftsetnb{k+1}}}$ where $\mcm_{\sigma_{k}} \in \sop{\hilb_{k,k+1}}{\hilb_{k-1}}$ and
  $\mcm_{\rho_{\leftsetnb{k+1}}} \in \linmap(\linop(\hilb_{k,k+1});$ $\linop(\leftsetnb{k-1}))$
  (\cref{lem:rec-state-as-matrix-composition}).
  Therefore, the reconstruction map is given by $\linrec_{k} = \matrec_{\mathcal L_{k}, \tr_{k+1}}$ with $\mathcal L_{k} = \mcm_{\rho_{\leftsetnb{k+1}}}^{\tra} [\tr_{\leftsetnb{k-2}}]^{\tra} = \mcm_{\sigma_{k}}^{\tra}$. 
\end{proof}

The result from \cref{thm:rec-mat-marginals} has been obtained previously in \parencite{Baumgratz2013} under the name \emph{reconstruction of quantum states} or \emph{MPO reconstruction}.
For a discussion of futher related work \parencite{Oseledets2010a}, see \cref{rem:rec-petz-implies-mat-longrange-prior-work}. 

\begin{rem}[Efficient recovery/reconstruction]
  \label{rem:rec-efficient}

  We call a recovery or reconstruction method to obtain $\rho \in \qs(\hilb_{1\dots\ns})$ \emph{efficient} if it satisfies the following conditions.
  The method provides an efficient representation of $\rho$ (cf.\ \cref{sec:rec-intro-mpo}).
  This representation of $\rho$ can be constructed from suitable input data in at most $\poly(\ns)$ computational time.
  As a consequence, the size of the input data may be at most $\poly(\ns)$ (i.e.~at most $\poly(\ns)$ complex numbers).
  The necessary input data may be obtained from at most $\poly(\ns)$ different tensor product expectation values, i.e.~expectation values of the form $\tr[(A_{1} \otimes \dots \otimes A_{\ns} \otimes A') \mcn(\rho \otimes \rho')]$ where $A_{k} \in \linop(\hilb_{k})$, $A' \in \linop(\hilb_{Y'})$, $\rho' \in \qs(\hilb_{Y'})$ and $Y'$ is an ancilla system of dimension $d_{Y'} \bigo{\poly(n)}$.%
  \footnote{
    We introduce the ancilla system to capture the precise definition of the measurements in \cref{sec:rec-spinchain-from-long-range-meas}.
  }
  The quantum operation $\mcn$ is constructed from at most $\poly(\ns)$ quantum operations whose input and output dimension is at most $\poly(\ns)$.
  This severely restricts the available measurements because the number of two-qubit gates required to implement e.g.\ an arbitrary $n$-qubit unitary is exponential in~$n$~\parencite{Nielsen2007}. 
  
  Standard quantum state tomography is not efficient because it fails to satisfy any of these criteria. For example, in quantum state tomography $\sim \exp(\ns)$ expectation values are required in order to determine $\rho$. 
  
  Clearly, \cref{thm:rec-petz-marginals} and \cref{thm:rec-mat-marginals} satisfy all of these criteria because efficient representations are provided and the necessary input data consists only of two- and three-spin marginals of $\rho$.
  \Cref{lem:rec-spinchain-seq-prep} also provides efficient MPO and PMPS representations for \cref{thm:rec-petz-marginals} and an efficient MPO representation for \cref{thm:rec-mat-marginals}. 
\end{rem}

Note that the operator Schmidt rank condition \eqref{eq:rec-mat-marginals-osr} is different from the mutual information condition \eqref{eq:rec-petz-marginals-mi} in that it contains $\{k-1\}$ instead of $\leftsetnb{k-1}$ on the very left.
If the partial trace $\tr_{\leftsetnb{k-2}}$ which maps $\leftsetnb{k-1}$ onto $\{k-1\}$ was left out, the state $\rho_{\leftsetnb{k+1}}$ would be needed to construct $\linrec_{k}$.
Construction of $\linrec_{\ns-1}$ would need $\rho_{\leftsetnb{\ns}} = \rho$ and the reconstruction would be neither efficient nor useful. 
Despite this difference, we show that the premise of state recovery (\cref{thm:rec-petz-marginals}) implies the premise of state reconstruction (\cref{thm:rec-mat-marginals}):

\begin{thm}
  \label{thm:rec-petz-implies-mat-marginals}

  Let $\rho \in \qs(\hilb_{1\dots \ns})$ a quantum state. The conditions
  \begin{align}
    \label{eq:rec-petz-mat-marginals-mi}
    I(\leftsetnb{k-1}:{k})_{\rho}
    &= I(\leftsetnb{k-1}: {k,k+1})_{\rho}
      \quad\quad (k \in \{2, \dots, \ns-1\}),
  \end{align}
  imply
  \begin{align}
    \osr({k-1}:{k})_{\rho}
    &= \osr(\leftsetnb{k-1}: {k,k+1})_{\rho}
      \quad\quad (k \in \{3, \dots, \ns-1\}).
      \label{eq:rec-petz-mat-marginals-osr}
  \end{align}
  In other words, if the state $\rho$ can be recovered with Petz recovery from the marginals $\rho_{k,k+1}$ ($k \in \{1, \dots, \ns-1\}$), then it can also be reconstructed with state reconstruction from the marginals $\rho_{k-1,k,k+1}$ ($k \in \{2, \dots, \ns-1\}$).
\end{thm}

\begin{proof}
  \Cref{eq:rec-petz-mat-marginals-mi} implies ($k \in \{2, \dots, \ns-1\}$, use \cref{eq:rec-intro-cmi})
  \begin{subequations}
    \label{eq:rec-petz-mat-marginals-proof-cmi}
    \begin{align}
      0
      & = I(\leftsetnb{k-1}:{k,k+1}) - I(\leftsetnb{k-1}:{k})
        = I(\leftsetnb{k-1}:{k+1}|{k})
      \\
      &= I(\leftsetnb{k-1},{k}:{k+1}) - I({k}:{k+1}).
      \\
      &= I(\leftsetnb{k}:{k+1}) - I({k}:{k+1}).
    \end{align}
  \end{subequations}
  We shift the index of the last equation by one and obtain, for $k \in \{3, \dots, \ns-1\}$,
  \begin{align}
    \label{eq:rec-petz-mat-marginals-proof-mi}
    I({k-1}:{k})
    & = I(\leftsetnb{k-1}:{k})
      = I(\leftsetnb{k-1}:{k,k+1}).
  \end{align}
  This mutual information equality implies the corresponding operator Schmidt rank equality (\cref{thm:rec-comparison}). 
\end{proof}

\begin{rem}
  For $\ns = 4$, \cref{thm:rec-petz-implies-mat-marginals} reduces to ``\eqref{eq:rec-comp-cond-mi-A-B-C-D} implies \cref{eq:rec-comp-cond-osr}'' from \cref{lem:rec-comp-fourpartite-implications} if one takes into account that ``$I(B:C) = I(AB:CD)$'' \eqref{eq:rec-comp-cond-mi-A-B-C-D} is equivalent to ``$I(A:B) = I(A:BC)$ and $I(AB:C) = I(AB:CD)$'' (see the following \cref{lem:rec-quadripartite-mi-shuffle-equiv}).
\end{rem}

\begin{lem}
  \label{lem:rec-quadripartite-mi-shuffle-equiv}
  
  $I(B:C) = I(AB:CD)$ holds if and only if $I(A:B) = I(A:BC)$ and $I(AB:C) = I(AB:CD)$.
\end{lem}

\begin{proof}
  ``$\Rightarrow$'':
  Let $I(B:C) = I(AB:CD)$ hold.
  We have $I(B:C) \le I(AB:C) \le I(AB:CD) = I(B:C)$, which implies that $I(B:C) = I(AB:C) = I(AB:CD)$.
  As a consequence,
  \begin{align}
    \label{eq:rec-quadripartite-mi-shuffle-equiv-proof}
    0 = I(AB:C) - I(B:C) = I(A:C|B) = I(A:BC) - I(A:B)
  \end{align}
  also holds (\cref{eq:rec-intro-cmi}).
  This already shows the proposed conclusion.

  ``$\Leftarrow$'':
  Let $I(A:B) = I(A:BC)$ and $I(AB:C) = I(AB:CD)$ hold.
  The former equality implies \cref{eq:rec-quadripartite-mi-shuffle-equiv-proof} and this shows that $I(B:C) = I(AB:C) = I(AB:CD)$ holds.
\end{proof}

\subsection{Long-ranged measurements
  \label{sec:rec-spinchain-from-long-range-meas}}

In this subsection, we generalize recovery and reconstruction to use certain long-range measurements as input and show that successful recovery implies that successful reconstruction is also possible.

\subsubsection{Recovery from long-ranged measurements
  \label{sec:rec-long-range-meas-recov}}

Recovery and reconstruction of a spin chain state from few-body marginals required that correlations (as measured by the mutual information or the operator Schmidt rank) do not decrease under the following partial traces (\cref{fig:rec-longrange-petz-meas}):
\begin{align}
  \rho_{\leftsetnb{k}} = (\id_{X_{k}} \otimes  \tr_{k+1})(\rho_{\leftsetnb{k+1}}) \quad\quad (k \in \{2, \dots, \ns-1\}),
  \quad\quad
  \rho_{\leftsetnb{\ns}} = \rho.
  \label{eq:rec-petz-marginals-rho-k-repeated}
\end{align}
In order to incorporate long-range measurements, we introduce ancillary systems $\rightsetp{k}$ ($k \in \{0, \dots, \ns\}$, $d_{\rightsetp{0}} = d_{\rightsetp{\ns}} = 1$), quantum operations $\issop{\mct_{k}}{\hilbsysand{k}\rightsetp{k}}{\rightsetp{k-1}}$ and define $\rho_{k}  \in \qs(\leftsetnb{k}\sysand\rightsetp{k})$ via (\cref{fig:rec-longrange-petz-meas})
\begin{align}
  \rho_{k-1} = (\id_{X_{k-1}} \otimes  \mct_{k})(\rho_{k}) \quad\quad (k \in \{2, \dots, \ns\}),
  \quad\quad
  \rho_{\ns} = \rho.
  \label{eq:rec-petz-longrange-rho-k-intro}
\end{align}%
The relation between long-range measurements and the maps $\mct_{k}$ is explained \cref{rem:rec-petz-longrange-meas}. 
If the mutual information $I(\leftsetnb{k-1}:{k},\rightsetp{k})_{\rho_{k}}$ does not decrease when $\mct_{k}$ is applied, then \cref{thm:rec-petz-longrange} provides a reconstruction of $\rho$ from the states $\sigma_{k} = \tr_{\leftsetnb{k-1}}(\rho_{k}) \in \qs(\hilbsysand{k} \rightsetp{k})$ (details are specified in the theorem).
Before we state the theorem, we explain that measurements on $\sigma_{k}$ correspond to recursively defined long-range measurements on $\rho$ and we observe that suitable ancilla systems $\rightsetp{k}$ and operations $\mct_{k}$ can be determined for any pure MPS.
\begin{figure}[t]
  \centering
  \isvgc{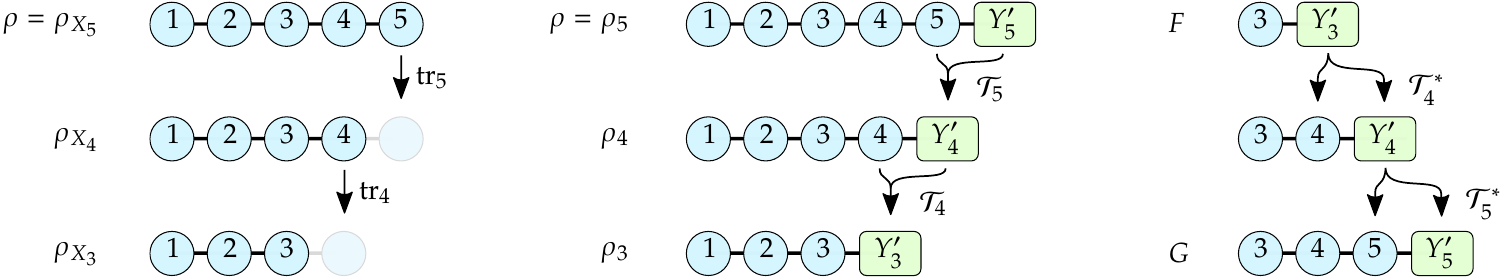}
  \caption{Left: Recursive definition of the reduced states $\rho_{\leftsetnb{k}}$ on $\leftsetdef{k}$ (\cref{eq:rec-petz-marginals-rho-k-repeated}).
    Middle: Spaces on which the recursively defined states $\rho_{k}$ act (\cref{eq:rec-petz-longrange-rho-k-intro,eq:rec-petz-longrange-rho-k}).
    As $\ns = 5$, we have $\rightsetp{5} = \C$.
    Right: Recursive definition of the long-ranged observable $G$ from the local observable $F$ (\cref{rem:rec-petz-longrange-meas}).
}%
  \label{fig:rec-longrange-petz-meas}%
\end{figure}%
\begin{rem}[Recursively defined long-range measurements]
  \label{rem:rec-petz-longrange-meas}
In \cref{thm:rec-petz-longrange} below, the state $\rho$ is reconstructed from the states $\sigma_{k} = \tr_{\leftsetnb{k-1}}(\rho_{k}) \in \qs(\hilbsysand{k}\rightsetp{k})$ ($k \in \{1\dots\ns\}$).
The states $\sigma_{k}$ can be reconstructed from the expectation values $\tr(F_{i}\sigma_{k})$ of a set observables $\{F_{i} \in \linop(\hilbsysand{k}\rightsetp{k})\}_{i}$ which is complete, i.e.\ spans the full vector space $\linop(\hilbsysand{k}\rightsetp{k})$.%
\footnote{
  The observables may be given e.g.\ by the elements of a POVM\@.
}
For simplicity, we drop the index $i$ and denote a possible observable by $F \in \linop(\hilbsysand{k}\rightsetp{k})$. 
The expectation value $\tr(F \sigma_{k})$ corresponds to the following expectation value in $\rho$ (\cref{fig:rec-longrange-petz-meas}):
\begin{subequations}%
\begin{align}%
  \tr(F \sigma_{k})
  &=
    \tr(F (\mct_{k+1}\dots\mct_{\ns})(\rho))
    =
    \tr((\mct_{\ns}^{\adjs}\dots\mct_{k+1}^{\adjs})(F) \rho)
    =
    \tr(G \rho),
  \\
  G
  &= (\mct_{\ns}^{\adjs}\dots\mct_{k+1}^{\adjs})(F) \in \linop(\hilb_{k\dots\ns}).
\end{align}%
\label{eq:rec-petz-longrange-meas}%
\end{subequations}
Here, $\issop{\mct_{k}^{\adjs}}{\rightsetp{k-1}}{\hilbsysand{k}\rightsetp{k}}$ are adjoint superoperators (\cref{sec:rec-notation}) and we used that $F$ is Hermitian, that the channels $\mct_{k}$ map Hermitian operators onto Hermitian operators (because they are completely positive) and that $d_{\rightsetp{\ns}} = 1$.
The observable $F$ describes a measurement on $\sigma_{k}$ and the recursively defined observable $G$ which acts on $\hilb_{k\dots\ns}$ describes a measurement on $\rho$.
\Cref{eq:rec-petz-longrange-meas} hence demonstrates that measurements on $\sigma_{k} \in \qs(\hilbsysand{k}\rightsetp{k})$ correspond to recursively defined long-range measurements on $\rho$.
\end{rem}

\begin{rem}[Example: Pure matrix product states]
  \label{rem:rec-petz-longrange-mps}
  Suppose that we fix $l \ge 1$ and set $\hilb_{\rightsetp{k}} = \hilb_{k+1\dots k+l}$.
  In this case, $\rho_{k} \in \qs(X_{k}\sysand Y'_{k}) = \qs(X_{k+l})$ (cf.\ \cref{eq:rec-petz-longrange-rho-k-intro}). 
  We define $\mct_{k}(\cdot) = \tr_{k+l}(U_{k} \cdot U_{k}^{\adjm})$ where $U_{k} \in \linop(\hilb_{k\dots k+l})$ are unitary operators.
  Suppose further that the unitaries have the property that they transform $\rho_{k}$ into
  \begin{align}
    \label{eq:rec-longrange-petz-unitary-decouple}
    (\idm \otimes U_{k}) \rho_{k} (\idm \otimes U_{k}^{\adjm}) =
    \rho'_{k} \otimes \ketbra{\phi_{k}}{\phi_{k}}
  \end{align}
  where $\rho'_{k} \in \qs(\hilb_{1\dots k+l-1})$ and $\ket{\phi_{k}} \in \hilb_{k+l}$ are states. 
  Then, $\rho_{k-1} = \mct(\rho_{k}) = \rho'_{k}$ and the action of $\mct_{k}$ on $\rho_{k}$ can be reversed with $\mck_{k}(\cdot) = U_{k}^{\adjm} (\cdot \otimes \ketbra{\phi_{k}}{\phi_{k}}) U_{k}$, i.e.\ $\mck_{k}(\rho_{k-1}) = \rho_{k}$.
  As a consequence, the mutual information $I(\leftsetnb{k-1}:{k},\rightsetp{k})_{\rho_{k}}$ does not decrease if $\mct_{k}$ is applied (\cref{thm:rec-petz-bipartite}) and we can apply \cref{thm:rec-petz-longrange} to reconstruct $\rho$.
  If $\rho$ is a pure state which has an MPS representation of bond dimension $D$, then unitaries which satisfy \cref{eq:rec-longrange-petz-unitary-decouple} exist if $l = \lceil \log_{d}(D) \rceil$ where $d = \max_{k} d_{k}$ is the maximal dimension of a single spin \parencite{Cramer2010}.
  In this case, \cref{thm:rec-petz-longrange} provides an efficient reconstruction if $D \bigo{\poly(\ns)}$ (cf.\ \cref{rem:rec-petz-longrange-eff}).
  As a consequence, any state which can be reconstructed with the pure-state reconstruction scheme based on unitary operations from Ref.~\parencite{Cramer2010} can also be reconstructed with \cref{thm:rec-petz-longrange} if the same unitaries are used. 
\end{rem}

\begin{thm}
  \label{thm:rec-petz-longrange}
  Let $\rho \in \qs(\hilb_{1\dots \ns})$ a quantum state. 
  Let $\rightsetp{k}$ ($k \in \{1, \dots, \ns\}$) be ancilla systems with $\dim(\rightsetp{\ns}) = 1$ and choose quantum operations $\issop{\mct_{k}}{\hilbsysand{k} \rightsetp{k}}{\rightsetp{k-1}}$ ($k \in \{2, \dots, \ns\}$).
  Define $\rho_{k} \in \qs(\leftsetnb{k}\sysand\rightsetp{k})$ ($k \in \{1, \dots, \ns\}$) recursively via
  \begin{align}
    \rho_{k-1} = (\id \otimes  \mct_{k})(\rho_{k}) \quad\quad (k \in \{2, \dots, \ns\}),
    \quad\quad
    \rho_{\ns} = \rho.
    \label{eq:rec-petz-longrange-rho-k}
  \end{align}
  If the conditions
  \begin{align}
    \label{eq:rec-petz-longrange-mi}
    I(\leftsetnb{k-1}:\rightsetp{k-1})_{\rho_{k-1}}
    &= I(\leftsetnb{k-1}: {k}, \rightsetp{k})_{\rho_{k}},
      \quad\quad k \in \{2, \dots, \ns\},
  \end{align}
  are satisfied, then
    \begin{align}
      \label{eq:rec-petz-longrange-rec}
      \rho_{k}
      &=
        (\id \otimes \cptprec_{k})(\rho_{k-1})
        \quad\quad (k \in \{2, \dots, \ns\}),
      &
        \rho
      &=
        \rho_{\ns} = \cptprec_{\ns}\cptprec_{\ns-1}\dots \cptprec_{3}\cptprec_{2}(\rho_{1})
    \end{align}
  where the recovery maps $\issop{\recmap_{k}}{\rightsetp{k-1}}{\hilbsysand{k} \rightsetp{k}}$ are given by Petz recovery maps $\cptprec_{k} = \petzrec_{\sigma_{k},\mct_{k}}$ (\cref{thm:rec-petz}) with $\sigma_{k} = \tr_{\leftsetnb{k-1}}(\rho_{k}) \in \qs(\hilbsysand{k} \rightsetp{k})$.
\end{thm}

\begin{proof}
  For $k \in \{2, \dots, \ns\}$, apply \cref{thm:rec-petz-bipartite} with $X = X' = \leftset{k-1}$, $\mcn_{X} = \id$, $Y = \hilb_{k,\rightsetp{k}}$, $Y' = \rightsetp{k-1}$ and $\mcn_{Y} = \mct_{k}$.
\end{proof}

\begin{rem}
  \Cref{thm:rec-petz-longrange} provides \cref{thm:rec-petz-marginals} by restricting to the special case $\hilb_{\rightsetp{k}} = \hilb_{k+1}$ ($k \in \{1, \dots \ns-1\}$), $\mct_{k}(\cdot) = \tr_{k+1}(\cdot)$ ($k \in \{2, \dots, \ns-1\}$), $\mct_{\ns} = \id$ and using \cref{eq:rec-petz-longrange-rec} only for $k \in \{2, \dots, \ns-1\}$.
\end{rem}

\begin{rem}
  \label{rem:rec-petz-longrange-eff}

  Denote by $d_{Y'} = \max_{k} \dim(\rightsetp{k})$ the maximal dimension of any ancillary system.
  If $d_{Y'} \bigo{\poly(\ns)}$, the recovery scheme from \cref{thm:rec-petz-longrange} is efficient (it satisfies all conditions from \cref{rem:rec-efficient}). 
  \cref{lem:rec-spinchain-seq-prep} provides efficient PMPS and MPO representations of $\rho$.   
\end{rem}

\subsubsection{Reconstruction from long-ranged measurements
  \label{sec:rec-long-range-meas-reconst}}

State reconstruction can be generalized similarly but it requires that additional ancillary systems $\leftsetp{k}$ and linear maps $\mcu_{k}$ are introduced (\cref{fig:rec-mat-longrange-spaces}):

\begin{thm}
  \label{thm:rec-mat-longrange}
  Let $\rho \in \linop(\hilb_{1 \dots \ns})$ be a linear operator.
  Let $\rightsetp{k}$ ($k \in \{0, \dots, \ns\}$) and $\leftsetp{k}$ ($k \in \{0, \dots, \ns-1\}$) ancilla systems with $\dim(\rightsetp{0}) = \dim(\rightsetp{\ns}) = \dim(\leftsetp{0}) = 1$.
  Choose linear maps $\mct_{k} \in \sop{\hilbsysand{k}\rightsetp{k}}{\rightsetp{k-1}}$ ($k \in \{2, \dots, \ns\}$) and  $\mcu_{k} \in \sop{\leftsetnb{k-1}}{\leftsetp{k-1}}$ ($k \in \{1, \dots, \ns\}$, $\mcu_{1} = 1$). As before (\cref{eq:rec-petz-longrange-rho-k}), we define $\rho_{k} \in \linop(\leftsetnb{k}\sysand\rightsetp{k})$ via
  \begin{align}
    \rho_{k-1} = (\id \otimes  \mct_{k})(\rho_{k}) \quad\quad (k \in \{2, \dots, \ns\}),
    \quad\quad
    \rho_{\ns} = \rho.
    \label{eq:rec-mat-longrange-rho-k}
  \end{align}
  In addition, we define (\cref{fig:rec-mat-longrange-spaces})
  \begin{alignat}{3}
    \label{eq:rec-mat-longrange-sigma-k}
    \sigma_{k}
    \;=\;& (\mcu_{k} \otimes \id)(\rho_{k}) &\quad\quad\in\;& \linop(\leftsetp{k-1}\andhilbsysand{k}\rightsetp{k})
      \quad\quad&& (k \in \{1, \dots, \ns),
    \\
    \label{eq:rec-mat-longrange-tau-k}
    \tau_{k}
    \;=\;& (\id \otimes \mct_{k})(\sigma_{k}) &\quad\quad\in\;& \linop(\leftsetp{k-1}\sysand\rightsetp{k-1})
      \quad\quad&& (k \in \{2, \dots, \ns\}).
  \end{alignat}
  If the conditions
  \begin{align}
    \osr(\leftsetp{k-1}:\rightsetp{k-1})_{\tau_{k}}
    &= \osr(\leftsetnb{k-1}: {k}, \rightsetp{k})_{\rho_{k}}
      \quad\quad (k \in \{2, \dots, \ns\})
      \label{eq:rec-mat-longrange-osr}
  \end{align}
  are satisfied, there are linear maps $\linrec_{k} \in \sop{\rightsetp{k-1}}{\hilbsysand{k}\rightsetp{k}}$ ($k \in \{2, \dots, \ns\}$) such that
\begin{align}
  \label{eq:rec-mat-longrange-rec}
  \rho_{k}
  &=
    (\id \otimes \recmap_{k})(\rho_{k-1})
    \quad\quad (k \in \{2, \dots, \ns\}),
  &
    \rho
  &=
    \recmap_{\ns} \recmap_{\ns-1} \dots \recmap_{3} \recmap_{2}(\rho_{1}).
\end{align}
The maps are given by $\linrec_{k} = \matrec_{\mathcal L_{k},\mct_{k}}$ (\cref{thm:rec-reconstr-bipartite}) where $\mathcal L_{k} = \mcm_{\sigma_{k}}^{\tra}$ and $\mcm_{\sigma_{k}} \in \linmap(\linop(\hilbsysand{k}\rightsetp{k}); \linop(\leftsetp{k-1}))$ (\cref{eq:rec-intro-state-to-map}).
\end{thm}

\begin{proof}
  For $k \in \{2, \dots, \ns\}$, apply \cref{thm:rec-reconstr-bipartite} with $X = \leftset{k-1}$, $X' = \leftsetp{k-1}$, $\mcn_{X} = \mcu_{k}$, $Y = \hilb_{k, \rightsetp{k}}$, $Y' = \rightsetp{k-1}$ and $\mcn_{Y} = \mct_{k}$.
  The equality $\rho = (\id \otimes \linrec_{Y'})(\id \otimes \mcn_{Y})(\rho)$ from the \namecref{thm:rec-reconstr-bipartite} becomes $\rho_{k} = (\id \otimes \linrec_{k})(\id \otimes \mct_{k})(\rho_{k})$.
  Recall that \cref{eq:rec-mat-longrange-sigma-k} implies $\mcm_{\sigma_{k}} = \mcu_{k} \mcm_{\rho_{k}}$ where $\mcm_{\sigma_{k}} \in \sop{\hilbsysand{k}\rightsetp{k}}{\leftsetp{k-1}}$ and $\mcm_{\rho_{k}} \in \sop{\hilbsysand{k}\rightsetp{k}}{\leftsetnb{k-1}}$ (\cref{lem:rec-state-as-matrix-composition}).
  Therefore, the reconstruction map is given by $\linrec_{k} = \matrec_{\mathcal L_{k}, \mct_{k}}$ with $\mathcal L_{k} = \mcm_{\rho_{k}}^{\tra} \mcu_{k}^{\tra} = \mcm_{\sigma_{k}}^{\tra}$.
\end{proof}

\begin{figure}[t]
  \centering
  \isvgc{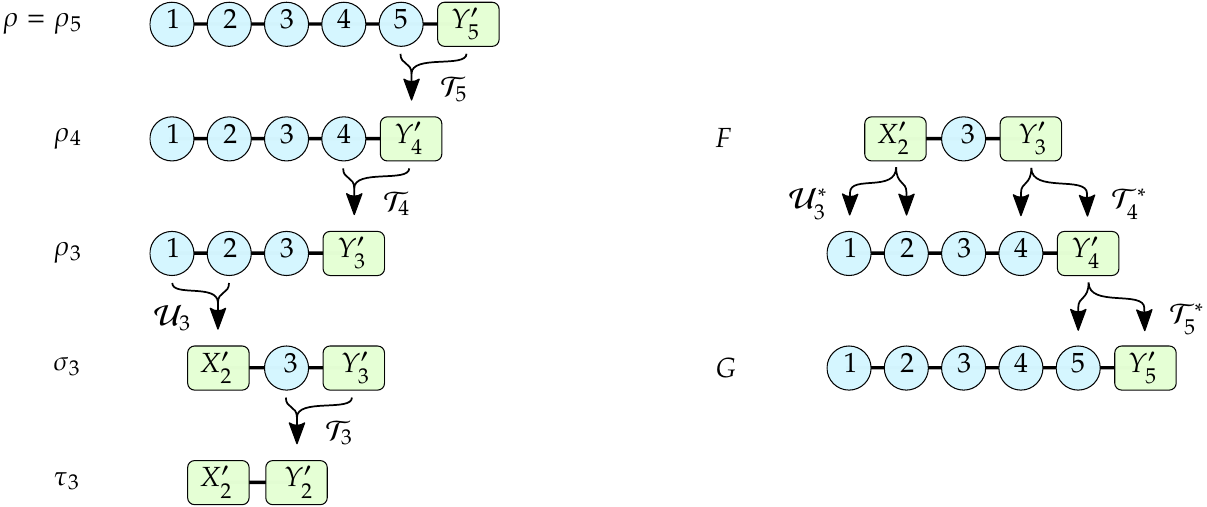}
  \caption{
    Left: Spaces on which the operators $\rho_{k}$, $\sigma_{k}$ and $\tau_{k}$ act (\cref{thm:rec-mat-longrange}). As $\ns = 5$, we have $\rightsetp{5} = \C$.
    Right: Spaces on which the operators $F$ and $G$ act (\cref{rem:rec-mat-longrange-meas}).
  }
  \label{fig:rec-mat-longrange-spaces}
\end{figure}

\begin{rem}
  \Cref{thm:rec-mat-longrange} provides \cref{thm:rec-mat-marginals} by restricting to the special case $\rightsetp{k} = \hilb_{k+1}$, $\leftsetp{k} = \hilb_{k}$, $\mct_{k} = \tr_{k+1}$, $\mcu_{k} = \tr_{\leftsetnb{k-2}}$ and using \cref{eq:rec-mat-longrange-rec} only for $k \in \{3, \dots, \ns-1\}$.%
  \;
\end{rem}

\begin{rem}
  \label{rem:rec-mat-longrange-eff}
  
  Denote by $d_{Y'} = \max_{k} \dim(\rightsetp{k})$ and $d_{X'} = \max_{k} \dim(\leftsetp{k})$ the maximal dimensions.
  If $d_{Y'} \bigo{\poly(\ns)}$ and $d_{X'} \bigo{\poly(\ns)}$, the reconstruction scheme from \cref{thm:rec-mat-longrange} is efficient (it satisfies all conditions from \cref{rem:rec-efficient}).
  \Cref{lem:rec-spinchain-seq-prep} provides an efficient MPO representation of $\rho$.

  Efficient reconstruction implies that a given state can be reconstructed from a number of expectation values which grows polynomially instead of exponentially with $n$.
  This improvement can only be achieved if the to-be-reconstructed state is not a completely general quantum state of $n$ systems.
  In the following, we show that the condition for efficient reconstruction in particular implies that the operator Schmidt ranks of the state are restricted to growing polynomially (instead of exponentially) with $n$.

  For $k = \ceil{\frac n2}$, the maximal value of the operator Schmidt rank $\osr(\leftsetnb{k}:\rightsetnb{k})_{\rho}$ is $(\min\{\ldim_{1}\dots\ldim_{k}, \ldim_{k+1}\dots\ldim_{\ns}\})^{2} \bigo{\exp(\ns)}$ and it is assumed e.g.\ for maximally entangled pure states.
  Suppose that $\rho$ can be reconstructed efficiently. 
  The equality $\rho = (\id_{\leftsetnb{k}} \otimes \linrec_{\ns}\dots\linrec_{k+1})(\rho_{k})$ (\cref{eq:rec-mat-longrange-rec}) implies $\osr(\leftsetnb{k}:\rightsetnb{k})_{\rho} \le \rk(\linrec_{k+1})$ (\cref{cor:rec-osr-data-processing}).
  The rank of $\linrec_{k+1}$ is, in turn, upper bounded by $\rk(\linrec_{k+1}) \le d_{Y'}^{2} \bigo{\poly(\ns)}$.
  I.e.~the operator Schmidt rank of $\rho$ is at most $\osr(\leftsetnb{k}:\rightsetnb{k})_{\rho} \bigo{\poly(\ns)}$.
  In conclusion, any state which admits an \emph{efficient} reconstruction with \cref{thm:rec-mat-longrange} has a \emph{small} operator Schmidt rank in the sense that it does not grows exponentially but only polynomially with the number of spins $\ns$.
\end{rem}

\begin{rem}[Recursively defined long-range measurements]
  \label{rem:rec-mat-longrange-meas}
In \cref{thm:rec-mat-longrange}, $\rho$ is reconstructed from $\sigma_{k} \in \linop(\leftsetp{k-1}\andhilbsysand{k}\rightsetp{k})$ ($k \in \{1\dots\ns\}$, noting that $\sigma_{1} = \rho_{1}$).
As above (\cref{rem:rec-petz-longrange-meas}), measurements on $\sigma_{k}$ correspond to recursively defined long-range measurements on $\rho$ (\cref{fig:rec-mat-longrange-spaces}):%
\begin{subequations}%
\begin{align}%
  \tr(F \sigma_{k})
  &=
    \tr(F (\mcu_{k} \otimes (\mct_{\ns}\dots\mct_{k+1}))(\rho))
    =
    \tr(G \rho),
  &
    F
  &\in \linop(\leftsetp{k-1}\andhilbsysand{k}\rightsetp{k}),
  \\
  G
  &=
    \sqbc{[\mcu_{k}^{\adjs} \otimes (\mct_{k+1}^{\adjs}\dots\mct_{\ns}^{\adjs})](F^{\adjm})}^{\adjm}
    \in \linop(\hilb_{1\dots\ns}).
\end{align}
\label{eq:rec-mat-longrange-meas}%
\end{subequations}
If the superoperators involved are quantum operations and the operator $F$ is Hermitian, $G$ is Hermitian as well and there is a correspondence between observables on $\sigma_{k}$ and recursively defined long-ranged observables on $\rho$.
Otherwise, the correspondence holds only in an abstract sense between operators $F \in \linop(\leftsetp{k-1}\andhilbsysand{k}\rightsetp{k})$ and $G\in \linop(\hilb_{1\dots\ns})$.
\end{rem}

\begin{rem}[Mixed state which requires long-range measurements]
  \label{rem:rec-petz-longrange-addstate}

  \Cref{rem:rec-petz-longrange-mps} showed
  \\
  that any pure MPS can be recovered with \cref{thm:rec-petz-longrange} if the unitary operations from \parencite{Cramer2010} are used. Below, we show that recovery with \cref{thm:rec-petz-longrange} implies that reconstruction with \cref{thm:rec-mat-longrange} is also possible (see \cref{thm:rec-petz-implies-mat-longrange}).   
  The following simple mixed state shows that \cref{thm:rec-petz-longrange,thm:rec-mat-longrange} can recover more than pure matrix product states and more than recovery or reconstruction from local reduced states (\cref{thm:rec-petz-marginals,thm:rec-mat-marginals}): The state
  \begin{align}
    \rho &= \frac12\sqbd{ \ketbra{0\ldots0}{0\ldots0} + \ketbra{100\dots01}{100\dots01}}
           \in \qs(\hilb_{1\dots\ns})
    \label{eq:rec-petz-addstate}
  \end{align}
  does not admit recovery or reconstruction from local reduced states because it turns into a product state if the first or last site is traced out; this unavoidably reduces the mutual information from non-zero to zero and the operator Schmidt rank from larger than one to one.
  The state admits recovery or reconstruction via \cref{thm:rec-petz-longrange} and \cref{thm:rec-mat-longrange} if the following definitions are used: 
  Assuming uniform local dimensions $\ldim = \ldim_{k}$ ($k \in \{1, \dots, \ns\}$), set $\hilb_{\rightsetp{k}} = \hilb_{k+1}$, $\mct_{k} = \tr_{k+1} \operatorname{SWAP}_{k,k+1}$ ($k \in \{2, \dots, \ns-1\}$), $\mct_{\ns} = \idm_{\ns}$, $\leftsetp{k-1} = \hilb_{k-1}$, $\mcu_{k} = \tr_{1,\dots,k-2} \operatorname{SWAP}_{1,k-1}$ ($k \in \{2,\dots,\ns\}$) and $\mcu_{1} = 1$.%
  \footnote{
    The swap gate is given by $\operatorname{SWAP}_{kl}(\cdot) = S_{kl} \cdot S_{kl}$, $S_{kl} = \sum_{i=1}^{d_{k}} \sum_{j=1}^{d_{l}} \ketbra{ij}{ji}$ where $\ket i$ and $\ket j$ are orthonormal bases of $\hilb_{k}$ and $\hilb_{l}$.
  }
  With these definitions, the states $\rho_{k}$ used in \cref{thm:rec-petz-longrange} and \cref{thm:rec-mat-longrange} are given by $\rho_{n} = \rho$,
  \begin{subequations}%
    \begin{align}%
      \rho_{k} &= \frac12\sqbd{ \ketbra{0\ldots0}{0\ldots0} + \ketbra{100\dots01}{100\dots01}}
                 \in \qs(\hilb_{1\dots k+1})
    \end{align}
    where $k \in \{1\dots n-1\}$.
    The states $\sigma_{k}$ and $\tau_{k}$ used in \cref{thm:rec-mat-longrange} are given by $\sigma_{1} = \rho_{1}$, $\sigma_{n} = \rho_{1} \in \qs(\hilb_{n-1,n})$,
    \begin{align}
      \sigma_{k} &= \frac12\sqbd{ \ketbra{000}{000} + \ketbra{101}{101} } \in \qs(\hilb_{k-1,k,k+1})
                   (k \in \{2\dots n-1\}),
      \\
      \tau_{k} &= \frac12\sqbd{ \ketbra{00}{00} + \ketbra{11}{11} } \in \qs(\hilb_{k-1,k})
                 (k \in \{2\dots n\}).
    \end{align}%
    \label{eq:rec-petz-addstate-steps}%
  \end{subequations}
  These states show that the conditions from \cref{thm:rec-petz-longrange} and \cref{thm:rec-mat-longrange} are satisfied.
\end{rem}

\subsubsection{Recovery vs. reconstruction for long-ranged measurements
  \label{sec:rec-long-range-meas-vs}}

In this section, we show that the conditions for state recovery (\cref{thm:rec-petz-longrange}) imply that state reconstruction (\cref{thm:rec-mat-longrange}) is also possible. 
The premise of \cref{thm:rec-petz-longrange} implies the premise of \cref{thm:rec-mat-longrange} for $\mcu_{k} = \id$.
However, \cref{thm:rec-mat-longrange} does not provide a useful reconstruction with $\mcu_{k} = \id$ because the necessary input $\sigma_{\ns}$ for the construction of $\linrec_{\ns}$ would be $\sigma_{\ns} = \rho$.
In \cref{thm:rec-petz-implies-mat-marginals} we used the symmetry of the conditional mutual information to work around this but this is no longer possible because $\mct_{k}$ was introduced. 
Note that \cref{eq:rec-petz-longrange-mi} implies the same equality for operator Schmidt ranks and that \cref{eq:rec-petz-longrange-rec} provides MPO representations of the $\rho_{k}$ (\cref{lem:rec-spinchain-seq-prep}).
It is well-known that maps $\mcu_{k}$ suitable for \cref{thm:rec-mat-longrange} can be obtained directly from the matrices of the MPO representation after the matrices have been transformed into a suitable orthogonal (mixed-canonical) form (\parencite{Schollwoeck2011,Oseledets2011}; see also \cref{rem:mpr-orthocore-unfolding-svd} in the appendix).
The maps $\mcu_{k}$ obtained in this way are given by partial isometries on the vector space of linear operators.
Such a map is not guaranteed to be completely positive or trace preserving, i.e.~it does not represent a quantum operation and it may not allow for an efficient implementation in a given quantum experiment.
An alternative construction has been put forward in \parencite{Oseledets2010a}:%
\footnote{
  We provide a formal description of the corresponding part of their work in \cref{lem:mpr-interface-inv,lem:rec-permsub}.
}
Here, maps $\mcu_{k}$ are provided whose matrix representation is given by a submatrix of a permutation matrix in a product basis of $\linop(\leftsetnb{k}) = \linop(\hilb_{1}) \otimes \dots \otimes \linop(\hilb_{k})$.
We use this result to prove that efficient recovery implies efficient reconstruction in \cref{thm:rec-petz-implies-mat-longrange}.
\Cref{rem:rec-petz-implies-mat-longrange-U-k-matrices} discusses advantages and disadvantages of the two different choices for $\mcu_{k}$ mentioned in this paragraph.

\begin{thm}
  \label{thm:rec-petz-implies-mat-longrange}

  Let the premise of \cref{thm:rec-petz-longrange} hold.
  Set $\dim(\rightsetp{0}) = 1$ and $\leftsetp{k} = \rightsetp{k}$.
  There are linear maps $\mcu_{k} \in \sop{\leftsetnb{k-1}}{\leftsetp{k-1}}$ ($k \in \{2, \dots, \ns\}$) such that
  \begin{align}
    \osr(\leftsetp{k-1}:\rightsetp{k-1})_{\tau_{k}}
    &= \osr(\leftsetnb{k-1}: {k}, \rightsetp{k})_{\rho_{k}}
      \quad\quad (k \in \{2, \dots, \ns\})
      \label{eq:rec-petz-implies-mat-longrange-osr}
  \end{align}
  holds where $\tau_{k} = (\mcu_{k} \otimes \id_{\rightsetp{k-1}})(\rho_{k-1})$ is the same operator as in \eqref{eq:rec-mat-longrange-tau-k}.
  Choose operator bases $F^{(k)}_{i_{k}}$, $F^{(X'_{k})}_{j_{k}}$ and $F^{(Y'_{k})}_{l_{k}}$.
  There is an efficient algorithm to choose suitable maps $\mcu_{k}$ and they can be chosen such that their matrix representation (\cref{eq:rec-intro-components}) is a submatrix of a permutation matrix.
  In this case, the resulting reconstruction is efficient (in the sense of \cref{rem:rec-efficient}) if recovery is efficient and if the operator bases are chosen such that they contain only Hermitian operators. 
\end{thm}

\begin{proof}
  \Cref{lem:rec-spinchain-seq-prep} provides an MPO representation of the states $\rho_{k}$ from \cref{eq:rec-petz-longrange-rec}.
  It is well-known that maps $\mcu_{k}$ can be chosen recursively such that $\osr(X'_{k-1}:Y'_{k-1})_{\tau_{k}} = \osr(X_{k-1}:Y'_{k-1})_{\rho_{k-1}}$ holds if an MPO representation of $\rho_{k-1}$ is given \parencite{Schollwoeck2011,Oseledets2011}.
  As $\osr(X_{k-1}:Y'_{k-1})_{\rho_{k-1}} = \osr(X_{k-1}:k, Y'_{k})_{\rho_{k}}$ is implied by \eqref{eq:rec-petz-longrange-mi} (\cref{thm:rec-comparison}), it is clear that \cref{eq:rec-petz-implies-mat-longrange-osr} holds as well. 
  It was also recognized that the maps $\mcu_{k}$ can be chosen such that their matrix representation is a submatrix of a permutation matrix \parencite{Oseledets2010a}.
  We provide a self-contained description of the corresponding procedure in \cref{lem:mpr-interface-inv}.

  Let $U_{k}$ the matrix representation of $\mcu_{k}$ and suppose that $U_{k}$ is a submatrix of a permutation matrix.
  Denote by $f_{k}(j) = \{(i_{1}, \dots, i_{k-1}) \colon [U_{k}]_{j,(i_{1} \dots i_{k-1})} = 1\}$ the set of columns with a non-zero entry in a given row of $U_{k}$ (where $i_{j} \in \{1 \dots d_{j}^{2}\}$ and $j \in \{1 \dots d_{X'_{k-1}}^{2}\}$).
  The matrix elements of $\sigma_{k} = (\mcu_{k} \otimes \id)(\rho_{k})$ are given by (insert an identity map \eqref{eq:rec-intro-op-identity} into \cref{eq:rec-mat-longrange-sigma-k})
  \begin{align}
    \label{eq:rec-petz-implies-mat-longrange-sigma-k-comp}
    & \skpc{ F^{(\leftsetp{k-1})}_{j} \otimes F^{(k)}_{i_{k}} \otimes F^{(\rightsetp{k})}_{l}}{ \sigma_{k} }
    \\
    & = \sum_{i_{1}\dots i_{k-1}}
      [U_{k}]_{j,i_{1}\dots i_{k-1}} \skpc{F^{(1)}_{i_{1}} \otimes \dots \otimes F^{(k-1)}_{i_{k-1}} \otimes F^{(k)}_{i_{k}} \otimes F^{(\rightsetp{k})}_{l}}{\rho_{k}}
      \nonumber \\
    & =
      \begin{cases}
        0, & \text{if } f_{k}(j) = \emptyset, \\
        \skpc{F^{(1)}_{i_{1}} \otimes \dots \otimes F^{(k)}_{i_{k}} \otimes F^{(\rightsetp{k})}_{l}}{\rho_{k}}, & \text{with } \{(i_{1}, \dots, i_{k-1})\} = f_{k}(j).
      \end{cases}
             \label{eq:rec-petz-implies-mat-longrange-obs}
  \end{align}
  Here, we used that $\abs{f_{k}(j)} \le 1$ because $U_{k}$ is a submatrix of a permutation matrix. 
  The last equation shows that $\sigma_{k}$, which needs to be known for reconstruction of $\rho$, can be determined from at most $(d_{X'_{k}} d_{k} d_{Y'_{k}})^{2}$ tensor product expectation values in $\rho_{k}$.
  The structure of the given expectation values is permitted for efficient reconstruction (\cref{rem:rec-efficient})) and the number of expectation values is at most $\poly(n)$ if recovery is efficient.
  Furthermore, efficient recovery implies that the MPO representation of the $\rho_{k}$ as well as the procedures to determine $\mcu_{k}$ and $f_{k}(j)$ are efficient as well (\parencite{Oseledets2010a}; for details see \cref{lem:mpr-interface-inv,lem:rec-permsub}).
  This finishes the proof of the theorem. 
\end{proof}

\begin{rem}
  \label{rem:rec-petz-implies-mat-longrange-U-k-matrices}
  
  The singular values of $\mcm_{\sigma_{k}} \in\sop{k\sysand Y'_{k}}{X'_{k-1}}$ equal those of $\mcm_{\rho_{k}} \in\sop{k\sysand Y'_{k}}{X_{k-1}}$ if the maps $\mcu_{k}$ are suitable partial isometries on the vector space of linear operators (cf.\ \cref{rem:mpr-orthocore-unfolding-svd}).
  For reconstruction stability (\cref{thm:rec-matrixdec-stability}), this is the optimal case (if the maps $\mct_{k}$ are predefined).
  If the maps $\mcu_{k}$ are restricted to submatrices of permutation matrices, the singular values of $\mcm_{\sigma_{k}}$ are smaller than or equal to those of $\mcm_{\rho_{k}}$ (because $\mcu_{k}$ has unit operator norm).
  If the smallest non-zero singular value decreases, then stability of the reconstruction is reduced (\cref{thm:rec-matrixdec-stability}; cf.\ \parencite{Oseledets2010a,Savostyanov2014}).
  In the worst case, the smallest non-zero singular value decreases by a factor exponential in $\ns$ because of the recursive construction of the $\mcu_{k}$ \parencite{Savostyanov2014}.
  However, empirical results show that this worst-case behaviour is usually not observed in practice \parencite{Baumgratz2013,Oseledets2010a,Savostyanov2014,Savostyanov2011}.

  If the maps $\mct_{k}$ are not predefined, the singular values of $\mcm_{\sigma_{k}}$ equal those of
  $\mcm_{\rho} \in \linmap(\linop(\hilb_{k\dots\ns});$ $\linop(\leftsetnb{k-1}))$
  if the maps $\mcu_{k}$ and $\mct_{k}$ are suitable partial isometries on the vector space of linear operators (\cref{rem:mpr-orthocore-unfolding-svd}).
  In this case, \cref{thm:rec-petz-implies-mat-longrange} allows reconstruction of an arbitrary MPO (or matrix product state/tensor train) with optimal reconstruction stability.
  However, it remains an open question whether this can be fully exploited e.g.\ in the reconstruction of quantum states as the necessary measurements may not allow for an efficient implementation if the maps $\mcu_{k}$ and $\mct_{k}$ are general partial isometries on the vector space of linear operators.
  
  The situation is different if the state $\rho$ is a pure matrix product state.
  Here, partial isometries which act on the Hilbert spaces themselves can be obtained (\parencite{Cramer2010}, cf.\ \cref{rem:rec-petz-longrange-mps,rem:mpr-orthocore-unfolding-svd}).
  These partial isometries can be implemented via unitary control of the quantum system and they have the property that they preserve the singular values of $\mcm_{\rho}$.
  This also shows that the tomography scheme for pure matrix product states based on local unitary operations and proposed in \parencite{Cramer2010} provides maps $\mct_{k}$ and $\mcu_{k}$ for state recovery and reconstruction with optimal stability. 
\end{rem}

\begin{rem}[Related work]
  \label{rem:rec-petz-implies-mat-longrange-prior-work}

  Note that nowhere in the proof of \cref{thm:rec-mat-longrange,thm:rec-reconstr-bipartite} did we use the fact that $\rho$ is a linear operator on $\hilb_{1\dots \ns}$.
  The theorems apply equally well to arbitrary vectors $\ket{\psi} \in \hilb_{1\dots \ns}$ on tensor product vector spaces $\hilb_{1\dots\ns} = \hilb_{1} \otimes \dots \otimes \hilb_{\ns}$.
  The components of $\braket{i_{1}\dots i_{\ns}}{\psi}$ of $\ket\psi$ in a product basis define a tensor $t \in \C^{\ldim_{1}\times \dots \times \ldim_{\ns}}$ (i.e.~an array with $\ns$ indices).
  
  A result similar to \cref{thm:rec-mat-longrange} has been obtained before in the context of tensor train representations \parencite{Oseledets2010a,Savostyanov2011,Savostyanov2014}.
  Their result is formulated for a tensor with $\ns$ indices, i.e.~replace $\linop(\hilb_{k})$ by $\hilb_{k}$, $\linop(\rightsetp{k})$ by $\rightsetp{k}$, $\osr(\leftsetp{k-1}:\rightsetp{k-1})_{\tau_{k}}$ by $\rk(\mcm_{\tau_{k}})$, etc.
  They restrict to $\dim(\leftsetp{k-1}) = \dim(\rightsetp{k-1}) = \rk(\mcm_{\tau_{k}})$.
  In this case, the pseudoinverse in the reconstruction maps $\recmap_{k}$ (defined in \cref{thm:rec-mat-longrange}) is just the regular inverse (cf.\ \cref{rem:rec-matrixdec}).
  They also restrict $\mcu_{k}$ and $\mct_{k}$ to submatrices of permutation matrices in a fixed product basis.
  In addition, they provide an algorithm which attempts to determine suitable maps $\mcu_{k}$ and $\mct_{k}$ incrementally and efficiently.
  Similar work has been carried out for the Tucker and hierarchical Tucker tensor representations \parencite{Caiafa2010,Caiafa2015,Ballani2013} and the relation between this work and the matrix reconstruction from \cref{sec:rec-reconstr-matrices} will be explored elsewhere \parencite{HolzaepfelPhD}. 
\end{rem}

\section*{Acknowledgements}

We acknowledge discussions with Oliver Marty. 
Work in Ulm was supported by an Alexander von Humboldt Professorship, the ERC Synergy grant BioQ, the EU projects QUCHIP and EQUAM, the US Army Research Office Grant No.~W91-1NF-14-1-0133. 
Work in Hannover was supported by the DFG through SFB~1227 (DQ-mat) and the RTG~1991, the ERC grants QFTCMPS and SIQS, and the cluster of excellence EXC201 Quantum Engineering and Space-Time Research.

\appendix

\section{Appendix
  \label{sec:rec-appendix}}

\subsection{Optimality of the stability bound
  \label{sec:rec-appendix-stability-optimal}}

The following examples show that the bound from \cref{thm:rec-matrixdec-stability} is optimal up to constants and that the reconstruction error can diverge as $\epsilon$ approaches zero if small singular values in $LMR$ are not truncated.

\Cref{thm:rec-matrixdec-stability} provides an upper bound on the reconstruction error of a reconstructible matrix $S$ (the signal) which is perturbed by some error matrix $E$.
The following example shows that the
upper bound from the theorem is optimal up to constants:
\begin{align*}
    L &= 
  \begin{pmatrix}
    1 & 0 & 0 \\ 0 & 1 & 0
  \end{pmatrix},
  &
  R &=
  \begin{pmatrix}
    1 & 0 \\ 0 & 1 \\ 0 & 0
  \end{pmatrix},
  & 
  S
  &=
    \begin{pmatrix}
      0 & \Delta & 1 \\
      \Delta & 0 & 0 \\
      1 & 0 & 0
    \end{pmatrix},
  &
  E
  &=
  \eta\epsilon
  \begin{pmatrix}
    0 & 0 & 0 \\ 0 & 1 & 0 \\ 0 & 0 & 0
  \end{pmatrix}.
\end{align*}
The eigenvalues of $S$ are $\pm \sqrt{1 + \Delta^{2}}$ so $\eta = \norm{S} = \sqrt{1 + \Delta^{2}}$ approaches unity as $\Delta \to 0$.
For simplicity, we might assume that using $\eta \approx 1$ is sufficient for the discussion of this example but we keep $\eta > 1$.
We have $\norm{E} / \eta = \epsilon$. 
Suppose that we choose $\epsilon$ and $\tau$ such that $0 < \epsilon \le \tau < \frac1{3\sqrt2}$.
Set $c = \tau + 2\epsilon$, i.e.\ $0 < c < \frac1{\sqrt2}$, and set $\Delta = c / \sqrt{1-c^{2}}$.
Then $\Delta < 1$ and $\eta < \sqrt 2$.
In addition, we obtain
\begin{align*}
  (\tau + 2\epsilon)\eta = c \sqrt{1 + \Delta^{2}} = c \sqrt{1 - \frac{c^{2}}{1-c^{2}}}
  = \Delta.
\end{align*}
The eigenvalues of $LSR$ are $\pm \Delta$.
This provides us $\gamma = \Delta / \eta = \tau + 2\epsilon$.
Therefore, the condition $\epsilon \le \tau < \gamma - \epsilon$ is automatically satisfied and, as a consequence, $2\epsilon < \gamma$ holds as well.
$LER$ can change the eigenvalues of
$LSR$ at most by $\epsilon$ (cf.\ proof of \cref{lem:rec-pi-perturbation-using-wedin}), so no truncation
occurs. In this case, the reconstruction error has exactly the scaling
from the theorem:
\begin{align*}
  \norm{ MR (LMR)_{\tau}^{\pinv} LM - S}
  =
  \norma{
  \begin{pmatrix}
    0 & 0 & 0 \\ 0 & \eta\epsilon & 0 \\ 0 & 0 & - \frac{\eta\epsilon}{\Delta^{2}}
  \end{pmatrix}
  }
  &=
    \max\cura{ \eta\epsilon, \frac{\eta\epsilon}{\Delta^{2}} }
    \ge
    \frac{\eta\epsilon}{\Delta^{2}}
    \ge
    \frac{\epsilon}{9\sqrt 2(\gamma-\epsilon)^{2}}
\end{align*}
Note that the conditions from above imply $\Delta^{2}/\eta = \eta\gamma^{2} < \sqrt 2\gamma^{2}$ and that $\gamma = \tau + 2\epsilon \ge 3\epsilon$.
The latter implies $-\epsilon \ge -\gamma/3$ and $3(\gamma - \epsilon) \ge 2\gamma \ge \gamma$.
Combining the relations provides the bound $\Delta^{2}/\eta < \sqrt 2 \gamma^{2} \le 9\sqrt 2(\gamma-\epsilon)^{2}$ used above. 

One may ask whether thresholds $\tau$ outside the interval permitted by the theorem reconstruct $M$ successfully. 
In this example, a threshold which is large enough to produce a different reconstruction will replace at least one of the two singular values of the reconstruction by zero.
As the two singular values of $S$ are equal, the reconstruction error will be at least $\norm{S}$ in this case.
i.e.\ larger thresholds do not provide a successful reconstruction in the sense that the error in operator norm is significantly smaller than $\norm{S}$.
In this example, neither smaller nor larger thresholds (than the ones permitted by \cref{thm:rec-matrixdec-stability}) provide an improved reconstruction: 
Smaller thresholds do not change the reconstruction at all because thresholding does not reduce the rank of $LMR$ in this example.
However, the following example shows that thresholding is in general necessary to obtain an error which satisfies the bound from \cref{thm:rec-matrixdec-stability}. 
We keep $L$ and $R$ from above and
choose
\begin{align*}
  S
  &=
    \begin{pmatrix}
      1 & 0 & 0 \\
      0 & 0 & 0 \\
      0 & 0 & 0
    \end{pmatrix},
  \quad
  E
  =
  \epsilon
  \begin{pmatrix}
    0 & 0 & 0 \\ 0 & \epsilon^{2} & 1-\epsilon^{2} \\ 0 & 1-\epsilon^{2} & \epsilon^{2}
  \end{pmatrix}.
\end{align*}
We have $\eta = \norm{S} = 1$, $\gamma = 1$ and the eigenvalues of $E/\epsilon$ are $1$ and $-1+2\epsilon^{2}$ such that $\norm{E} = \epsilon$; we choose
$\epsilon < \frac12$ such that choosing a $\tau$ from
$\epsilon \le \tau < 1 - \epsilon$ is permitted. The eigenvalues of $LMR$ are $1$ and $\epsilon^{3}$. We obtain (using $\epsilon \le \frac12$)
\begin{align*}
  \norm{MR (LMR)^{\pinv} LM - S}
  &=
  \norma{
  \begin{pmatrix}
    0 & 0 & 0 \\ 0 & \epsilon^{3} & \epsilon(1-\epsilon^{2}) \\ 0 & \epsilon(1-\epsilon^{2}) & \frac{(1-\epsilon^{2})^{2}}{\epsilon}
  \end{pmatrix}
  }
  \ge
  \frac{(1 - \epsilon)^{2}}{\epsilon}
  \ge
  \frac{1 - 2 \epsilon^{2}}{\epsilon}
  \ge 
  \frac{1}{2\epsilon}.
\end{align*}
Without truncating small singular values, the error diverges as $\epsilon \to 0$, i.e.\ it does not satisfy the bound from \cref{thm:rec-matrixdec-stability}.
Here, the effect of $E$ is completely erased by truncation:
\begin{align*}
    (LMR)_{\tau}^{\pinv}
  &=
    \begin{pmatrix}
      1 & 0 \\ 0 & 0
    \end{pmatrix},
  &
  \norm{MR (LMR)_{\tau}^{\pinv} LM - S}
  =
    0.
\end{align*}

\subsection{The stability bound for matrices with non-unit operator norm
  \label{sec:rec-appendix-normalization}}

In this section, we provide an argument which extends the proof of \cref{thm:rec-matrixdec-stability} from matrices $S$ with unit operator norm to matrices $S$ with arbitrary operator norm. 
Suppose that the matrix $M$ is the sum of a signal $S$ and an noise contribution $E$, $M = S + E$.
The signal satisfies $\rk(S) = \rk(LSR)$, but we only know the strength $\norm{E}$ of the noise.
Suppose that for $\norm S = \norm L = \norm R = 1$, we obtain some error bound of the form
\begin{align}
  \norm{
  MR (LMR)_{\tau}^{\pinv} LM - S
  }
  \le f(\epsilon, \gamma, \tau), \quad
  \epsilon = \norm{E}, \quad
  \gamma = \sigma_\text{min}(LMR).
\end{align}
We can obtain an error bound for $M' = S' + E'$ where $S'$, $L'$, and $R'$ have arbitrary norms as follows: Set $M = M' / \norm{S'}$, $S = S / \norm{S'}$, $E = E' / \norm{E'}$ $L = L' / \norm{L'}$, $R = R' / \norm{R'}$.
With these definitions, we have
\begin{align}
  \norm{
  MR (LMR)_{\tau}^{\pinv} LM - S
  }
  =
  \frac{
  \norm{
  M'R' (L'M'R')_{\tau'}^{\pinv} L'M' - S'
  }
  }{\norm{S'}},
\end{align}
where $\tau' = \norm{L} \norm{R} \norm{S} \tau$. Therefore, the bound from the last but one equation implies
\begin{align}
  &
  \norm{
  M'R' (L'M'R')_{\tau'}^{\pinv} L'M' - S'
  }
  \le
  \norm{S'}
  f(\epsilon, \gamma, \tau),
  \\ 
  &
  \quad\quad
    \epsilon = \frac{\norm{E'}}{\norm{S'}},
  \quad
  \gamma = \frac{\sigma_\text{min}(L'M'R')}{\norm{L'}\norm{R'}\norm{S'}}, 
  \quad
  \tau = \frac{\tau'}{\norm{L'}\norm{R'}\norm{S'}}.
\end{align}

In proofs, we assume $\norm{S} = \norm{L} = \norm{R} = 1$ and we use $\epsilon = \norm{E}$.

\subsection{Alternative proof of the stability bound
  \label{sec:rec-appendix-caiafa-based-bound}}

In this section we obtain a bound similar to the one from \cref{thm:rec-matrixdec-stability} using the ansatz by \textcite{Caiafa2015}.

As above, we use $M = S + E$, $\norm{S} = \norm{L} = \norm{R} = 1$ and $\rk(S) = \rk(LSR)$.

Note that $\rk(S) = \rk(LSR)$ implies $\rk(S) = \rk(LS) = \rk(SR) = \rk(LSR)$. We use the matrix reconstruction \cref{prop:rec-matrixdec} several times, sometimes with $L$ or $R$ replaced by the identity matrix. The proposition e.g.\ provides $S = S (LS)^{\pinv} LS$. Using that idenity, we obtain the following two equalities:
\begin{align}
  S &= S (LS)^{\pinv} L (M-E) \\
  M &= S + E =
      S (LS)^{\pinv} LM + (\idm - S (LS)^{\pinv} L) E
      \label{eq:m-decomp-caiafa-left}
\end{align}
In the same way, we obtain:
\begin{align}
  S &= (M - E) R (SR)^{\pinv} S \\
  M &= S + E =
      MR (SR)^{\pinv} S + E (\idm - R (SR)^{\pinv} S)
      \label{eq:m-decomp-caiafa-right}
\end{align}
We will also use
\begin{align}
  LMR (LMR)_{\tau}^{\pinv} LMR = (LMR)_{\tau}
  = [(LMR)_{\tau} - LMR] + LSR + LER.
  \label{eq:lmr-tau-parts}
\end{align}
We decompose into three parts:
\begin{align}
  MR (LMR)_{\tau}^{\pinv} LM = \text{(A)} + \text{(B)} + \text{(C)}
\end{align}
We insert \cref{eq:m-decomp-caiafa-left}) for $M$ at the beginning of the expression and \cref{eq:m-decomp-caiafa-right}) for $M$ on the end of the expression. In the following equations, spaces separate factors which come from different equations.  In part (A) below, we insert \cref{eq:lmr-tau-parts}.
\begin{subequations}
\begin{align}
  \text{(A)}
  &=\phantom{\pinv}
    S (LS)^{\pinv} LM \;\; R (LMR)_{\tau}^{\pinv} L \;\; MR (SR)^{\pinv} S
    \nonumber
    \\
  &=\phantom{\pinv}
    S (LS)^{\pinv} \;\; [(LMR)_{\tau} - LMR] \;\; (SR)^{\pinv} S
    \nonumber
  \\
  &\phantom{=}+
    S (LS)^{\pinv} \;\; LSR \;\; (SR)^{\pinv} S
    \nonumber
  \\
  &\phantom{=}+
    S (LS)^{\pinv} \;\; LER \;\; (SR)^{\pinv} S
    \label{eq:m-decomp-caiafa-partA-S}
  \\[2mm]
  \text{(B)}
  &=\phantom{\pinv}
    S (LS)^{\pinv} LM \;\; R (LMR)_{\tau}^{\pinv} L \;\; E (\idm - R (SR)^{\pinv} S)
    \nonumber
  \\
  &\phantom{=}+
    (\idm - S (LS)^{\pinv} L) E \;\; R (LMR)_{\tau}^{\pinv} L \;\; MR (SR)^{\pinv} S
  \\[2mm]
  \text{(C)}
  &=\phantom{\pinv}
    (\idm - S (LS)^{\pinv} L) E \;\; R (LMR)_{\tau}^{\pinv} L \;\; E (\idm - R (SR)^{\pinv} S)  
\end{align}
\end{subequations}
The expression in \cref{eq:m-decomp-caiafa-partA-S} is equal to $S$. We use the relation $\norm{A A_{\tau}^{\pinv}} = \norm{A_{\tau} A_{\tau}^{\pinv}} \le 1$ and obtain the following bound:
\begin{align}
  & \norm{MR (LMR)_{\tau}^{\pinv} LM - S}
    \nonumber
  \\
  &\le\phantom{\pinv}
    \norm{(LS)^{\pinv}} \norm{(SR)^{\pinv}} 
    \parend{ \norm{(LMR)_{\tau} - LMR} + \epsilon }
    \nonumber
  \\&\phantom{\le}+
      \norm{(LS)^{\pinv}} \norm{\idm - R(SR)^{\pinv}S} \epsilon 
    \nonumber
  \\&\phantom{\le}+
      \norm{(SR)^{\pinv}} \norm{\idm - S (LS)^{\pinv} L} \epsilon
    \nonumber
  \\&\phantom{\le}+
      \norm{\idm - S (LS)^{\pinv} L}\norm{\idm - R(SR)^{\pinv}S}
      \norm{(LMR)_{\tau}^{\pinv}} \epsilon^{2}
\end{align}
This bound has been given by \textcite{Caiafa2015} for the case that $L$ and $R$ have exactly $r = \rk(S)$ rows and columns (such that the matrix $LSR$ is invertible).  They proceed by defining constants $a$, $b$ and $c$ which are independent of the threshold $\tau$ and of noise strength $\epsilon = \norm{E}$ and obtain a bound of the form $a \tau + b \epsilon + c \epsilon^{2} / \tau$.

We continue by analyzing how all terms in the last equation depend on $L$, $R$ and $S$. This will provide a bound similar to that of \cref{thm:rec-matrixdec-stability}.

Because $LSR$, $LS$ and $SR$ have all rank $r = \rk(S)$, the relation $\sigma_\text{min}(LSR) = \sigma_{r}(LSR)$ holds for these three matrices. 
We obtain
\begin{align}
  \gamma = \sigma_\text{min}(LSR) = \sigma_{r}(LSR)
  \le
  \sigma_{r}(LS) \sigma_{1}(R)
  \le
  \sigma_{r}(LS)
  =
  \sigma_\text{min}(LS),
\end{align}
where the first inequality is provided by Ref.~\parencite{Horn1991a} (Theorem~3.3.16, page~178). This provides
\begin{align}
  \norm{S (LS)^{\pinv} L} \le
  \norm{(LS)^{\pinv}} = \frac{1}{\sigma_\text{min}(LS)} \le \frac{1}{\gamma}
\end{align}
and the same bound applies to $\norm{R(SR)^{\pinv}S}$. Note that $\gamma \le \norm{S} = 1$. Further, for any square matrix $A$, we have
\begin{align}
  \norm{\idm - A} \le \max\{ \sigma_{1}(\idm), \sigma_{1}(A) - 1 \}
  = \max\{ 1, \norm{A} - 1\}.
\end{align}
Using $1 \le 1/\gamma$, we obtain
\begin{align}
  \norm{\idm - S (LS)^{\pinv} L} \le \frac1\gamma, \quad
  \norm{\idm - R (SR)^{\pinv} S} \le \frac1\gamma
\end{align}
and
\begin{align}
  \norm{MR (LMR)_{\tau}^{\pinv} LM - S} 
  \le
  \frac{
  \norm{(LMR)_{\tau} - LMR} + 3\epsilon + \norm{(LMR)_{\tau}^{\pinv}} \epsilon^{2}
  }{\gamma^{2}}.
\end{align}
The inequality holds for arbitrary values of $\gamma$, $\tau$ and $\epsilon$. 

Now, we assume $\epsilon \le \tau < \gamma - \epsilon$ and use bounds from \cref{lem:rec-pi-perturbation-using-wedin}. This provides
\begin{align}
  \norm{MR (LMR)_{\tau}^{\pinv} LM - S} 
  \le
  \frac{4\epsilon}{\gamma^{2}}
  +
  \frac{\epsilon^{2}}{\gamma^{2}(\gamma - \epsilon)}
  \le
  \frac{5\epsilon}{\gamma^{2}}.
  \label{eq:bound-from-caiafa-ansatz}
\end{align}

Without the assumption $\epsilon \le \tau < \gamma - \epsilon$, we obtain
\begin{align}
  \norm{MR (LMR)_{\tau}^{\pinv} LM - S} 
  \le
  \frac{\tau + 3\epsilon}{\gamma^{2}}
  +
  \frac{\epsilon^{2}}{\gamma^{2}\tau}
\end{align}
This is again the bound $a \tau + b \epsilon + c \epsilon^{2} / \tau$, but with $a = 1/\gamma^{2}$, $b = 3/\gamma^{2}$ and $c = 1/\gamma^{2}$.
If we select $\epsilon = \tau$ in this bound, we obtain exactly \cref{eq:bound-from-caiafa-ansatz}.
As $(LMR)_{\tau}^{\pinv}$ is the same operator for all $\tau \in [\epsilon, \gamma - \epsilon)$, the bound holds not only for $\epsilon = \tau$ but for all $\tau \in [\epsilon, \gamma - \epsilon)$: We obtain \cref{eq:bound-from-caiafa-ansatz} again.

\subsection{Known results on matrix product representations
  \label{sec:rec-appendix-tensor-unprepare}}

This section reviews known results on matrix product state/tensor train representations used in \cref{sec:rec-long-range-meas-vs}.
It also provides full formal details for the results which were used. 

Given a tensor $t \in \C^{\ldim_{1} \times \dots \times \ldim_{\ns}}$, a matrix product representation of the tensor is given by 
\begin{align}
  \label{eq:mpr-matrix-product-repr}
  [t]_{i_{1},\dots,i_{\ns}}
  &= G_{1}(i_{1}) \dots G_{m}(i_{m}) H_{m} G_{m+1}(i_{m+1}) G_{\ns}(i_{\ns}) \quad (i_{k} \in \{1, \dots, d_{k}\})
\end{align}
where $D_{0} = D_{\ns} = 1$, $H_{m} \in \C^{D_{m} \times D_{m}}$, $G_{k}(i_{k}) \in \C^{D_{k-1} \times D_{k}}$, $i_{k} \in \{1 \dots d_{k}\}$ and $m \in \{1 \dots \ns-1\}$.
For simplicity, $H_{m} = \idm_{D_{m}}$ may be used.
The $G_{k}$ are called the cores of the representation while the matrices $G_{k}(i_{k})$ give the representation its name.
The left and right unfoldings of the cores%
\footnote{
  This notation is partially inspired by \parencite{Kressner2014} but notation in the tensor train literature does not seem to be uniform. 
}
are given by
\begin{align}
  \label{eq:mpr-unf-k}
  G_{k}^{\lunf} &\in \C^{D_{k-1} d_{k} \times D_{k}}
  & G_{k}^{\runf} &\in \C^{D_{k-1} \times d_{k} D_{k}}
\end{align}
and they have the same entries as $G_{k}$, e.g.\ $[G_{k}^{\lunf}]_{b_{k-1}i_{k},b_{k}} = [G_{k}(i_{k})]_{b_{k-1},b_{k}}$.
The left and right interface matrices are given by
\begin{alignat}{2}
  \label{eq:mpr-left-interface}
  G_{\le k} &= (G_{\le k-1} \otimes \idm_{d_{k}}) G_{k}^{\lunf} \in \C^{\ldim_{1}\dots \ldim_{k} \times D_{k}}
  \quad\quad
  & (k &\in \{1 \dots m\}),
  \\
  \label{eq:mpr-right-interface}
  G_{> k} &= (\idm_{d_{k+1}} \otimes G_{>k+1}) (G_{k+1}^{\runf})^{\tra} \in \C^{\ldim_{k+1}\dots \ldim_{\ns} \times D_{k}}
  \quad\quad
  & (k &\in \{m \dots \ns-1\})
\end{alignat}
where $G_{\le 0} = 1$ and $G_{> \ns} = 1$. 
The unfolding $t_{k}$ is the $\ldim_{1} \dots \ldim_{k} \times \ldim_{k+1} \dots \ldim_{\ns}$ matrix with the same entries as $t$ and it can be written as 
\begin{align}
  \label{eq:mpr-unfolding-with-interface}
  t_{m} &= G_{\le m} H_{m} (G_{>m})^{\tra}
\end{align}

It is well-known that a singular value decomposition of the unfolding $t_{k}$ can be obtained efficiently \parencite{Schollwoeck2011,Oseledets2011}:

\begin{rem}
  \label{rem:mpr-orthocore-unfolding-svd}
  
  Fix $m \in \{1, \dots, \ns-1\}$, let $t$ have a matrix product representation as in \cref{eq:mpr-matrix-product-repr} with positive-semidefinite diagonal $H_{m}$ and assume orthogonal cores,%
  \footnote{
    See Ref.~\parencite{Oseledets2011}; this is called a mixed-canonical representation in Ref.~\parencite{Schollwoeck2011}.
  }
  i.e.
  \begin{subequations}
  \begin{align}
    \label{eq:mpr-orthocore-def}
    (G_{k}^{\lunf})^{\adjm} G_{k}^{\lunf} &= \idm_{D_{k}} \quad (k \le m)
    & (G_{k}^{\runf})^{\adjm} G_{k}^{\runf} &= \idm_{D_{k}} \quad (k > m) \\
    \label{eq:mpr-orthocore-int}
    (G_{\le k})^{\adjm} G_{\le k} &= \idm_{D_{k}} \quad (k \le m)
    & (G_{> k})^{\adjm} G_{> k} &= \idm_{D_{k}} \quad (k \ge m)
    \end{align}
  \end{subequations}
  An arbitrary matrix product representation can be efficiently converted into such an orthogonal representation \parencite{Oseledets2011,Schollwoeck2011}. 
  Then
  \begin{align}
    \label{eq:mpr-orthocore-unfolding-svd}
    t_{m} &= G_{\le m} H_{m} (G_{> m})^{\tra}
  \end{align}
  is a singular value decomposition of the unfolding matrix $t_{m}$.
  Let $U_{\le m} = (G_{\le m})^{\adjm}$ and $V_{>m} = (G_{>m})^{\adjm}$. Then
  \begin{align}
    \label{eq:mpr-orthocore-unfolding-svals}
    U_{\le m} t_{m} (V_{>m})^{\tra}  &= H_{m}
  \end{align}
  has the same singular values as $t_{m}$.
\end{rem}

The following \namecref{lem:mpr-interface-inv} provides an efficient, incremental construction of matrices $U_{\le k}$ and $V_{>k}$ such that the matrix $U_{\le m} t_{m} (V_{>m})^{\tra}$ has the same rank as $t_{m}$.
More general matrices are permitted than in \cref{eq:mpr-orthocore-unfolding-svals} and the rank is preserved (\cref{eq:mpr-unfolding-rk-result}) but the singular values of $U_{\le m} t_{m} (V_{>m})^{\tra}$ can differ from those of $t_{m}$.
The proof of \cref{lem:mpr-interface-inv,lem:rec-permsub} has been sketched in \parencite{Oseledets2010a}. 
In the premise of the following \namecref{lem:mpr-interface-inv}, it is possible to choose $U_{k}$ and $V_{k}$ as submatrices of permutation matrices (the case considered in \parencite{Oseledets2010a}), but the actual proof is independent of this choice. 

\begin{lem}
  \label{lem:mpr-interface-inv}
  
  Assume a matrix product representation of $t$ as in \cref{eq:mpr-matrix-product-repr} with $H_{m} = \idm$.
  In the following, $m \in \{1 \dots n-1\}$ is fixed, $j \in \{1 \dots m\}$ and $k \in \{m \dots n-1\}$.
  Choose matrices $U_{j} \in \C^{D_{j} \times D_{j-1}\ldim_{j}}$ and $V_{k+1} \in \C^{D_{k} \times \ldim_{k+1}D_{k+1}}$.
  Set $U_{\le 0} = 1$, $V_{> n} = 1$ and
  \begin{subequations}
  \begin{alignat}{3}
    \label{eq:mpr-left-int-inv}
    U_{\le j}
    \;=\;& U_{j}(U_{\le j-1} \otimes \idm_{\ldim_{j}}) &\;\in\;& \C^{D_{j} \times \ldim_{1}\dots\ldim_{j}}
    , \\
    \label{eq:mpr-right-int-inv}
    V_{>k} \;=\;& V_{k+1}(\idm_{\ldim_{k+1}} \otimes V_{>k+1}) &\;\in\;& \C^{D_{k} \times \ldim_{k+1}\dots\ldim_{\ns}}
    .
    \end{alignat}
  \end{subequations}
  In addition, set
  \begin{subequations}
  \begin{alignat}{3}
    \label{eq:mpr-left-int-inv-parts}
    \tilde U_{j}
    \;=\;& [(U_{\le j-1} G_{\le j-1}) \otimes \idm_{\ldim_{j}}] G_{k}^{\lunf} &\;\in\;& \C^{D_{j-1}d_{j} \times D_{j}}
    , \\
    \label{eq:mpr-right-int-inv-parts}
    \tilde V_{k} \;=\;& [\idm_{\ldim_{k+1}} \otimes (V_{>k+1} G_{>k+1})] (G_{k+1}^{\runf})^{\tra} &\;\in\;& \C^{d_{k+1}D_{k+1} \times D_{k}}
    .
    \end{alignat}%
    \label{eq:mpr-int-inv-parts}%
  \end{subequations}
  If the rank equalities
  \begin{subequations}
  \begin{align}
    \label{eq:mpr-left-int-inv-rk-cond} 
    \rk(U_{j}\tilde U_{j}) &= \rk(\tilde U_{j})
    \\
    \label{eq:mpr-right-int-inv-rk-cond} 
    \rk(V_{k} \tilde V_{k}) &= \rk(\tilde V_{k})
    \end{align}
  \end{subequations}
  hold, then the following rank equalities hold as well:
  \begin{subequations}
  \begin{align}
    \rk(U_{\le j} G_{\le j}) &= \rk(G_{\le j})
    \label{eq:mpr-left-int-inv-rk-result}
    \\
    \rk(V_{>k} G_{>k}) &= \rk(G_{>k})
    \label{eq:mpr-right-int-inv-rk-result}
    \\ \rk(U_{\le m} t_{m} (V_{>m})^{\tra}) &= \rk(t_{m})
    \label{eq:mpr-unfolding-rk-result}
    \end{align}
  \end{subequations}
\end{lem}

\begin{proof}
  Note that
  \begin{align}
    \label{eq:mpr-int-inv-part-composition}
    U_{\le j} G_{\le j} &= U_{j} \tilde U_{j},
    & V_{>k} G_{>k} &= V_{k} \tilde V_{k}.
  \end{align}
  The matrices $\tilde U_{k}$ and $\tilde V_{k}$ can be computed efficiently by using \cref{eq:mpr-int-inv-part-composition} to compute $U_{\le j} G_{\le j}$ and $V_{>k} G_{>k}$.  
  If we refer to \cref{prop:rec-matrixdec} in the remainder of the proof, we use the fact that $\rk(LM) = \rk(M)$ implies $\rk(LMR) = \rk(MR)$ for three matrices $L$, $M$ and $R$.
  
  For $j = 1$, we have $U_{\le j} = U_{j}$ and $\tilde U_{j} = G_{j}^{\lunf} = G_{\le j}$, i.e. \eqref{eq:mpr-left-int-inv-rk-cond} implies \eqref{eq:mpr-left-int-inv-rk-result} for $j = 1$.
  Suppose that \cref{eq:mpr-left-int-inv-rk-result} holds for some $j-1 \in \{1 \dots m-1\}$, i.e.
  \begin{align}
    \rk(U_{\le j-1} G_{\le j-1}) = \rk(G_{\le j-1})
  \end{align}
  which implies
  \begin{align}
    \rk((U_{\le j-1} G_{\le j-1}) \otimes \idm_{\ldim_{j}}) = \rk(G_{\le j-1} \otimes \idm_{\ldim_{j}}).
  \end{align}
  This in turn implies (\cref{prop:rec-matrixdec})
  \begin{align}
    \rk(\tilde U_{j}) = \rk([(U_{\le j-1} G_{\le j-1}) \otimes \idm_{\ldim_{j}}] G_{j}^{\lunf}) = \rk([G_{\le j-1} \otimes \idm_{\ldim_{j}}] G_{j}^{\lunf}) = \rk(G_{\le j}).
  \end{align}
  Then
  \begin{align}
    \rk(U_{\le j} G_{\le j}) = \rk(U_{j}\tilde U_{j}) = \rk(\tilde U_{j}) = \rk(G_{\le j})
  \end{align}
  where we used in turn \cref{eq:mpr-int-inv-part-composition}, \cref{eq:mpr-left-int-inv-rk-cond} and the last but one equation. This shows that \cref{eq:mpr-left-int-inv-rk-result} holds for $j \in \{1 \dots m\}$.
  
  The proof of \cref{eq:mpr-right-int-inv-rk-result} proceeds in the same way:
  For $k = n-1$, we have $V_{>k} = V_{k+1}$ and $\tilde V_{k} = (G_{k+1}^{\runf})^{\tra} = G_{>k}$, i.e. \eqref{eq:mpr-right-int-inv-rk-cond} implies \eqref{eq:mpr-right-int-inv-rk-result} for $k = n-1$.
  Suppose that \eqref{eq:mpr-right-int-inv-rk-result} holds for some $k+1 \in \{m+1\dots n\}$, i.e.
  \begin{align*}
    \rk(\tilde V_{k})
    &=
    \rk([\idm_{\ldim_{k+1}} \otimes (V_{>k+1} G_{>k+1})] (G_{k+1}^{\runf})^{\tra})
    =
    \rk([\idm_{\ldim_{k+1}} \otimes (G_{>k+1})] (G_{k+1}^{\runf})^{\tra})
    \\
    &=
    \rk(G_{>k}).
  \end{align*}
  This implies
  \begin{align}
    \rk(V_{>k}G_{>k}) 
    &=
      \rk(V_{k} \tilde V_{k})
      =
      \rk(\tilde V_{k})
      =
      \rk(G_{>k}).
  \end{align}
  We have used, in turn, \eqref{eq:mpr-int-inv-part-composition}, \eqref{eq:mpr-right-int-inv-rk-cond} and the last but one equation.
  This implies that \eqref{eq:mpr-right-int-inv-rk-cond} holds for $k \in \{m \dots n-1\}$. 

  The unfolding can be written as $t_{m} = G_{\le m} (G_{>m})^{\tra}$ (\cref{eq:mpr-unfolding-with-interface}).
  Applying \cref{eq:mpr-left-int-inv-rk-result,eq:mpr-right-int-inv-rk-result} and \cref{prop:rec-matrixdec} provides
  \begin{align}
    \rk(t_{m}) &= \rk(G_{\le m} (G_{>m})^{\tra}) = \rk(U_{\le m}G_{\le m} (G_{>m})^{\tra}) = \rk(U_{\le m}G_{\le m} (G_{>m})^{\tra} (V_{>m})^{\tra}) \nonumber \\
    &= \rk(U_{\le m} t_{m} (V_{>m})^{\tra}). 
  \end{align}
  This shows that \cref{eq:mpr-unfolding-rk-result} holds and finishes the proof. 
\end{proof}

In the last Lemma, it was possible to choose $U_{j}$ as submatrices of permutation matrices.
The following Lemma shows that this implies that the $U_{\le j}$ are submatrices of permutation matrices as well and that the position of the non-zero entries of $U_{\le j}$ can be computed efficiently. 

\begin{lem}
  \label{lem:rec-permsub}
  
  In the following, let $j \in \{1 \dots m\}$.
  Choose matrices $U_{j} \in \C^{D_{j} \times D_{j-1} d_{j}}$ which are submatrices of permutation matrices. 
  Set $U_{\le 0} = 1$ and set
  \begin{align}
    U_{\le j} = U_{j} (U_{\le j-1} \otimes \idm_{d_{j}}) \quad \in \C^{D_{j} \times d_{1} \dots d_{j}}.
    \label{eq:rec-permsub-def}
  \end{align}
  In the following, we also assume $i_{j} \in \{1 \dots d_{j} \}$ and $b_{j} \in \{1 \dots D_{j}\}$. 
  Denote the set of unit entries in a given row of $U_{j}$ and $U_{\le j}$ by
  \begin{subequations}
    \begin{alignat}{3}
      \label{eq:rec-permsub-rowset-j}
      f_{j}(b_{j})
      &= \bigl\{ (b_{j-1},i_{j}) \colon&& [U_{j}]_{b_{j},(b_{j-1},i_{j})} &&= 1 \bigr\},
      \\
      \label{eq:rec-permsub-rowset-le-j}
      f_{\le j}(b_{j})
      &= \bigl\{ (i_{1},\dots, i_{j}) \colon&& [U_{\le j}]_{b_{j},(i_{1} \dots i_{j})} &&= 1 \bigr\}.
    \end{alignat}
  \end{subequations}
  The latter set is given by
  \begin{align}
    \label{eq:rec-permsub-rowset-result}
    f_{\le j}(b_{j})
    &=
      \cura{
      (i_{1},\dots,i_{j}) \colon
      \;
      (b_{j-1}, i_{j}) \in f_{j}(b_{j})
      \;
      \text{ and }
      \;
      (i_{1},\dots,i_{j-1}) \in f_{\le j-1}(b_{j-1})
      \;
      },
  \end{align}
  and $U_{\le j}$ is a submatrix of a permutation matrix, i.e.\ we have $\abs{f_{\le j}(b_{j})} \le 1$ and the entries of $U_{\le j}$ are given by
  \begin{align}
    \label{eq:rec-permsub-entry-result}
    [U_{\le j}]_{b_{j},(i_{1} \dots i_{j})} =
    \begin{cases}
      1, & \text{if } (i_{1}, \dots, i_{j}) \in f_{\le j}(b_{j}), \\
      0, & \text{otherwise.}
    \end{cases}
  \end{align}
\end{lem}

\begin{proof}
  As $U_{j}$ is a submatrix of a permutation matrix, we have $\abs{f_{j}(b_{j})} \le 1$ and the entries of $U_{j}$ are given by
  \begin{align}
    \label{eq:rec-permsub-j-entries}
    [U_{j}]_{b_{j},(b_{j-1},i_{j})} =
    \begin{cases}
      1, & \text{if } (b_{j-1},i_{j}) \in f_{j}(b_{j}), \\
      0, & \text{otherwise.}
    \end{cases}
  \end{align}
  The entries of $U_{\le j}$ are given by (\cref{eq:rec-permsub-def,eq:rec-permsub-j-entries})
  \begin{subequations}
    \begin{align*}
      [U_{\le j}]_{b_{j},(i_{1}\dots i_{j})}
      &=
        \sum_{b_{j-1}} [U_{j}]_{b_{j},(b_{j-1},i_{j})} [U_{\le j-1}]_{b_{j-1},(i_{1} \dots i_{j-1})}
      =
        \sum_{\substack{b_{j-1}\colon\\(b_{j-1}, i_{j}) \in f_{j}(b_{j})}} [U_{\le j-1}]_{b_{j-1},(i_{1} \dots i_{j-1})}
        \nonumber
      \\
      &=
        \begin{cases}
          [U_{\le j-1}]_{b_{j-1},(i_{1} \dots i_{j-1})} \text{ where } (b_{j-1}, i_{j}) \in f_{j}(b_{j}) &
          \text{if } \abs{f_{j}(b_{j})} = 1, \\
          0 & \text{otherwise.}
        \end{cases}
    \end{align*}
  \end{subequations}
  Here, we used that there is at most one element $(b_{j-1}, i_{j}) \in f_{j}(b_{j})$.
  Note that $\abs{f_{\le j}(b_{j})} \le 1$ and \cref{eq:rec-permsub-entry-result} hold for $j = 0$ (as $f_{\le 0}(1) = \{ 1 \}$); assume that the two conditions hold for some $j-1 \in \{0 \dots m-1\}$.
  In the following, we show that they also hold for $j$.
  If we apply the assumption to the last equation, we obtain
  \begin{align}
    [U_{\le j}]_{b_{j},(i_{1}\dots i_{j})}
    &=
      \begin{cases}
        1, & \text{if } (b_{j-1}, i_{j}) \in f_{j}(b_{j})
        \;\text{ and }\;
        (i_{1}, \dots, i_{j-1}) \in f_{\le j-1}(b_{j-1}), \\
        0, & \text{otherwise.}
      \end{cases}
  \end{align}
  The set $f_{\le j}(b_{j})$ has at most one element because $f_{j}(b_{j})$ has at most one element and because we assumed that $f_{\le j - 1}(b_{j-1})$ (for the single possible value of $b_{j-1}$) has at most one element as well. 
  This finishes the proof. 
\end{proof}

\subsection{Sequence of local quantum operations as PMPS representation
  \label{sec:rec-appendix-pmps-and-seq-prep}}

This section introduces the locally purified matrix product state (PMPS) representation and discusses the known fact that a sequentially prepared mixed quantum state can be represented as a PMPS or as an MPO (\cref{lem:rec-spinchain-seq-prep}). 
The PMPS representation \parencite{Verstraete2004,DelasCuevas2013} provides an alternative to the MPO representation for positive semidefinite operators such as mixed quantum states. 
The purification is given in terms of $\ns$ ancilla systems of dimensions $\ldim'_{k}$ with bases $\{\ket{\varphi^{(k)}_{i'_{k}}}\}_{k=1}^{\ldim'_{k}}$. 
A PMPS representation of $\rho$ is given by
\begin{align}
  \nonumber
  \rho
  &= \tr_{1'\dots\ns'}(\ketbra\Psi\Psi)
  &
    D_{0}
  &= D_{\ns} = 1
  \\
  \label{eq:rec-intro-pmps-repr}
  \braketc{
  \phi^{(1)}_{i_{1}}\varphi^{(1)}_{i'_{1}} \dots \phi^{(\ns)}_{i_{\ns}}\varphi^{(\ns)}_{i'_{\ns}}
  }{\Psi}
  &=
    G_{1}(i_{1},i'_{1}) G_{2}(i_{2},i'_{2}) \dots G_{\ns}(i_{\ns},i'_{\ns})
  &
    i_{k}
  &\in \{1, \dots, \ldim_{k}\}
  \\ \nonumber
  G_{k}(i_{k},i'_{k})
  &\in \C^{D_{k-1} \times D_{k}}
  &
    i'_{k}
  &\in \{1, \dots, \ldim'_{k}\}
\end{align}
Given the tensors $G_{k}$ of a PMPS representation, the tensors $\tilde G_{k}$ of an MPO representation are given by
\begin{align}
  \label{eq:rec-intro-pmps-as-mpo}
  \tilde G_{k}(i_{k}, j_{k})
  &=
    \sum_{i'_{k}=1}^{\ldim'_{k}} G_{k}(i_{k}, i'_{k}) \otimes \overline{G_{k}(i'_{k}, j_{k})}
\end{align}
where the overline denotes the complex conjugate.
\Cref{eq:rec-intro-pmps-as-mpo} shows that given a PMPS representation of bond dimension $D$, we can directly construct an MPO representation with bond dimension $D^{2}$.
However, an MPO representation with bond dimensions smaller than $D^{2}$ can exist.
It has been shown that there is a family of quantum states on $\ns$ systems which can be represented as an MPO with bond dimension independent of $\ns$ but the bond dimension of any PMPS representation of those states increases with $\ns$ \parencite{DelasCuevas2013}.
This is an advantage of the MPO representation, but on the other hand, deciding whether a given MPO representation represents a positive semidefinite operator is an NP-hard problem, i.e. a solution in polynomial (in $n$) time is unlikely \parencite{Kliesch2014a}.
The PMPS representation has the advantage that it always represents a positive semidefinite operator by definition. 
The relative merits of the MPO and PMPS representations of a mixed quantum state depend on the application.

Suppose that a quantum state $\rho \in \qs(\hilb_{1\dots \ns})$ was prepared via quantum operations
\\
$\issop{\preparemap_{k}}{\hilb_{k-1}}{\hilb_{k-1,k}}$,
i.e.
\begin{align}
  \rho = \preparemap_{\ns} \preparemap_{\ns-1} \dots \preparemap_{3} \preparemap_{2}(\sigma),
  \quad\quad
  \sigma \in \qs(\hilb_{1}).
\end{align}
Clearly, this is an efficient representation of the quantum state $\rho$ as it is described by at most $\ns \ldim^{6}$ parameters.
It is known that such a representation can be efficiently -- i.e.~with at most $\poly(\ns)$ computational time -- converted into an MPO representation or a PMPS representation \parencite{Poulin2011,Kliesch2014a}.
The following \namecref{lem:rec-spinchain-seq-prep} provides the technical details of the conversion:
\begin{lem}
  \label{lem:rec-spinchain-seq-prep}
  
  Let $\rightsetp{k}$ ($k \in \{1, \dots, \ns\}$) be systems with $\dim(\rightsetp{\ns}) = 1$.
  Let $\rho_{1} \in \linop(\hilbsysand{1}\rightsetp{1})$ and $\preparemap_{k} \in \sop{\rightsetp{k-1}}{\hilbsysand{k} \rightsetp{k}}$.
  A linear operator $\rho \in \linop(\hilb_{1\dots \ns})$ is given by
  \begin{align}
    \label{eq:rec-spinchain-seq-prep}
    \rho = \rho_{\ns}, \quad\quad\quad\quad
    \rho_{k} &= (\id \otimes \preparemap_{k})(\rho_{k-1}) \quad\quad (k \in \{2, \dots, \ns\}).
  \end{align}
  Let $F^{(k)}_{i_{k}}$ and $F^{(\rightsetp{k})}_{b_{k}}$ be orthonormal operator bases of systems $k$ and $\rightsetp{k}$.
  An MPO representation of $\rho$ is given by
  \begin{subequations}
    \label{eq:rec-spinchain-seq-prep-mpo}
    \begin{align}
      [G_{1}(i_{1})]_{b_{0},b_{1}}
      &= \skp{F^{(1)}_{i_{1}} \otimes F^{(\rightsetp{1})}_{b_{1}}}{\rho_{1}},
      \\
      [G_{k}(i_{k})]_{b_{k-1},k_{k}}
      &= \skp{F^{(k)}_{i_{k}} \otimes F^{(\rightsetp{k})}_{b_{k}}}{\preparemap_{k}(F^{(\rightsetp{k-1})}_{b_{k-1}})},
      \quad\quad (k \in \{2, \dots, \ns\}).
    \end{align}
  \end{subequations}
  The bond dimensions of the representation are given by $D_{k} = (\ldim_{\rightsetp{k}})^{2}$ ($k \in \{1, \dots, \ns - 1\}$, $D_{0} = D_{\ns} = 1$).
  Now, let $\rho_{1}$ be a quantum state and $\preparemap_{k}$ be CPTP linear maps. In addition, let
  \begin{align}
    \label{eq:rec-spinchain-seq-prep-decomp-for-pmps}
    \rho_{1}
    &= \sum_{i=1}^{r_{1}} \ketbra{\psi_{i}}{\psi_{i}},
    &
      \preparemap_{k}(\cdot)
    &=
      \sum_{i=1}^{r_{k}}
      E^{(k)}_{i}
      \;\cdot\;
      (E^{(k)}_{i})^{\adjm}
      \quad\quad (k \in \{2, \dots, \ns\})
  \end{align}
  be a decomposition of $\rho_{1}$ into $r_{1} = \rk(\rho_{1})$ orthogonal vectors and a Kraus decomposition of $\preparemap_{k}$ where $r_{k}$ equals the Kraus rank of $\preparemap_{k}$ (i.e.\ $r_{1} \le \ldim_{1} \ldim_{\rightsetp{1}}$ and $r_{k} \le \ldim_{k} \ldim_{\rightsetp{k}} \ldim_{\rightsetp{k-1}}$).
  Let $\{\ket{i_{k}}_{k}\}_{i_{k}}$ and $\{\ket{b_{k}}_{\rightsetp{k}}\}_{b_{k}}$ be orthonormal bases of $\hilb_{k}$ and $\rightsetp{k}$, respectively. 
  A PMPS representation of $\rho$ is given by
  \begin{subequations}
    \label{eq:rec-spinchain-seq-prep-pmps}
    \begin{align}
      [G_{1}(i_{1},i'_{1})]_{b_{0},b_{1}}
      &\,=\; {}_{1}\bra{i_{1}} \, {}_{\rightsetp{1}}\!\braket{b_{1}}{\psi_{i'_{1}}}, \\
      [G_{k}(i_{k},i'_{k})]_{b_{k-1},b_{k}}
      &\,=\; {}_{k}\bra{i_{k}} \, {}_{\rightsetp{k}}\!\bra{b_{k}} E^{(k)}_{i'_{k}} \ket{b_{k-1}}_{\rightsetp{k-1}}
      \quad\quad (k \in \{2, \dots, \ns\}).
    \end{align}
    The bond dimensions of the representation are given by $D_{k} = \ldim_{\rightsetp{k}}$ ($k \in \{1, \dots, \ns - 1\}$, $D_{0} = D_{\ns} = 1$) and the ancilla dimensions are given by $r_{k}$. 
  \end{subequations}  
\end{lem}

\begin{proof}
  For the MPO representation, evaluate the operator basis elements of $\rho$ from \cref{eq:rec-spinchain-seq-prep} and compare with the operator basis elements of the representation (\cref{eq:rec-intro-mpo-repr-opbasis}). For the PMPS representation, evaluate the matrix entries of $\rho$ from \cref{eq:rec-spinchain-seq-prep} inserting \cref{eq:rec-spinchain-seq-prep-decomp-for-pmps} and compare with the matrix entries of the representation (\cref{eq:rec-intro-pmps-repr}).
\end{proof}

\subsection{General linear maps as measurements
  \label{sec:rec-appendix-intro-linmap-meas}}

Consider a quantum state $\rho \in \qs(Y)$ and an arbitrary linear map $\mcn \in \sop{Y}{X}$.
Below, we work with linear maps $\mcn$ which are not necessarily CPTP and therefore do not represent a physical operation on the quantum state.
Such a map $\mcn$ is of relevance only if $\mcn(\rho)$ can be obtained from the outcomes of physical measurements on $\rho$.
In this section, we show how this can be achieved, allowing the reconstruction scheme from \cref{sec:rec-reconstr-bipartite-states-subsec} to be used for quantum state tomography. 

Firstly, we construct a set of observables $G_{i} \in \linop(Y)$ ($i \in \{1, \dots, 2d_{X}^{2}\}$) whose expectation values $\tr(G_{i} \rho)$ in the state $\rho$ can be used to compute $\mcn(\rho)$.
Secondly, we construct a POVM with $2d_{x}^{2} + 1$ elements $E_{i} \in \linop(Y)$ such that the outcome probabilities $\tr(E_{i} \rho)$ in the state $\rho$ can also be used to compute $\mcn(\rho)$.

We denote the components of $\mcn(\rho)$ in an operator basis $F^{(X)}_{i}$ of $X$ by $s_{i}$:
\begin{align}
  s_{i} &= \skp{F^{(X)}_{i}}{\mcn(\rho)}
  & i &\in \{1, \dots, \ldim_{X}^{2}\}
  & \ldim_{X} &= \dim(X)
\end{align}
The key tool is the following property of the map $\mcn^{\adjs}$:
\begin{align}
  s_{i} = \skp{\mcn^{\adjs}(F^{(X)}_{i})}{\rho} = \tr(H_{i}\rho), \quad H_{i} = \parena{\mcn^{\adjs}(F^{(X)}_{i})}^{\adjm}
\end{align}
Since $H_{i}$ may not be Hermitian, we use its Hermitian and skew-Hermitian components:
\begin{align}
  G_{i} &= \frac12 \parena{ H_{i} + H_{i}^{\adjm} }
  &
    G_{i + \ldim_{X}^{2}} &= \frac1{2\ii} \parena{ H_{i} - H_{i}^{\adjm} }
  &
    H_{i} &= G_{i} + \ii G_{i + \ldim_{X}^{2}}
\end{align}
Using the observables $G_{i}$, the components $s_{i}$ can be expressed as follows:
\begin{align}
  s_{i} &= \skp{\mcn^{\adjs}(F^{(X)}_{i})}{\rho} = \tr(G_{i}\rho) + \ii \tr(G_{i+\ldim_{X}^{2}} \rho)
\end{align}
In other words, the expectation values of the $G_{i}$ provide the real and imaginary parts of $s_{i}$:
\begin{align}
  \label{eq:rec-intro-mcn-rho-coeff-imag-real}
  \Re(s_{i}) &= \tr(G_{i} \rho) & \Im(s_{i}) &= \tr(G_{i+\ldim_{X}^{2}} \rho)
\end{align}
If these expectation values can be measured, we already obtain a way to obtain $\mcn(\rho)$ from physical measurements on $\rho$ even if $\mcn$ is not CPTP\@.
Furthermore, we construct a POVM whose measurement on $\rho$ also allows to determine $\mcn(\rho)$.
We choose coefficients $c_{i} \in \R$ and $c > 0$ such that the following operators become positive semidefinite:
\begin{align}
  G_{i} + c_{i} \idm &\ge 0
  &
    \idm - c \sum_{i=1}^{2d_{X}^{2}} (G_{i} + c_{i} \idm) &\ge 0
\end{align}
We define $2d_{X}^{2} + 1$ POVM elements by
\begin{align}
  E_{0} &= \idm - \sum_{i=1}^{2d_{X}^{2}} E_{i}
  &
    E_{i} &= c (G_{i} + c_{i} \idm)
\end{align}
Clearly, the expectation values of the $G_{i}$ are related to the POVM probabilities by
\begin{align}
  \tr(G_{i} \rho) &= \frac1c \parena{ \tr(E_{i} \rho) - c_{i} }.
\end{align}
The coefficients $s_{i}$ of $\mcn(\rho)$ can be obtained from these expectation values using \cref{eq:rec-intro-mcn-rho-coeff-imag-real}.
As a consequence, the POVM probabilities of the given POVM allow us to determine $\mcn(\rho)$ even if $\mcn$ is not CPTP\@.

\printbibliography

\end{document}
